\documentclass[onecolumn,10pt]{IEEEtran}
\usepackage[mathscr]{eucal}
\usepackage[cmex10]{amsmath}
\usepackage{epsfig,epsf,psfrag}
\usepackage{amssymb,amsmath,amsthm,amsfonts,latexsym}
\usepackage{amsmath,graphicx,bm,xcolor,url,overpic}
\usepackage{array}
\usepackage{verbatim}
\usepackage{bm}
\usepackage{algorithmic}
\usepackage{algorithm}
\usepackage{verbatim}
\usepackage{textcomp}
\usepackage{mathrsfs}
\usepackage{epstopdf}

\newcommand{\openone}{\leavevmode\hbox{\small1\normalsize\kern-.33em1}}

\catcode`~=11 \def\UrlSpecials{\do\~{\kern -.15em\lower .7ex\hbox{~}\kern .04em}} \catcode`~=13 

\allowdisplaybreaks[4]

\newcommand{\nn}{\nonumber}

\newcommand{\calA}{\mathcal{A}}
\newcommand{\calB}{\mathcal{B}}
\newcommand{\calC}{\mathcal{C}}

\newcommand{\calF}{\mathcal{F}}

\newcommand{\calH}{\mathcal{H}}
\newcommand{\calI}{\mathcal{I}}

\newcommand{\calM}{\mathcal{M}}
\newcommand{\calN}{\mathcal{N}}

\newcommand{\calP}{\mathcal{P}}
\newcommand{\calQ}{\mathcal{Q}}
\newcommand{\calR}{\mathcal{R}}

\newcommand{\calT}{\mathcal{T}}

\newcommand{\calX}{\mathcal{X}}
\newcommand{\calY}{\mathcal{Y}}

\newcommand{\bT}{\mathbf{T}}

\newcommand{\bV}{\mathbf{V}}

\newcommand{\bx}{\mathbf{x}}
\newcommand{\bX}{\mathbf{X}}
\newcommand{\by}{\mathbf{y}}
\newcommand{\bY}{\mathbf{Y}}

\newcommand{\rmc}{\mathrm{c}}

\newcommand{\rmd}{\mathrm{d}}

\newcommand{\rmG}{\mathrm{G}}

\newcommand{\rmH}{\mathrm{H}}

\newcommand{\rmM}{\mathrm{M}}

\newcommand{\rmo}{\mathrm{o}}

\newcommand{\rmr}{\mathrm{r}}

\newcommand{\rmS}{\mathrm{S}}
\newcommand{\rmt}{\mathrm{t}}
\newcommand{\rmT}{\mathrm{T}}

\newcommand{\rmU}{\mathrm{U}}

\newcommand{\rmV}{\mathrm{V}}

\newcommand{\bbE}{\mathsf{E}}

\newcommand{\bbN}{\mathbb{N}}

\newcommand{\bbP}{\mathbb{P}}

\newcommand{\bbR}{\mathbb{R}}

\DeclareMathAlphabet{\mathbsf}{OT1}{cmss}{bx}{n}
\DeclareMathAlphabet{\mathssf}{OT1}{cmss}{m}{sl}

\newcommand{\rvbPh}{\bsfPhi}

\DeclareSymbolFont{bsfletters}{OT1}{cmss}{bx}{n}  
\DeclareSymbolFont{ssfletters}{OT1}{cmss}{m}{n}
\DeclareMathSymbol{\bsfGamma}{0}{bsfletters}{'000}
\DeclareMathSymbol{\ssfGamma}{0}{ssfletters}{'000}
\DeclareMathSymbol{\bsfDelta}{0}{bsfletters}{'001}
\DeclareMathSymbol{\ssfDelta}{0}{ssfletters}{'001}
\DeclareMathSymbol{\bsfTheta}{0}{bsfletters}{'002}
\DeclareMathSymbol{\ssfTheta}{0}{ssfletters}{'002}
\DeclareMathSymbol{\bsfLambda}{0}{bsfletters}{'003}
\DeclareMathSymbol{\ssfLambda}{0}{ssfletters}{'003}
\DeclareMathSymbol{\bsfXi}{0}{bsfletters}{'004}
\DeclareMathSymbol{\ssfXi}{0}{ssfletters}{'004}
\DeclareMathSymbol{\bsfPi}{0}{bsfletters}{'005}
\DeclareMathSymbol{\ssfPi}{0}{ssfletters}{'005}
\DeclareMathSymbol{\bsfSigma}{0}{bsfletters}{'006}
\DeclareMathSymbol{\ssfSigma}{0}{ssfletters}{'006}
\DeclareMathSymbol{\bsfUpsilon}{0}{bsfletters}{'007}
\DeclareMathSymbol{\ssfUpsilon}{0}{ssfletters}{'007}
\DeclareMathSymbol{\bsfPhi}{0}{bsfletters}{'010}
\DeclareMathSymbol{\ssfPhi}{0}{ssfletters}{'010}
\DeclareMathSymbol{\bsfPsi}{0}{bsfletters}{'011}
\DeclareMathSymbol{\ssfPsi}{0}{ssfletters}{'011}
\DeclareMathSymbol{\bsfOmega}{0}{bsfletters}{'012}
\DeclareMathSymbol{\ssfOmega}{0}{ssfletters}{'012}

\newcommand{\hatH}{\hat{H}}
\newcommand{\tilh}{\tilde{h}}

\newcommand{\hatK}{\hat{K}}

\newcommand{\tilP}{\tilde{P}}

\newcommand{\tilQ}{\tilde{Q}}

\newcommand{\hatT}{\hat{T}}

\newcommand{\bari}{\bar{i}}
\newcommand{\barj}{\bar{j}}

\newcommand{\bkappa}{\bm{\kappa}}

\newcommand{\bmu}{\bm{\mu}}

\newcommand{\bSigma	}{\bm{\Sigma}}

\newcommand{\tlambda}{\tilde{\lambda}}

\newcommand{\tOmega}{\tilde{\Omega}}
\newcommand{\tPsi}{\tilde{\Psi}}

\DeclareMathOperator*{\argmin}{arg\,min}

\newcommand{\Var}{\mathrm{Var}}

\newtheorem{theorem}{Theorem} 
\newtheorem{lemma}{Lemma}

\usepackage{subfigure,epstopdf,bbm} 
\usepackage[utf8]{inputenc} 
\usepackage[T1]{fontenc}    
\usepackage{url,bbm}            
\usepackage{booktabs}       
\usepackage{amsfonts}       
\usepackage{nicefrac}       
\usepackage{microtype}      
\usepackage{cite}
\usepackage{subfigure,transparent,color,graphicx} 

\newcommand{\bbo}{\mathbbm{1}}
\usepackage[ colorlinks = true,
         linkcolor = blue,
         urlcolor  = blue,
         citecolor = red,
         anchorcolor = green]{hyperref}
\allowdisplaybreaks[2]
\begin{document}
\title{Large and Small Deviations for \\Statistical Sequence Matching}

\author{Lin Zhou, Qianyun Wang, Jingjing Wang, Lin Bai and Alfred O. Hero

\thanks{L. Zhou, Q. Wang, J. Wang and L. Bai are with the School of Cyber Science and Technology, Beihang University, Beijing, China, 100083 (Emails: \{lzhou, wangqianyun, drwangjj, l.bai\}@buaa.edu.cn).}
\thanks{A. O. Hero is with the department of Electrical Engineering and Computer Science, University of Michigan, Ann Arbor (Email: hero@eecs.umich.edu).}

}

\maketitle

\begin{abstract}
We revisit the problem of statistical sequence matching between two databases of sequences initiated by Unnikrishnan (TIT 2015) and derive theoretical performance guarantees for the generalized likelihood ratio test (GLRT). We first consider the case where the number of matched pairs of sequences between the databases is known. In this case, the task is to accurately find the matched pairs of sequences among all possible matches between the sequences in the two databases. We analyze the performance of the GLRT by Unnikrishnan and explicitly characterize the tradeoff between the mismatch and false reject probabilities under each hypothesis in both large and small deviations regimes. Furthermore, we demonstrate the optimality of Unnikrishnan's GLRT test under the generalized Neyman-Person criterion for both regimes and illustrate our theoretical results via numerical examples. Subsequently, we generalize our achievability analyses to the case where the number of matched pairs is unknown, and an additional error probability needs to be considered. When one of the two databases contains a single sequence, the problem of statistical sequence matching specializes to the problem of multiple classification introduced by Gutman (TIT 1989). For this special case, our result for the small deviations regime strengthens previous result of Zhou, Tan and Motani (Information and Inference 2020) by removing unnecessary conditions on the generating distributions.
\end{abstract}


\section{Introduction}

Hypothesis testing lies in the intersection of information theory, signal processing and statistics~\cite{blahut1974hypothesis,lehmann2006testing,batu2013testing}. In the simplest model of binary hypothesis testing, one is given a sequence of a certain length and two known distributions $(P,Q)$ that could have generated the sequence. It is assumed that the observed sequence is generated i.i.d. from one of the two distributions. There are two hypotheses, each of which specifies a possible generating distribution of the test sequence. The task is to design a test to correctly identify the true hypothesis. However, in practical applications such as image classification and junk mail identification, the generating distribution under each hypothesis is usually unavailable. Thus, the hypothesis testing framework fails to apply directly.

To resolve the above problem, Gutman~\cite{gutman1989asymptotically} proposed the framework of statistical classification, where under each hypothesis, a training sequence generated i.i.d. from the unknown generating distribution is available. For the binary case, the task is thus to design a test based on the training data without knowledge of the generating distributions. Naturally, there are two performance criteria: the type-I and type-II error probabilities, each of which specifies the probability of an error under the respective null and alternative hypotheses. Gutman proposed a threshold based generalized likelihood ratio test (GLRT), analyzed its asymptotically achievable error exponent rate for the type-I error probability and proved the asymptotic optimality of the test under the generalized Neyman-Pearson criterion. Specifically, the generalized Neyman-Pearson criterion requires one to consider tests that ensure exponential decay of the error probability under each hypothesis with a certain exponent rate for all possible generating distributions. The optimality under the generalized Neyman-Pearson criterion implies that a test has smallest false reject probability under each hypothesis for any generating distributions under the above condition. Note that Gutman's GLRT for binary classification generalizes the Hoeffding's test~\cite{hoeffding1965} for the binary hypothesis testing problem when the generating distribution $P$ under hypothesis $\rmH_1$ is known while the generating distribution $Q$ under hypothesis $\rmH_2$ is unknown. Recently, Zhou, Tan and Motani~\cite{zhou2018binary} refined Gutman's results by explicitly deriving the tradeoff between the type-I and type-II error probabilities in the small deviations regime that provides approximations to the performance of an optimal test in the non-asymptotic setting, using finite-length testing and training sequences. The authors of \cite{gutman1989asymptotically,zhou2018binary} also considered the case of multiple hypotheses with a reject option, where the reject option claims that the testing sequence is not matched to any generating distributions of training sequences. As~\cite{gutman1989asymptotically,zhou2018binary} assumed that at least one of the multiple hypotheses is valid, the reject option indicates that further investigation is required to make a reliable decision. In~\cite{gutman1989asymptotically,zhou2018binary}, the reject option is critical to establishing the optimality of Gutman's test. 

Motivated by studies of privacy of anonymized databases and applications in accurate user targeting for advertisement recommender systems, Unnikrishnan~\cite[Section IV]{unnikrishnan2015asymptotically} generalized Gutman's framework to statistical sequence matching in a pair of databases of sequences. This problem strictly generalizes the statistical classification problem from testing a single  sequence for match to a distribution to finding matching pairs of distributions based on realizations of multiple sample sequences in a pair of databases. Specifically, in statistical sequence matching, there are two databases of sequences, where each sequence of each database is generated i.i.d. from an unknown distribution. If two sequences, each from one of the two databases, are generated from the same distribution, the two sequences are called a matched pair; otherwise, the two sequences are called an unmatched pair. The task is to design a test, without knowledge of the generating distributions of any of the sequences, to correctly identify all matched pairs of sequences between the two databases or to claim a reject option, which implies that no matched pair of sequences between the two databases is found. 

Unnikrishnan assumed that the number of matched pairs (matches) is positive and known. In this case, under each hypothesis, there are two performance criteria: the mismatch probability and the false reject probability. The mismatch probability quantifies the probability that an incorrect hypothesis other than the reject option is claimed while the false reject probability quantities the probability the test mistakenly declares that there are no matching sequences. In the formulation of \cite{unnikrishnan2015asymptotically}, any decision of reject is a false reject since it is assumed that the number of matched pairs is non-zero. Unnikrishnan proposed a threshold-based GLRT, analyzed its achievable mismatch probability and proved the optimality of the test when the length of each sequence tends to infinity, analogous to Gutman's results for statistical classification. However, the false reject probability was not explicitly bounded and the tradeoff between the two types of error probabilities was not studied.

In this paper, we refine Unnikrishnan's result ~\cite[Section IV]{unnikrishnan2015asymptotically} in two ways. Firstly, we explicitly derive the tradeoff between the mismatch and false reject probabilities in both large and small deviations regimes. Secondly, we generalize our achievability results to the case where the true number of matches, which could be zero, is unknown and derive performance tradeoffs among the probabilities of mismatch, the false reject and false alarm  for a variant of Unnikrishnan's test. Our main contributions are summarized as follows.

\subsection{Main Contributions}
When the number of matched pairs is known, we characterize the optimal tradeoff between the mismatch and false reject probabilities in both large deviations~\cite{dembo2009large} and small deviations regimes~\cite{polyanskiy2010finite,hayashi2009information}. Specifically, in both regimes, the mismatch probability decays exponentially fast while the false reject probability behaves differently over the two regimes. In the large deviations regime, the false reject probability also decays exponentially and thus the exponential tradeoff of mismatch and false reject probabilities are characterized. In contrast, in the small deviations regime, the false reject probability is upper bounded by a non-vanishing constant. For both regimes, we prove optimality of Unnikrishnan's GLRT test under the generalized Neyman-Pearson criterion for all possible tuples of generating distributions of the sequences. When specialized to multiple classification, our large deviations result specializes to the corresponding results of Gutman~\cite[Theorem 3]{gutman1989asymptotically} while our small deviations result refines \cite[Theorem 4.1]{zhou2018binary} by removing the unnecessary condition on the unknown tuple of generating distributions of training sequences.

We first consider the large deviations regime and establish the first order expansion (better known as the error exponent rates) when the length of each sequence tends to infinity. In fact, if one derives a second-order expansion in the large deviations regime~\cite{altug14a,altug14b,erseghe2016tit} for the present problem, the corresponding result provides a good approximation to the finite sample size setting. However, the derivation is more complicated and we leave it for future work. As shown by Unnikrishnan~\cite[Theorem 4.1]{unnikrishnan2015asymptotically}, the mismatch probability decays exponentially fast with the speed proportional to the threshold $\lambda$ of the test (cf. \eqref{test:unn}). Simultaneously, we completely characterize the achievable false reject exponent rate as a function of the threshold $\lambda$ and the evaluated tuples of unknown generating distributions. Thus, we reveal the asymptotic rate tradeoff between the mismatch and false reject probabilities. If the mismatch exponent rate $\lambda$ increases, the false reject exponent rate decreases; and if $\lambda$ is larger than a threshold value, dependent on the unknown generating distributions, the false reject exponent rate equals zero. Finally, using both mathematical analysis and numerical experiments, we compare the performance of Unnikrishnan's test and a simple test in Algorithm \ref{simpletest} that repeatedly applies Gutman's statistical classification test to find matching sequences. Our results strengthen the results of comparison in~\cite[Section IV. A]{unnikrishnan2015asymptotically} by explicitly characterizing the false reject exponent rates of both tests.

To provide more accurate insights on the achievable performance of an optimal test in the finite sample size setting, we also derive the second-order expansion in the small deviations regime. Specifically, we derive a non-asymptotic upper bound on the false reject probability and apply the multi-variate Berry-Esseen theorem to yield a bound that involves the complementary cumulative distribution function (cdf) of a multivariate Gaussian random vector. Furthermore, we show that the false reject probability is upper bounded by a constant if the mismatch exponent rate $\lambda$ is a particular function of the unknown generating distributions and the sample size $n$. In particular, as the sample size tends to infinity, the value of the particular function tends to the threshold value of $\lambda$, above which the false reject exponent is exactly zero in the large deviations regime. Therefore, the small deviations regime refines the large deviations regime for the special case of zero false reject exponent, which leads to a constant false reject probability. We illustrate our small deviations results and compare the performance of Unnikrishnan's test and the simple Gutman's test.

Finally, we generalize the achievability results to the case where the number of matched pairs is unknown. In this case, we need to consider the additional error probability that bounds the probability of the error event when the number of matched pairs is exactly zero. This is because, when the number of matched pairs is unknown, it can be either zero or strictly positive. When the number of matched pairs is zero, any decision of the reject option is an error, which we call the false alarm, where the test mistakenly claims that some pairs of matched sequences are found. In contrast, when the number of matched pairs is strictly positive, we have the same mismatch and false reject probabilities in the case of a known number of matches. For this case, we construct a test in two steps, where the first step estimates the number of matches and the second step applies Unnikrishnan's GLRT test when the estimated number of matches is positive. We analyze the achievable performance of the proposed test in both large and small deviations regimes and thus elucidate the tradeoff among the probabilities of mismatch, false reject and false alarm. When specialized to multiple classification, our results generalize the corresponding results in~\cite[Theorem 3]{gutman1989asymptotically} and \cite[Theorem 4.1]{zhou2018binary} to the more practical case where the testing sequence is allowed to be generated from a distribution different from the generating distribution of any training sequence.

\subsection{Other Related Works}
We recall other related works on statistical classification and sequence matching. Merhav and Ziv~\cite{merhav1991bayesian} studied the Bayesian setting of statistical classification and derived the achievable error exponent rate when the lengths of testing and training sequences tend to infinity. The results of \cite{merhav1991bayesian} were recently refined by Saito and Matsushima~\cite{shota2020beyasian,shota2021beyasian} who derived the corresponding result in the finite blocklength setting using Bayes codes. Unnikrishnan and Huang~\cite{unnikrishnan2016weak} proposed the weak convergence analysis and provided tight bounds for error probabilities of statistical classification. Hsu and Wang~\cite{hsu2020binary} generalized Gutman's result~\cite{gutman1989asymptotically} to the mismatched case where under the true hypothesis, the generating distributions of the training sequence and the testing sequence deviate slightly and explicitly characterized the impact of the distribution deviation on achievable error exponents. Haghifam, Tan and Khisti~\cite{mahdi2021sequential} generalized the achievability part of Gutman's result~\cite{gutman1989asymptotically} to the semi-sequential setting where the testing sequence is observed sequentially in a streaming manner. The results for the binary case of \cite{mahdi2021sequential} were recently refined by Hsu, Li and Wang~\cite{Ihwang2022sequential} who considered two fully sequential settings and derived tight results with matching achievability and converse bounds. Gutman's results have also been generalized to large alphabet~\cite{kelly2013}, distributed detection~\cite{he2019distributed}, outlier hypothesis testing~\cite{li2014,zhou2022second} and two-phase classification~\cite{zhou2022twophase}.

When one database of sequences is replaced by a database of known distributions, the problem of statistical sequence matching reduces to the problem of matching sequences to known generating distributions~\cite[Section III]{unnikrishnan2015asymptotically}. Unnikrishnan fully characterized the exponent rate tradeoff between the mismatch and false reject probabilities when the length of each sequence tends to infinity. The special case when both databases contain the same number of sequences was studied in \cite[Chapter 10]{ahlswede1987} and \cite{ahlswede2006}.

\subsection{Organization for the Rest of the Paper}

In Section \ref{sec:pf}, we set up the notation, formulate the problem of statistical sequence matching and recall Unnikrishnan's GLRT test and asymptotic results. Subsequently, in Sections \ref{sec:known} and \ref{sec:unknown}, we present and discuss our results for the case of known and unknown number of matches, respectively. The proofs of our results are presented in Sections \ref{sec:proofs:known} and \ref{sec:proofs:unknown}. Finally, in Section \ref{sec:conc}, we conclude the paper and discuss future directions. For smooth presentation of main results, the proofs of supporting lemmas are deferred to appendices.

\section{Problem Formulation and Existing Results}
\label{sec:pf}

\subsection*{Notation}
Random variables and their realizations are in upper case (e.g.,  $X$) and lower case (e.g.,  $x$), respectively. All sets are denoted in calligraphic font (e.g.,  $\mathcal{X}$). We use $\bbR$, $\bbR_+$, and $\bbN$ to denote the set of real numbers, non-negative real numbers, and  natural numbers respectively. Given any number $a\in\bbN$, we use $[a]$ to denote the collection of natural numbers between $1$ and $a$. We use superscripts to denote the length of vectors, e.g., $X^n:=(X_1,\ldots,X_n)$. All logarithms are base $e$. The set of all probability distributions on a finite set $\calX$ is denoted as $\calP(\calX)$. Notation concerning the method of types follows~\cite{ZhouBook}. Given a vector $x^n = (x_1,x_2,\ldots,x_n) \in\calX^n$, the {\em type} or {\em empirical distribution} is denoted as $\hat{T}_{x^n}(a)=\frac{1}{n}\sum_{i=1}^n \mathbbm{1}\{x_i=a\},a\in\calX$. The set of types formed from length-$n$ sequences with alphabet $\calX$ is denoted as $\calP_{n}(\calX)$. Given $P\in\calP_{n}(\calX)$, the set of all sequences of length $n$ with type $P$, the {\em type class}, is denoted as $\calT^n_P$. For any $k\in\bbN$ and $(x_1,\ldots,x_k)\in\bbR^k$, let $\rvbPh_k(x_1,\ldots,x_k;\bmu,\bSigma)$ be the multivariate generalization of the Gaussian cdf, i.e.,
\begin{align}
\rvbPh_k(x_1,\ldots,x_k;\bmu,\bSigma)
&:=\int_{-\infty}^{x_1}\ldots\int_{-\infty}^{x_k}\calN(\bx;\bmu;\bSigma)\rmd \bx\label{def:kQ},
\end{align}
where $\calN(\bx; \bmu;\bSigma)$ is the probability density function of a $k$-variate Gaussian with mean vector $\bmu$ and non-singular covariance matrix $\bSigma$~\cite{Tanbook}. When $k=1$, we use $\Phi(x)$ to denote the cdf of a normal random variable with mean zero and variance one. Finally, for any $k\in\bbN$, we use $\mathbf{1}_k$ to denote a row vector of length $k$ with all elements being one and we use $\mathbf{0}_k$ similarly.

\subsection{Case of Known Number of Matches}
\label{sec:pf:known}
We first consider the case where the number of matched pairs of sequences across the two databases are known. Fix integers $(M_1,M_2,K,N,n)\in\bbN^5$ such that $M_1\geq M_2\geq K$. Let $\bX^N:=\{X_1^N,\ldots,X_{M_1}^N\}$ denote a database of $M_1$ sequences, where for each $i\in[M_1]$, $X_i^N$ is generated i.i.d. from an unknown distribution $P_i$ defined on the finite alphabet $\calX$. Let $\bY^n:=\{Y_1^n,\ldots,Y_{M_2}^n\}$ be another database of $M_2$ sequences, where for each $i\in[M_2]$, $Y_i^n$ is generated i.i.d. from an unknown distribution $Q_i$ defined on $\calX$. Without loss of generality, we assume that $N=n\alpha$ for some $\alpha\in\bbR_+$. For simplicity, we assume that the length of each sequence in a database is exactly the same. This assumption can be relaxed to databases that have different sequence lengths by using the method of types~\cite{csiszar1998mt} in the same spirit of \cite{unnikrishnan2015asymptotically,gutman1989asymptotically}. The only required changes is to calculate the types from sequences of different lengths and use sequence length as a parameter when designing the test.

Following the setting of Unnikrishnan~\cite{unnikrishnan2015asymptotically}, we assume that each sequence in each database is generated by a distinct distribution, i.e., there is no redundant element in either the set $P^{M_1}:=\{P_1,\ldots,P_{M_1}\}$ or the set $Q^{M_2}:=\{Q_1,\ldots,Q_{M_2}\}$. Furthermore, assume that there are $K$ pairs of sequences that are generated from the same distribution, i.e., there exists two subsets $\calA\subseteq[M_1]$ and $\calB\subseteq[M_2]$ such that $|\calA|=|\calB|=K$ and there exists a unique mapping $\sigma:\calA\to\calB$ such that for each $i\in\calA$, $X_i^n$ and $Y_{\sigma(i)}^n$ are generated from the same distribution, i.e., $P_i=Q_{\sigma(i)}$.

As argued by Unnikrishnan~\cite{unnikrishnan2015asymptotically}, if the distinct distribution assumption is removed for $\bY^n$, the problem reduces to repeated version of the $M$-ary classification problem~\cite{gutman1989asymptotically,zhou2018binary}, which is solved by testing whether $Y_i^n$ is generated from the same distribution as some sequence in $\bX^N$ for each $i\in[M_2]$. Furthermore, if $K=M_2=1$, this problem is exactly the $M_1$-ary classification problem. Therefore, the statistical matching problem significantly generalizes the $M_1$-ary classification problem.

Note that there are in total $T_K:={M_1\choose K}{M_2\choose K}K!$ possibilities of $K$-matches between the two databases. To represent each possibility (hypothesis) explicitly, we need the following definitions. Let $\calC_1^K$ be the collection of all ${M_1\choose K}$ subsets of $[M_1]$ with size $K$ and let $\calC_2^K$ be the collection of all ${M_2\choose K}$ subsets of $[M_2]$ with size $K$. For each $\calA^K\times\calB^K\in\calC_1^K\times\calC_2^K$, let $\calC_\mathrm{Per}^K(\calA^K,\calB^K)$ denote the set of all $K!$ unique mappings from $\calA^K\to\calB^K$. For each $l\in[T_K]$, a hypothesis $\rmH_l^K$ corresponds to a triple $(\calA_l^K,\calB_l^K,\sigma_l^K)\in\calC_1^K\times\calC_2^K\times\calC_\mathrm{Per}^K(\calA_l^K,\calB_l^K)$ such that for each $i\in\calA_l^K$, the sequences $X_i^N$ and $Y_{\sigma_l^K(i)}^n$ are generated from the same distribution $P_i=Q_{\sigma_l^K(i)}$.  Following~\cite{unnikrishnan2015asymptotically}, we also define $\calM_l^K:=\{(i,\sigma_l^K(i))\}_{i\in\calA_l^K}$ as the set of indices of matched sequences across the two databases, which is represented by a bipartie graph with weight $K$ between two sets of vertices with sizes $M_1$ and $M_2$, respectively.  Note that $(\calA_l^K,\calB_l^K,\sigma_l^K)$ is equivalent to $\calM_l^K$. We find it convenient to use the notation $(\calA_l^K,\calB_l^K,\sigma_l^K)$ in the presentation and derivation for some of our theoretical results. Furthermore, $\calA_l^K=\{i\in[M_1]:\exists~j\in[M_2]\mathrm{~s.t.~}(i,j)\in\calM_l^K\}$ and $\calB_l^K=\{j\in[M_2]:\exists~i\in[M_1]\mathrm{~s.t.~}(i,j)\in\calM_l^K\}$.

Our task is to design a test $\phi_{n,N}:\calX^{M_1N}\times\calY^{M_2n}\to\{\{\rmH_l^K\}_{l\in[T_K]},\rmH_{\rmr}\}$ to correctly identify the unique $K$-match between the two databases with a no match decision $\rmH_\rmr$ that calls for further investigation, i.e.,
\begin{itemize}
\item $\rmH_l^K$ with $l\in[T_K]$: the sequence $X_i^N$ and $Y_{\sigma_l^K(i)}^n$ are generated from the same distribution for each $i\in\calA_l^K$.
\item $\rmH_\rmr$: there is no $K$-match between the two databases $\bX^N$ and $\bY^n$.
\end{itemize} 
Under hypothesis $\rmH_\rmr$, $K=0$, $T_K=1$ and thus $(\calA_l^K,\calB_l^K,\calM_l^K)$ are all empty sets for $l\in[T_K]$. We remark that the adoption of the no match decision $\rmH_\rmr$ is consistent with the literature on statistical classification~\cite{gutman1989asymptotically,zhou2018binary} and statistical matching with the reject option~\cite{unnikrishnan2015asymptotically}. Furthermore, since the generating distributions are unknown, one would like to a design a universal test with good performance under any tuples of generating distributions. If the null hypothesis is not introduced, one aims to minimize the maximal mismatch probability under all hypotheses for all possible tuples of distributions. The resulting maximal mismatch probability can be very large, even close to one, since the worst case of generating distributions dominates. By introducing a null hypothesis, when it is hard to make a reliable decision under a particular tuple of generating distributions, the test can declare the null hypothesis and requires further investigation to make sure that the mismatch probability is small. In particular, as we shall show in our main results, below, allowing the additional null hypothesis enables us to derive tight results for optimal tests in both the large and small deviations regimes under the generalized Neyman-Pearson criterion~\cite{gutman1989asymptotically,zhou2018binary}.

For each $l\in[T_K]$, define the following set of generating distributions
\begin{align}
\calP_l^K:=
\big\{(\tilP^{M_1},\tilQ^{M_2})\in\calP(\calX)^{M_1+M_2}:~P_i=Q_j~\mathrm{iff}~(i,j)\in\calM_l^K\big\}\label{def:calp:lk}.
\end{align}
Note that $\calP_l^K$ denotes all possible tuples of generating distributions under hypothesis $\rmH_l^K$. To evaluate the performance of a test $\phi_{n,N}$, for each $l\in[T_K]$, under hypothesis $\rmH_l^K$ and any tuple of generating distributions $(P^{M_1},Q^{M_2})\in\calP_l^K$, we consider the following two probabilities:
\begin{align}
\beta(\phi_{n,N}|P^{M_1},Q^{M_2})&:=\Pr\big\{\phi_{n,N}(\bX^N,\bY^n)\notin\{\rmH_l^K,\rmH_\rmr\}\big\},\label{def:betal}\\
\zeta(\phi_{n,N}|P^{M_1},Q^{M_2})&:=\Pr\big\{\phi_{n,N}(\bX^N,\bY^n)=\rmH_\rmr\big\}\label{def:zetal}.
\end{align}
Note that $\beta(\phi_{n,N}|P^{M_1},Q^{M_2})$ is the mismatch probability, corresponding to the probability that an incorrect $K$-match is decided under hypothesis $\rmH_l^K$ when the generating distributions are $(P^{M_1},Q^{M_2})\in\calP_l^K$, while $\zeta(\phi_{n,N}|P^{M_1},Q^{M_2})$ is the false reject probability, corresponding to the probability that a no-match decision is output under hypothesis $\rmH_l^K$ with the same generating distributions. We remark that in the case of known number of matches, we only consider hypotheses that identify all $K$ matches simultaneously. The partial match case that identifies less than $K$ pairs of matched sequences is not considered. It would be worthwhile to generalize our analyses to cover partial match in future work.

Note that the generating distributions $(P^{M_1},Q^{M_2})$ are \emph{unknown} when we design and run the test $\phi_{n,N}$. However, we need the knowledge of these generating distributions to evaluate the performance of the test $\phi_{n,N}$. Ideally, we would like the test to be universal so that regardless of distributions $(P^{M_1},Q^{M_2})$, under each hypothesis, both mismatch and false reject probabilities are extremely small. Towards this goal, for each $l\in[T_K]$, given any $(P^{M_1},Q^{M_2})\in\calP_l^K$ and non-negative target false reject exponent $E\in\bbR_+$, for any sample sizes $(n,N)$, in the large deviations regime, we aim to characterize the first order expansion of the universal mismatch exponent $\lambda^*_{\mathrm{LD}}(n,N,E|P^{M_1},Q^{M_2}):=-\frac{1}{n}\log\beta_{\mathrm{LD}}^*(n,N,E|P^{M_1},Q^{M_2})$, where
\begin{align}
\beta_{\mathrm{LD}}^*(n,N,E|P^{M_1},Q^{M_2})
&:=\inf_{\substack{\phi_{n,N}:\\\zeta(\phi_{n,N}|P^{M_1},Q^{M_2})\leq \exp(-nE)}}\max_{l\in[T_K]}\sup_{(\tilP^{M_1},\tilQ^{M_2})\in\calP_l^K}\beta(\phi_{n,N}|\tilP^{M_1},\tilQ^{M_2})
\label{def:e*:ld}.
\end{align}
Note that in \eqref{def:e*:ld}, we take an inner supremum over $\sup_{(\tilP^{M_1},\tilQ^{M_2})\in\calP_l^K}$ to consider the worst case mismatch probability over all possible tuples of generating distributions under a particular hypothesis and we take another maximum over $l\in[T_K]$ to consider the maximal mismatch probability under all hypotheses. Thus, $\beta_{\mathrm{LD}}^*(n,N,E|P^{M_1},Q^{M_2})$ denotes the maximal universal mismatch probability under any hypothesis over all possible tuples of generating distributions of any test that ensures exponential decay of the false reject probability with rate of at least $E$ under hypothesis $\rmH_l^K$ with a particular tuple of generating distributions $(P^{M_1},Q^{M_2})\in\calP_l^K$. Equivalently, given any target mismatch exponent $\lambda\in\bbR_+$, one can  characterize the false reject exponent $E^*_{\mathrm{LD}}(n,N,\lambda|P^{M_1},Q^{M_2}):=-\frac{1}{n}\log\zeta_{\mathrm{LD}}^*(n,N,\lambda|P^{M_1},Q^{M_2})$, where
\begin{align}
\zeta_{\mathrm{LD}}^*(n,N,E|P^{M_1},Q^{M_2})
&:=\inf_{\substack{\phi_{n,N}:~\forall~l\in[T_K]\mathrm{~and~}(\tilP^{M_1},\tilQ^{M_2})\in\calP_l^K,\\\beta(\phi_{n,N}|\tilP^{M_1},\tilQ^{M_2})\leq \exp(-n\lambda)}}\zeta(\phi_{n,N}|P^{M_1},Q^{M_2})
\label{def:e*:ld:2}.
\end{align}
It follows from \eqref{def:e*:ld} and \eqref{def:e*:ld:2} that
\begin{align}
\lambda^*_{\mathrm{LD}}(n,N,E|P^{M_1},Q^{M_2})
=\sup\big\{\lambda\in\bbR_+:~E^*_{\mathrm{LD}}(n,N,\lambda|P^{M_1},Q^{M_2})\geq E\big\}\label{eqity:fr:mismatch}.
\end{align}
For ease of notation, we explicitly bound the false reject exponent $E^*_{\mathrm{LD}}(n,N,\lambda|P^{M_1},Q^{M_2})$ in Theorem \ref{ld:known} and the bound on $\lambda^*_{\mathrm{LD}}(n,N,E|P^{M_1},Q^{M_2})$ follows from \eqref{eqity:fr:mismatch}. In the achievability analysis, we use a distribution free test and prove its performance under any tuple of generating distributions $(P^{M_1},Q^{M_2})\in\calP_l^K$. In the converse part, we adopt the generalized Neyman-Pearson criterion of Gutman~\cite{gutman1989asymptotically}, which results in an universal performance constraint on the maximal mismatch probability over all possible tuples of generating distributions, i.e., the constraint inside the infimum of \eqref{def:e*:ld:2}. Since the generating distributions of observed sequences are unknown, the above formulation puts an universal constraint on the mismatch probability under each non-null hypothesis and a non-universal constraint on false reject probability under the null hypothesis. Such a setting is known as partial universal and has been considered in the literature~\cite{gutman1989asymptotically,zhou2018binary,zhou2022second,unnikrishnan2015asymptotically,he2019distributed}.

We remark that the non-universal constraint on the false reject probability can be easily generalized to be hold for a set of unknown generating distributions and the fundamental limit follows from \eqref{def:e*:ld} and \eqref{def:e*:ld:2}. Specifically, let $\calQ\subseteq\calP^{M_1+M_2}$ be a set of generating distributions. One can impose the constraint inside the infimum of \eqref{def:e*:ld} for all distributions $(P^{M_1},Q^{M_2})\in\calQ$ and obtain the corresponding fundamental limit $\beta_{\mathrm{LD}}^*(n,N,E|\calQ)$. Correspondingly, the equivalent fundamental limit analogous to \eqref{def:e*:ld:2} is $\zeta_{\mathrm{LD}}^*(n,N,E|\calQ)$ with an additional supremum over $(P^{M_1},Q^{M_2})\in\calQ$ required for \eqref{def:e*:ld:2}. It follows that 
\begin{align}
\beta_{\mathrm{LD}}^*(n,N,E|\calQ)&=\sup_{(P^{M_1},Q^{M_2})\in\calQ}\beta_{\mathrm{LD}}^*(n,N,E|P^{M_1},Q^{M_2}),\\
\zeta_{\mathrm{LD}}^*(n,N,E|\calQ)&=\sup_{(P^{M_1},Q^{M_2})\in\calQ}\zeta_{\mathrm{LD}}^*(n,N,E|P^{M_1},Q^{M_2}).
\end{align}
Thus, it suffices to study the partial universal fundamental limit in \eqref{def:e*:ld} and \eqref{def:e*:ld:2}.

In the small deviations regime, given any positive real number $\varepsilon\in(0,1)$, we aim to characterize the second-order expansion of the universal mismatch exponent $\lambda^*_{\mathrm{SD}}(n,N,\varepsilon|P^{M_1},Q^{M_2}):=-\frac{1}{n}\log \beta_{\rm{SD}}^*(n,N,\varepsilon|P^{M_1},Q^{M_2})$, where
\begin{align}
\beta_{\rm{SD}}^*(n,N,\varepsilon|P^{M_1},Q^{M_2})
:=\inf_{\substack{\phi_{n,N}:\\\zeta(\phi_{n,N}|P^{M_1},Q^{M_2})\leq \varepsilon}}\max_{l\in[T_K]}\sup_{(\tilP^{M_1},\tilQ^{M_2})\in\calP_l^K}\beta(\phi_{n,N}|\tilP^{M_1},\tilQ^{M_2}).
\label{def:l*:sd}
\end{align}
Similarly, $\beta_{\rm{SD}}^*(n,N,\varepsilon|P^{M_1},Q^{M_2})$ denotes the maximal universal mismatch probability under all tuples of generating distributions $(\tilP^{M_1},\tilQ^{M_2})\in\calP_l^K$ subject to a constant false reject probability $\varepsilon$ under hypothesis $\rmH_l^K$ for a particular tuple of generating distributions $(P^{M_1},Q^{M_2})\in\calP_l^K$.

Our first contribution is to characterize the first-order expansion of $\lambda_{\rm{LD}}^*(n,N,E|P^{M_1},Q^{M_2})$ and the second-order expansion of $\lambda_{\rm{SD}}^*(n,N,\varepsilon|P^{M_1},Q^{M_2})$.

\subsection{Case of Unknown Number of Matches}
\label{sec:pf:unknown}

A more practical setting is where the number of matches $K$ is \emph{unknown} a priori. This setting is more challenging since prior information on the number of matches is seldom available. In this case, the number of matches must be estimated and the pairs of matched sequences must be identified. 
To account for the possibility that no match between two databases exists, we define the null hypothesis, denoted as the reject hypothesis $\rmH_\rmr$, which corresponds to $K=0$. For each $K\in[M_2]$, we use $\calH_K$ to denote the set of all $T_K$ hypotheses when the number of matches is $K$. Thus, when the number of matches is unknown, the total number of hypotheses increases to $T+1$ where $T:=\sum_{K=1}^{M_2}T_K$. 

Correspondingly, our task is to design a test $\phi_{n,N}:\calX^{M_1N}\times\calY^{M_2n}\to\{\{\calH_K\}_{K\in[M_2]},\rmH_{\rmr}\}$ to correctly identify among the following hypotheses:
\begin{itemize}
\item $\rmH_l^K\in\calH_K$ where $K\in[M_2]$ and 
$l\in[T_K]$: the sequence $X_i^N$ and $Y_{\sigma_l^K}(i)$ are generated from the same distribution for each $i\in\calA_l^K\in\calC_1^K$.
\item $\rmH_\rmr$: there is no matched sequences between the two databases $\bX^N$ and $\bY^n$.
\end{itemize} 

Analogously to \eqref{def:calp:lk}, define the following set of possible distributions under the null hypothesis:
\begin{align}
\calP_0:=\big\{(\tilP^{M_1},\tilQ^{M_2})\in\calP(\calX)^{M_1+M_2}:~P_i\neq Q_j~\forall~(i,j)\in[M_1]\times[M_2]\big\}
\label{def:calP:r}.
\end{align}
To evaluate the performance of the test$\phi_{n,N}$, for each $K\in[M_2]$ and $l\in[T_K]$, under the non-null hypothesis $\rmH_l^K$ and generating distributions $(P^{M_1},Q^{M_2})\in\calP_l^K$, we consider the mismatch probability $\beta(\phi_{n,N}|P^{M_1},Q^{M_2})$ and the false reject probability $\zeta(\phi_{n,N}|P^{M_1},Q^{M_2})$ as in the case of known number of matches. Furthermore, under the null hypothesis $\rmH_\rmr$, for any tuple of generating distributions $(P^{M_1},Q^{M_2})\in\calP_0$, we also need the following false alarm probability:
\begin{align}
\eta(\phi_{n,N}|P^{M_1},Q^{M_2}):=\Pr\big\{\phi_{n,N}(\bX^N,\bY^n)\neq \rmH_\rmr\big\}\label{def:etar}.
\end{align}
Note that $\eta(\phi_{n,N}|P^{M_1},Q^{M_2})$ quantifies the probability that the test declares that there exists a matched pair of sequences between the two databases when there is none.

When the number of matches is unknown, we need to study the tradeoff among the probabilities of mismatch, false reject and false alarm. Analogous to the case where the number of matches is known, we derive achievability bounds for all three probabilities and discover the tradeoff among them for a variant of Unnikrishnan's test.

\subsection{Unnikrishnan's Test and Result}

We next recall Unnikrishnan's test for the case when the number of matches is \emph{known}. Recall that $\alpha=\frac{N}{n}$ is defined as the ratio between the lengths of sequences of the two databases.

To present the test, we need the following definitions. Given any distributions $(P,Q)\in\calP^2(\calX)$, for any positive constant $a\in\bbR_+$,
define the following generalized Jensen-Shannon divergence~\cite[Eq. (2.3)]{zhou2018binary} for a positive constant $\alpha\in\bbR_+$:
\begin{align}
\mathrm{GJS}(P,Q,\alpha):=\alpha D\bigg(P\bigg\|\frac{\alpha P+Q}{1+\alpha}\bigg)+D\bigg(Q\bigg\|\frac{\alpha P+Q}{1+\alpha}\bigg)\label{def:GJS}.
\end{align}
Note that $\mathrm{GJS}(P,Q,\alpha)$ measures the distance between the distributions $P$ and $Q$ via a linear combination of KL divergences between the distribution $P/Q$ and the convex combination $\frac{\alpha P+Q}{1+\alpha}$ of distributions $P$ and $Q$. When $\alpha=1$, $\mathrm{GJS}(P,Q,1)$ is twice of the Jensen-Shannon divergence~\cite[Eq. (4.1)]{lin1991divergence} when the weights of two distributions are the same. When $\alpha\to\infty$, $\mathrm{GJS}(P,Q,\alpha)\to D(Q\|P)$. The definition of the generalized Jensen-Shannon divergence in \eqref{def:GJS} dates back to Gutman~\cite{gutman1989asymptotically} in the asymptotic studies of statistical classification. In particular, $\mathrm{GJS}(P,Q,\alpha)$ is the first-order expansion of the small deviations regime for the optimal test in binary classification where one needs to determine whether a testing sequence $Y^n$ is generated i.i.d. from an unknown distributions $P$ or $Q$ while training sequences $(X_1^N,X_2^N)$ generated i.i.d. from $P$ and $Q$, respectively, are available. The GJS function has also been used in other studies including~\cite{
li2014,unnikrishnan2015asymptotically,zhou2018binary,he2019distributed,hsu2020binary,mahdi2021sequential,zhou2022second,zhou2022twophase}.

For any two sets of distributions $P^{M_1}=(P_1,\ldots,P_{M_1})$ and $Q^{M_2}=(Q_1,\ldots,Q_{M_2})$, for each $l\in[T_K]$, let
\begin{align}
\rmG_l^K(P^{M_1},Q^{M_2},\alpha)
&:=\sum_{i\in\calA_l^K}\mathrm{GJS}(P_i,Q_{\sigma_l^K(i)},\alpha)=\sum_{j\in\calB_l^K}\mathrm{GJS}(P_{(\sigma_l^K)^{-1}(j)},Q_j,\alpha)=\sum_{(i,j)\in\calM_l^K}\mathrm{GJS}(P_i,Q_j,\alpha)
\label{def:Gl},
\end{align}
where $(\calA_l^K,\calB_l^K,\sigma_l^K)$ specify the matched pairs of sequences under hypothesis $\rmH_l^K$ as explained in Section \ref{sec:pf:unknown}. Specifically, $\calA_l^K\in\calC_1^K$ identifies the set of indices of matched sequences of the database $\bX^N$, $\calB_l^K\in\calC_2^K$ identifies the set of indices of matched sequences of the database $\bY^n$ and $\sigma_l^K$ identifies the unique match among $\{x_i^n\}_{i\in\calA_l^K}$ and $\{y_j^n\}_{j\in\calB_l^K}$.  We remark that the three equivalent definitions of $\rmG_l^K(P^{M_1},Q^{M_2},\alpha)$ are related via the definitions of $(\calA_l^K,\calB_l^K,\sigma_l^K)$ and $\calM_l^K$. We find it convenient to use these three forms in different parts of our analyses.

Consider any realizations of two databases $\bx^N=\{x_1^N,\ldots,x_{M_1}^N\}$ and $\by^n=\{y_1^n,\ldots,y_{M_2}^n\}$. Let $\hatT_{\bx^N}:=(\hatT_{x_1^N},\ldots,\hatT_{x_{M_1}^n})$ and let $\hatT_{\by^n}:=(\hatT_{y_1^n},\ldots,\hatT_{y_{M_2}^n})$ be the collection of empirical distributions. For each $t\in[T_K]$, define the scoring function 
\begin{align}
\rmS_t^K(\bx^N,\by^n):=\rmG_t^K(\hatT_{\bx^N},\hatT_{\by^n},\alpha)\label{def:Sl}.
\end{align}
Furthermore, let
\begin{align}
l_K^*(\bx^N,\by^n)&:=\argmin_{t\in[T_K]}\rmS_t^K(\bx^N,\by^n),\label{def:l4unn}\\
h_K(\bx^N,y^n)&:=\min_{t\in[T_K]:l\neq l_K^*(\bx^N,\by^n)}\rmS_t^K(\bx^N,\by^n)\label{def:h4unn},
\end{align}
denote the index of the hypothesis whose scoring function is minimal and the value of the second minimal scoring function, respectively. Let $\lambda\in\bbR_+$ be any positive real number. For each $n\in\bbN$, define 
\begin{align}
\lambda_n:=\lambda+\frac{K|\calX|\log ((1+\alpha)n+1)}{n}\label{def:lambdan}.
\end{align}

Unnikrishnan's test operates as follows:
\begin{align}
\label{test:unn}
\phi_{n,N}^{\rmU,K}(\bx^N,\by^n,\lambda)
&=\left\{
\begin{array}{ll}
\rmH_l^K&\mathrm{if~}l_K^*(\bx^N,\by^n)=l\mathrm{~and~}h_K(\bx^N,\by^n)>\lambda_n,\\
\rmH_\rmr&\mathrm{if~}h_K(\bx^N,\by^n)\leq\lambda_n.
\end{array}
\right.
\end{align}
Note that when $M_2=K=1$, the test in \eqref{test:unn} reduces to a test for statistical classification of $M_1$ hypotheses, which classifies whether a test sequence $y^n$ is generated from the same distribution as one of the training sequence in $\bx^N$, studied by Gutman~\cite{gutman1989asymptotically}. We use $\phi_{n,N}^{\rmU,K=1}(\bx^N,y^n,\lambda)$ to denote the test for this special case.

The following performance guarantee was rephrased from \cite[Lemma 5 and Theorem 4.1]{unnikrishnan2015asymptotically}.
\begin{theorem}
\label{th:unn}
For each $l\in[T_K]$, under the hypothesis $\rmH_l^K$, Unnikrishnan's test is first-order asymptotically optimal by achieving $\beta_{\rm{SD}}^*(n,N,\varepsilon|P^{M_1},Q^{M_2})$ as $n\to\infty$ for any $\varepsilon\in(0,1)$ and any tuple of generating distributions $(P^{M_1},Q^{M_2})\in\calP_l^K$.
\end{theorem}
In a nutshell, Theorem \ref{th:unn} implies that Unnikrishnan's test ensures exponential decay of mismatch probabilities and is optimal in the generalized Neyman-Pearson sense as it minimizes the false reject probability among all tests that ensure exponential decay of mismatch probabilities with the same exponent rate.

Although Theorem \ref{th:unn} provides insights, Theorem \ref{th:unn} does not explicitly bound the false reject probability and thus does not reveal the tradeoff between the mismatch and false reject probabilities under each hypothesis. Furthermore, Theorem \ref{th:unn} holds when $n\to\infty$ and thus does not apply to the more practical finite sample size setting. The first contribution in this paper resolves these problems by providing theoretical bounds on the false reject probability in both the large and small deviations regimes. Our theoretical results demonstrate the performance tradeoff between the mismatch and false reject probabilities. Subsequently, we generalize our results to the case where the number of matches is unknown. Our main results are presented in the next two sections.

\section{Results for Known Number of Matches}
\label{sec:known}
In this section, we provide a refined analyses of Unnikrishnan's test $\phi_{n,N}^{\rmU,K}$ for known number of matches in both large and small deviations regimes. Furthermore, we demonstrate the optimality of Unnikrishnan's test in both regimes under the generalized Neyman-Pearson criterion.

\subsection{Asymptotic Intuition}
\label{asymp:intuition}
We first explain why Unnikrishnan's test $\phi_{n,N}^{\rmU,K}$ works using intuition obtained from the weak law of large numbers. Fix any $l\in[T_K]$ and consider any tuple of generating distributions $(P^{M_1},Q^{M_2})\in\calP_l^K$. Under hypothesis $\rmH_l^K$, the sequences $\{x_i^N\}_{i\in\calA_l^K}$ and $\{y_j^n\}_{j\in\calB_l^K}$ are matched and $\sigma_l^K:\calA_l^K\to\calB_l^K$ specifies the unique pair of matches, i.e., for each $i\in\calA_l^K$, $X_i^N$ and $Y_{\sigma_l^K(i)}^n$ are generated from the same distribution $P_i=Q_{\sigma_l^K(i)}$. For each $i\in\calA_l^K$, the weak law of large numbers implies that the empirical distribution $x_i^N$ tends to $P_i$ and the empirical distribution of $y_{\sigma_l^K(i)}^n$ tends to $Q_{\sigma_l^K(i)}=P_i$ and thus $\mathrm{GJS}(\hatT_{x^N},\hatT_{y^n_{\sigma_l^K(i)}},\alpha)$ tends to zero. Therefore, it follows from \eqref{def:Sl} that $\rmS_l^K(\bx^N,\by^n)$ tends to \emph{zero}. 

Consider any $t\in[T_K]$ such that $t\neq l$. Under hypothesis $\rmH_l^K$, the scoring function $\rmS_t^K(\bx^N,\by^n)$ satisfies
\begin{align}
\rmS_t^K(\bx^N,\by^n)
&=\sum_{(i,j)\in\calM_t^K}\mathrm{GJS}(\hatT_{x_i^N},\hatT_{y_j^n},\alpha)\\
&=\sum_{\substack{(i,j)\in(\calM_t^K\cap\calM_l^K)}}\mathrm{GJS}(\hatT_{x_i^N},\hatT_{y_j^n},\alpha)+\sum_{\substack{(i,j)\in(\calM_t^K\setminus\calM_l^K)}}\mathrm{GJS}(\hatT_{x_i^N},\hatT_{y_j^n},\alpha),\\
&\to 0+\sum_{\substack{(i,j)\in(\calM_t^K\setminus\calM_l^K)}}\mathrm{GJS}(P_i,Q_j,\alpha)\label{almostsure:conve1}\\
&>0\label{ineq},
\end{align}
where \eqref{almostsure:conve1} holds almost surely due to the weak law of large numbers and the continuous property of $\mathrm{GJS}(P,Q,\alpha)$ and 
\eqref{ineq} holds since $P_i=Q_j$ holds only for $(i,j)\in\calM_l^K$. Therefore, under hypothesis $\rmH_l^K$, the second minimal scoring function satisfies
\begin{align}
h_K(\bx^N,\by^n)
&=\min_{t\in[T_K]:t\neq l}\rmS_t^K(\bx^N,\by^n)\\
&\to\min_{t\in[T_K]:t\neq l}\rmG_t^K(P^{M_1},Q^{M_2},\alpha)\label{same3:almostsure}\\
&=:\Lambda_l(P^{M_1},Q^{M_2},K,\alpha)\label{def:Lambdal},
\end{align}
where \eqref{same3:almostsure} holds almost surely similarly to \eqref{almostsure:conve1}. Thus, asymptotically if the threshold $\lambda<\Lambda_l(P^{M_1},Q^{M_2},K,\alpha)$, no mismatch or false reject event can occur. Otherwise, if $\lambda\geq \Lambda_l(P^{M_1},Q^{M_2},K,\alpha)$, the false reject event occurs with probability tends to one asymptotically. In the following, we characterize the false reject probabilities as a function of the threshold $\lambda$ for the large and small deviations regimes, respectively.

\subsection{Large Deviations}
We first derive the tradeoff between the decay rates (exponent rates) of mismatch and false reject probabilities for the asymptotic setting of large sample size $n$. To present and discuss our result, we need the following definitions and lemma. Fix $l\in[T_K]$. Given any two tuples of distributions $(P^{M_1},Q^{M_2})\in\calP_l^K$ and $(\Omega^{M_1},\Psi^{M_2})\in\calP^{M_1+M_2}(\calX)$, define the following linear combination of KL divergences:
\begin{align}
E_l(P^{M_1},Q^{M_2},\Omega^{M_1},\Psi^{M_2},\alpha)
&:=\sum_{i\in[M_1]}\alpha D(\Omega_i\|P_i)+\sum_{j\in\calB_l^K}D(\Psi_j\|P_{(\sigma_l^K)^{-1}(j)})+\sum_{j\notin\calB_l^K}D(\Psi_j\|Q_j)
\label{def:El}.
\end{align}
As we show in \eqref{useEl}, below, under hypothesis $\rmH_l^K$, $\exp(-nE_l(P^{M_1},Q^{M_2},\Omega^{M_1},\Psi^{M_2},\alpha))$ upper bounds the probability of the set of sequences $(\bx^N,\by^n)$ such that $\hatT_{x_i^N}=\Omega_i$ for each $i\in[M_2]$ and $\hatT_{y_i^n}=\Psi_i$ for each $i\in[M_2]$ when the generating distributions are $(P^{M_1},Q^{M_2})\in\calP_l^K$.

Furthermore, given any non-negative real number $\lambda\in\bbR_+$, let
\begin{align}
F_l(P^{M_1},Q^{M_2},\alpha,\lambda,K)
&:=\min_{\substack{(t,s)\in[T_K]^2:\\t\neq s}}\min_{\substack{(\Omega^{M_1},\Psi^{M_2})\in\calP^{M_1+M_2}(\calX):\\\rmG_t^K(\Omega^{M_1},\Psi^{M_2},\alpha)\leq \lambda\\\rmG_s^K(\Omega^{M_1},\Psi^{M_2},\alpha)\leq \lambda}}E_l(P^{M_1},Q^{M_2},\Omega^{M_1},\Psi^{M_2},\alpha)\label{def:Fl}.
\end{align}
As we shall show in Theorem \ref{ld:known}, below, the expression $F_l(P^{M_1},Q^{M_2},\alpha,\lambda,K)$ lower bounds the error exponent of the false reject probability under hypothesis $\rmH_l^K$ with generating distributions $(P^{M_1},Q^{M_2})\in\calP_l^K$. Such a result resembles the definition of the type-II error exponent in binary hypothesis testing when the type-I error probability decays exponentially fast with a exponent rate of $\lambda$~\cite{blahut1974hypothesis}. To illustrate, we plot $F_l(P^{M_1},Q^{M_2},\alpha,\lambda,K)$ in Fig. \ref{illus:fl}.
\begin{figure}[tb]
\centering
\includegraphics[width=.5\columnwidth]{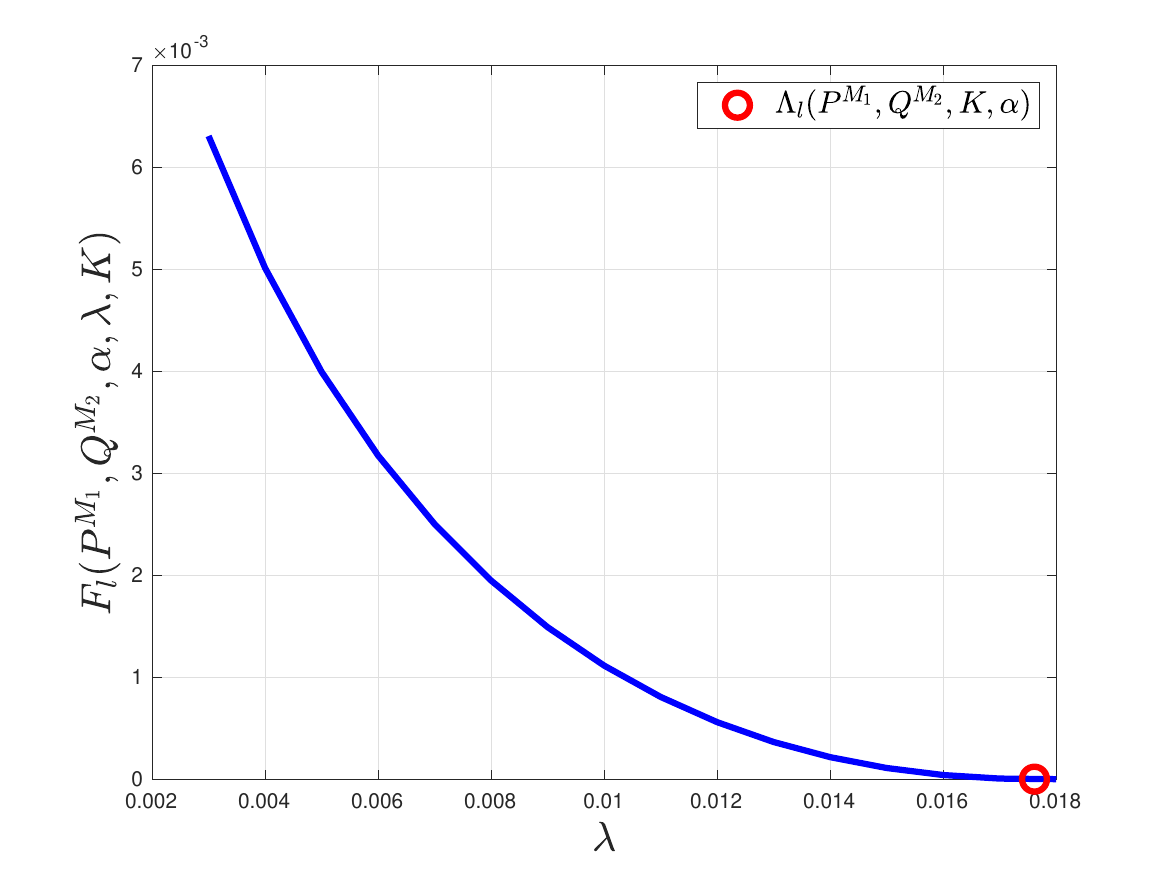}
\caption{Plot of $F_l(P^{M_1},Q^{M_2},\alpha,\lambda,K)$ as a function of $\lambda$ when $\alpha=2$, $M_1=4$, $M_2=K=2$, $(P_1,P_2,P_3,P_4)=\mathrm{Bern}(0.1,0.2,0.3,0.4)$ and $Q_1=P_1$, $Q_2=P_2$.}
\label{illus:fl}
\end{figure}

Finally, given $l\in[T_K]$ and any $(t,s)\in[T_K]^2$ such that $t\neq s$, for each $j\in[M_2]$, let
\begin{align}
f_j^{l,t,s}(P^{M_1},\alpha)
&:=\min_{\Psi\in\calP(\calX)} \Big(D(\Psi\|P_{(\sigma_l^K)^{-1}(j)})+\sum_{\substack{i\in[M_1]:\sigma_t^K(i)=j\\\mathrm{or~}\sigma_s^K(i)=j}}\alpha D(\Psi\|P_i)\Big),\label{def:fj}\\
g_j^{l,t,s}(P^{M_1},Q^{M_2},\alpha)
&:=\min_{\Psi\in\calP(\calX)} \Big(D(\Psi\|Q_j)+\sum_{\substack{i\in[M_1]:\sigma_t^K(i)=j\\\mathrm{or~}\sigma_s^K(i)=j}}\alpha D(\Psi\|P_i)\Big)\label{def:gj},
\end{align}
and let
\begin{align}
\Upsilon_l(P^{M_1},Q^{M_2},K,\alpha)
&:=\min_{\substack{(t,s)\in[T_K]^2:\\t\neq s}}\bigg(\sum_{j\in\calB_l^K:j\in(\calB_t^K\cup\calB_s^K)}f_j^{l,t,s}(P^{M_1},\alpha)+\sum_{j\in(\calB_l^K)^\rmc:j\in(\calB_t^K\cup\calB_s^K)}g_j^{l,t,s}(P^{M_1},Q^{M_2},\alpha)\bigg)\label{def:Upsilonl}.
\end{align}

Some properties of the exponent function $F_l(\cdot)$ are summarized in the following lemma.
\begin{lemma}
\label{prop:Fl}
The following claims hold for each $l\in[T_K]$.
\begin{enumerate}
\item Monotonicity: $F_l(P^{M_1},Q^{M_2},\alpha,\lambda,K)$ is non-increasing in $\lambda$.
\item Zero Condition: $F_l(P^{M_1},Q^{M_2},\alpha,\lambda,K)=0$ if and only if
\begin{align}
\lambda
&\geq\Lambda_l(P^{M_1},Q^{M_2},K,\alpha).
\end{align}
\item Maximal value: the maximal value of $F_l(P^{M_1},Q^{M_2},\alpha,\lambda,K)$ is 
\begin{align}
F_l(P^{M_1},Q^{M_2},\alpha,0,K)
&=\Upsilon_l(P^{M_1},Q^{M_2},K,\alpha).
\end{align}
\end{enumerate}
\end{lemma}
The proof of Lemma \ref{prop:Fl} is provided in Appendix \ref{proof:prop:Fl}.

Recall that $\alpha=\frac{N}{n}$ is defined as the ratio between the lengths of sequences of the two databases. Our first result establishes a relation between the mismatch and false reject exponent rates, through the exponent function $F_l(\cdot)$, for any threshold $\lambda>0$.
\begin{theorem}
\label{ld:known}
For each $l\in[T_K]$, given any non-negative real number $\lambda\in\bbR_+$ and any tuple of generating distributions $(P^{M_1},Q^{M_2})\in\calP_l^K$, 
\begin{align}
\lim_{n\to\infty}E_{\mathrm{LD}}^*(n,N,\lambda|P^{M_1},Q^{M_2})=F_l(P^{M_1},Q^{M_2},\alpha,\lambda,K).
\end{align}
\end{theorem}
The proof of Theorem \ref{ld:known} is given in Section \ref{proof:ld:known}. In the achievability part, we analyze Unnikrishnan's test and prove its asymptotic optimality in the large deviations regime. The analysis of the mismatch probability is similar to \cite[Theorem 4.1]{unnikrishnan2015asymptotically} and the analysis of the false reject probability follows from the method of types~\cite{csiszar1998mt}. The converse proof relies on a non-asymptotic converse bound on the false reject probability that we derive in Lemma \ref{converse:gnp}, which refines the asymptotic converse argument of \cite[Lemma 5]{unnikrishnan2015asymptotically}. Several remarks are as follows.

Firstly, Theorem \ref{ld:known} generalizes Unnikrishnan's result~\cite[Theorem 4.1]{unnikrishnan2015asymptotically} (cf. Theorem \ref{th:unn}) by deriving explicit bounds on the exponential decay rates of the false reject probability under each hypothesis. In particular, Unnikrishnan wrote that ``analytical expressions for the rejection exponents are difficult to obtain''. We manage to explicitly characterize the false reject exponent rate for Unnikrishnan's test as $F_l(P^{M_1},Q^{M_2},\alpha,\lambda,K)$ under hypothesis $\rmH_l^K$ when the generating distributions are $(P^{M_1},Q^{M_2})\in\calP_l^K$. The explicit performance guarantee for the false reject probability helps to further understand the impact of the threshold $\lambda$ of the test and guides the choice of $\lambda$ in practice.

Secondly, the achievability part of Theorem \ref{ld:known} reveals the tradeoff between the universal mismatch exponent and the false reject exponent under each hypothesis. Specifically, for each $l\in[T_K]$, under hypothesis $\rmH_l^K$, Unnikrishnan's test ensures that the mismatch probability decays exponentially fast with a rate no less than $\lambda$ under all tuples of unknown generating distributions $(\tilP^{M_1},\tilQ^{M_2})\in\calP_l^K$ and ensures that the false reject probability decays exponentially fast with the exponent rate of $F_l(P^{M_1},Q^{M_2},\alpha,\lambda,K)$ under the particular tuple of generating distributions $(P^{M_1},Q^{M_2})\in\calP_l^K$. Under each hypothesis $l\in[T_K]$ and any unknown generating distributions $(P^{M_1},Q^{M_2})\in\calP_l^K$, it follows from Claim (i) of Lemma \ref{prop:Fl} that if the mismatch exponent rate $\lambda$ increases, the false reject exponent rate $F_l(P^{M_1},Q^{M_2},\alpha,\lambda,K)$ decreases and vice versa. In particular, Claims (ii) and (iii) of Lemma \ref{prop:Fl} show that for each $l\in[T_K]$, if the universal mismatch exponent rate is large enough so that $\lambda\geq \Lambda_l(P^{M_1},Q^{M_2},K,\alpha)$, the false reject exponent rate $F_l(P^{M_1},Q^{M_2},\alpha,\lambda,K)=0$ while the false reject exponent rate takes the largest value $\Upsilon_l(P^{M_1},Q^{M_2},K,\alpha)$ when the mismatch exponent rate $\lambda=0$. 

Thirdly, the converse part of Theorem \ref{ld:known} demonstrates the optimality of Unnikrishnan's test in the large deviations regime in the generalized Neyman-Pearson sense of Gutman~\cite{gutman1989asymptotically}, which has been adopted in the follow up studies, e.g.,~\cite{zhou2018binary,zhou2022second,he2019distributed}. Specifically, the generalized Neyman-Pearson criterion requires each test to guarantee universal exponential decay of mismatch probabilities with a rate no less than $\lambda$. Such a criterion generalizes the Neyman-Pearson criterion for hypothesis testing with known generating distributions to hypothesis testing problems with unknown distributions and allows one to derive optimality guarantee for corresponding tests, e.g.,~\cite{zhou2018binary,zhou2022second,unnikrishnan2015asymptotically}.

Finally, when specialized to $M_2=K=1$, our derivation of the false reject exponent rate specializes to the corresponding result for statistical classification of multiple hypotheses~\cite{gutman1989asymptotically}. In this case, we have $T_1=M$, $\calA_l^K=\{l\}$, $\calB_l^K=1$, $\sigma_l^K=1$, $\calM_l^K=(l,1)$, $Q=P_l$ under hypothesis $\rmH_l^K$ for each $l\in[T_1]$ and
\begin{align}
\Lambda_l(P^{M_1},Q,\alpha)
&=\min_{t\in[M_1]:t\neq l}\mathrm{GJS}(P_t,P_l,\alpha),
\\
E_l(P^{M_1},Q,\Omega^M,\Psi,\alpha)
&=\sum_{i\in[M_1]}\alpha D(\Omega_i\|P_i)+D(\Psi\|P_l).
\end{align}
Thus, under hypothesis $\rmH_l^K$, the false reject exponent rate is
\begin{align}
F_l(P^{M_1},Q,\alpha,\lambda)
&=\min_{\substack{(t,s)\in[M_1]^2:t\neq s}}\min_{\substack{(\Omega^{M_1},\Psi)\in\calP^{M_1+1}(\calX):\\\rm{GJS}(\Omega_s,\Psi,\alpha)\leq \lambda\\\rm{GJS}(\Omega_t,\Psi,\alpha)\leq \lambda}}E_l(P^{M_1},Q,\Omega^{M_1},\Psi,\alpha)\\
&=\min_{\substack{(t,s)\in[M_1]^2:t\neq s}}\min_{\substack{(\Omega^{M_1},\Psi)\in\calP^{M+1}(\calX):\\\rm{GJS}(\Omega_s,\Psi,\alpha)\leq \lambda\\\rm{GJS}(\Omega_t,\Psi,\alpha)\leq \lambda}} \Big(\alpha D(\Omega_t\|P_t)+\alpha D(\Omega_s\|P_s)+D(\Psi\|Q)\Big)\label{ld:frej:classify},
\end{align}
which is precisely the false reject exponent rate for $M$-ary classification~\cite[Eq. (5.103)]{zhou2018binary}. Finally, the maximal value for the false reject exponent rate is
\begin{align}
F_l(P^{M_1},Q,\alpha,0)
&=\min_{(t,s)\in[M_1]^2:t\neq s}\min_{\Psi\in\calP(\calX)}\Big(\alpha D(\Psi\|P_t)+\alpha D(\Psi\|P_s)+D(\Psi\|Q)\Big),
\end{align}
which is precisely the corresponding result for statistical classification~\cite[Eq. (5.104)]{zhou2018binary}.

\subsection{Small Deviations}
To shed lights on the performance of Unnikrishnan's test for finite $n$, in the small deviations regime, we derive a second-order expansion of the universal mismatch exponent subject to a constant false reject probability. To present our results, we need the following definitions. Given any distributions $(P,Q)\in\calP^2(\calX)$, for each $x\in\calX$, we need the following two information densities (log likelihood ratios)~\cite[Eq. (3.4)]{zhou2018binary}
\begin{align}
\imath_1(x|P,Q,\alpha)&:=\log\frac{(1+\alpha)P(x)}{\alpha P(x)+Q(x)}\label{def:i1},\\
\imath_2(x|P,Q,\alpha)&:=\log\frac{(1+\alpha)Q(x)}{\alpha P(x)+Q(x)}\label{def:i2}.
\end{align}
Fix $l\in[T_K]$ and a tuple of generating distributions $(P^{M_1},Q^{M_2})\in\calP_l^K$. For each $(t_1,t_2)\in[T_K]^2$, define the covariance function
\begin{align}
\mathrm{Cov}_{t_1,t_2}(P^{M_1},Q^{M_2},\alpha)
\nn&=\bigg(
\sum_{(i,j)\in\calM_{t_1}^K}\sum_{(\bari,\barj)\in\calM_{t_2}^K:\bari=i}\alpha\mathrm{Cov}\big(\imath_1(X_i|P_i,Q_j,\alpha),\imath_1(X_{\bari}|P_{\bari},Q_{\barj},\alpha)\big)\\*
&\qquad\qquad\qquad+\sum_{(i,j)\in\calM_{t_1}^K}\sum_{(\bari,\barj)\in\calM_{t_2}^K:\barj=j}\mathrm{Cov}\big(\imath_2(Y_j|P_i,Q_j,\alpha),\imath_2(Y_{\barj}|P_{\bari},Q_{\barj},\alpha)\big)\bigg)\label{def:vt1t2},
\end{align}
where $X_i\sim P_i$, $Y_j\sim Q_j$ for each $i\in[M_1]$ and $j\in[M_2]$, and the random variables $(X_1,\ldots,X_{M_1})$ and $(Y_1,\ldots,Y_{M_2})$ are independent of each other.

Furthermore, define the set 
\begin{align}
\calI_l(P^{M_1},Q^{M_2}):=\{t\in[T_K]:~t\neq l,~\rmG_t^K(P^{M_1},Q^{M_2},\alpha)=\Lambda_l(P^{M_1},Q^{M_2},K,\alpha)\}\label{def:calIl}.
\end{align}
Let $\tau_l:=|\calI_l(P^{M_1},Q^{M_2})|$ denote the size of the set. Furthermore, define the function $\rmo:\calI_l(P^{M_1},Q^{M_2})\to [\tau_l]$ such that for each $t\in\calI_l(P^{M_1},Q^{M_2})$, $\rmo(t)=i$ if $t$ is the $i$-th smallest element in the set $\calI_l(P^{M_1},Q^{M_2})$. Now define the square matrix $\bV^l(P^{M_1},Q^{M_2},\alpha)$ of dimension $\tau_l\times \tau_l$ such that for each $(t_1,t_2)\in[\tau_l]^2$, $\bV_{t_1,t_2}^l(P^{M_1},Q^{M_2},\alpha)=\mathrm{Cov}_{\rmo^{-1}(t_1),\rmo^{-1}(t_2)}(P^{M_1},Q^{M_2},\alpha)$.

Recall that $N=n\alpha$ and $\rvbPh$ denotes the multivariate Gaussian cdf. Recall the definition of $\Lambda_l(P^{M_1},Q^{M_2},K,\alpha)$ in \eqref{def:Lambdal}. 
Given any $\varepsilon\in(0,1)$, let
\begin{align}
\nu_l^*(\varepsilon|P^{M_1},Q^{M_2},K,\alpha)
&:=\inf\bigg\{L\in\bbR:~\rvbPh_{\tau_l}\big(L\times \mathbf{1}_{\tau_l};\mathbf{0}_{\tau_l};\bV^l(P^{M_1},Q^{M_2},\alpha)\big)\geq 1-\varepsilon\bigg\},\label{def:nul*}\\
\chi_l^*(n,\varepsilon|P^{M_1},Q^{M_2},K,\alpha)
&:=\Lambda_l(P^{M_1},Q^{M_2},K,\alpha)-\frac{\nu_l^*(\varepsilon|P^{M_1},Q^{M_2},K,\alpha)}{\sqrt{n}}\label{def:lambdal*}.
\end{align}
We remark that the dimension $\tau_l$ can degenerate to one. In this case, the multivariate Gaussian cdf above in \eqref{def:nul*} reduces to a Gaussian cdf function and $\nu_l^*(\varepsilon|P^{M_1},Q^{M_2},K,\alpha)=\Phi^{-1}(\varepsilon)\sqrt{\mathrm{Cov}_{t,t}(P^{M_1},Q^{M_2},\alpha)}$, where $\Phi^{-1}(\cdot)$ is the inverse cdf function of a Gaussian random variable with mean zero and variance one and $t$ is the only element in  $\calI_l(P^{M_1},Q^{M_2})$. A numerical illustration of $\chi_l^*(n,\varepsilon|P^{M_1},Q^{M_2},K,\alpha)$ is available in Fig. \ref{illus:chi}.

\begin{figure}[tb]
\centering
\includegraphics[width=.5\columnwidth]{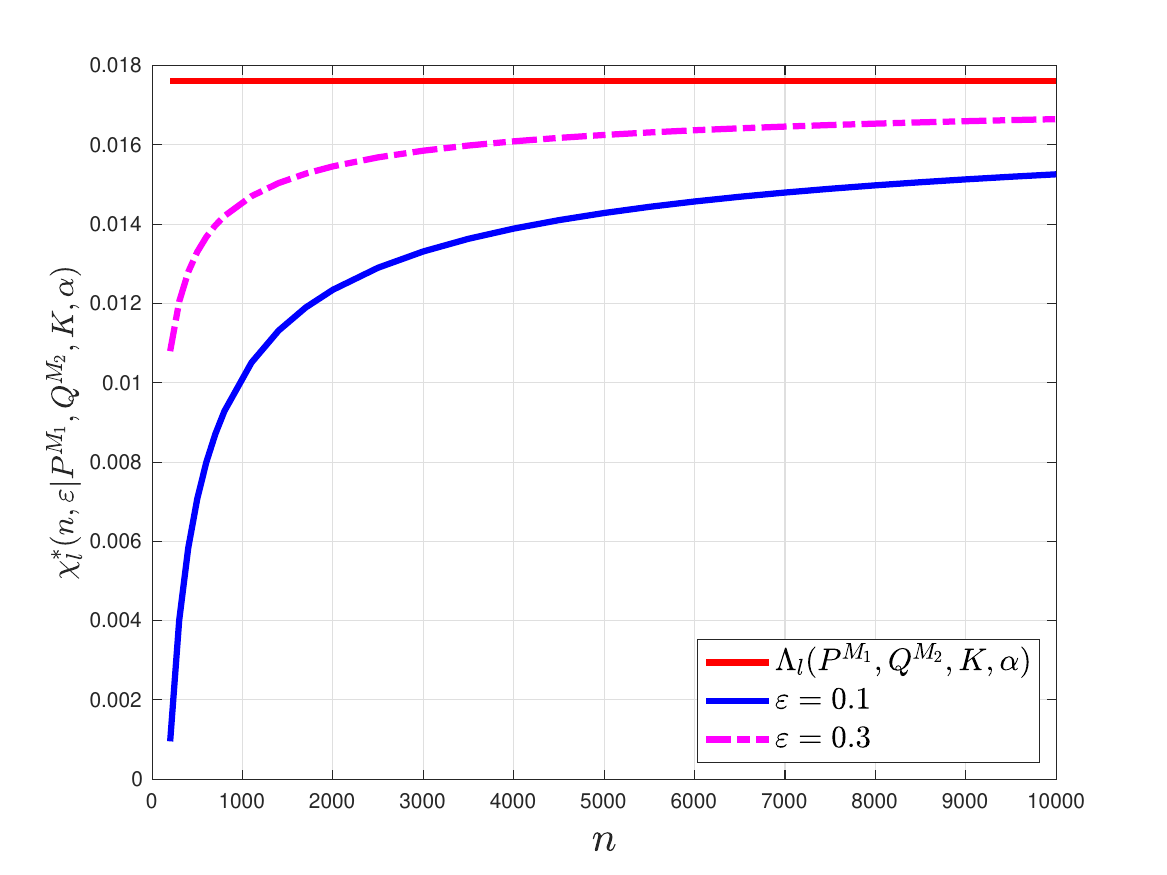}
\caption{Plot of $\chi_l^*(n,\varepsilon|P^{M_1},Q^{M_2},K,\alpha)$ as a function of $n$ for various values of $\varepsilon$ when $\alpha=2$, $M_1=4$, $M_2=K=2$, $(P_1,P_2,P_3,P_4)=\mathrm{Bern}(0.1,0.2,0.3,0.4)$ and $Q_1=P_1$, $Q_2=P_2$. In this case, $\tau_l=1$. Thus, the multivariate Gaussian cdf degenerates to the univariate Gaussian cdf.}
\label{illus:chi}
\end{figure}

Our result for the small deviations regime states as follows.
\begin{theorem}
\label{sr:known}
For each $l\in[T_K]$, given any positive real number $\varepsilon\in(0,1)$ and any tuple of generating distributions $(P^{M_1},Q^{M_2})\in\calP_l^K$, 
\begin{align}
\lambda_{\mathrm{SD}}^*(n,N,\varepsilon|P^{M_1},Q^{M_2})
&=\chi_l^*(n,\varepsilon|P^{M_1},Q^{M_2},K,\alpha)+O(\log n/n).
\end{align}
\end{theorem}
The proof of Theorem \ref{sr:known} is given in Section \ref{proof:sr:known}. In the achievability part, we analyze Unnikrishnan's test in \eqref{test:unn} to demonstrate its optimality. Compared with the large deviations analyses in Theorem \ref{ld:known}, the difference lies in the way that we bound the false reject probability. To do so, we apply Taylor expansion to the scoring function of $\rmS_l^K(\bx^N,\by^n)$ for typical realizations $(\bx^N,\by^n)$ of the database sequences and subsequently apply the multivariate Berry-Esseen theorem \cite[Cor. 29]{watanabe2015} (see also \cite[Cor. 1.1]{Tanbook}) to the derived non-asymptotic bounds. As shown in the proof of Theorem \ref{sr:known}, we have independent but not identically distributed random vectors with covariance matrix not equal to the identity matrix. Therefore, the multivariate Berry-Esseen theorem for random vectors with identity covariance matrix (cf. \cite{Ben03}) and the generalizations of the Berry-Esseen theorem for functions of i.i.d. random vectors~(cf. \cite[Prop. 1]{molavianjazi2015second} and \cite[Prop. 1]{iri2015third}) are not applicable. The converse proof uses the non-asymptotic converse bound in Lemma \ref{converse:gnp} on the false reject probability under the generalized Neyman-Pearson criterion and proceeds similarly to the achievability proof using the Taylor expansion of scoring functions and the multivariate Berry-Esseen theorem.

Compared with Theorem \ref{ld:known} that characterizes the \emph{asymptotic} performance tradeoff between the mismatch and false reject probabilities under each hypothesis when the sample size $n\to\infty$, Theorem \ref{sr:known} provides approximation to the performance of Unnikrishnan's test in the finite sample size setting. In particular, Theorem \ref{sr:known} elucidates the tradeoff between the target false reject probability $\varepsilon$ under a particular tuple of distributions and the universal mismatch exponent rate $\chi_l^*(n,\varepsilon|P^{M_1},Q^{M_2},K,\alpha)$ when the sample size $n$ is \emph{finite}. Such a result is known as a second-order expansion in the small deviations regime since $\chi_l^*(n,\varepsilon|P^{M_1},Q^{M_2},K,\alpha)$ characterizes the more refined second-order expansion $\frac{\nu_l^*(\varepsilon|P^{M_1},Q^{M_2},K,\alpha)}{\sqrt{n}}$ beyond the first-order expansion $\Lambda_l(P^{M_1},Q^{M_2},K,\alpha)$ that is also revealed from the large deviations analyses. Theorem \ref{ld:known} complements Theorem \ref{ld:known} by showing that the false reject probability is a constant when the mismatch probability decays with a exponent rate close to $\Lambda_l(P^{M_1},Q^{M_2},K,\alpha)$ on the order of $\frac{1}{\sqrt{n}}$ while Theorem \ref{ld:known} shows that the false reject probability decays exponentially fast if the mismatch exponent rate $\lambda$ is strictly less than $\Lambda_l(P^{M_1},Q^{M_2},K,\alpha)$ asymptotically.

Finally, when specialized to $M_2=K=1$, Theorem \ref{sr:known} establishes a second-order expansion in the small deviations regime for statistical classification of multiple hypotheses and refines~\cite[Theorem 4.1]{zhou2018binary} by removing a restrictive uniqueness assumption. Under this case, $T_K=M$, $\calA_l^K=\{l\}$, $\calB_l^K=1$, $\sigma_l^K=1$, $\calM_l^K=(l,1)$, $Q=P_l$ under hypothesis $\rmH_l^K$ for each $l\in[M]$. Fix $l\in[M]$ and a tuple of generating distributions $(P^{M_1},Q)\in\calP_l^1$. For each $(t_1,t_2)\in[M]^2$,
\begin{align}
\bV_{t_1,t_2}^l(P^{M_1},Q,\alpha)
\nn&=\bigg(\bbo(t_1=t_2)\alpha\mathrm{Cov}\Big(\imath_1(X_{t_1}|P_{t_1},Q,\alpha),\imath_1(X_{t_2}|P_{t_2},Q,\alpha)\Big)\\*
&\qquad\qquad\qquad+\mathrm{Cov}\Big(\imath_2(Y|P_{t_1},Q,\alpha),\imath_2(Y|P_{t_2},Q,\alpha)\Big)\bigg),
\end{align}
where $(X_{t_1},X_{t_2},Y)\sim P_{t_1}\times P_{t_2}\times Q$. If $\lambda\leq \chi_l^*(n,\varepsilon|P^{M_1},Q^{M_2},K,\alpha)$, the false reject probability is upper bound by $\varepsilon$ asymptotically. Our results refine \cite[Theorem 4.1]{zhou2018binary} by removing the constraint on unknown generating distributions $P^{M_1}$, which requires that the minimizer of $\min_{t\in[M_1]:t\neq l}\mathrm{GJS}(P_l,P_t,\alpha)$ is unique for each $l\in[M_1]$. For example, when $M_1=4$ and $(P_1,P_2,P_3,P_4)=\mathrm{Bern}(0.2,0.1,0.327,0.4)$, we have $\tau_1=2$. This case could not be handled by~\cite[Theorem 4.1]{zhou2018binary}. In constrast, using Theorem \ref{sr:known}, we can obtain the second-order expansion in the small deviations regime. In particular, we have 
\begin{align}
\Lambda_l(P^{M_1},Q^{M_2},K,\alpha)&=0.275,\\
\calI_l(P^{M_1},Q^{M_2})&=\{2,3\},\\
\bV^l(P^{M_1},Q^{M_2},\alpha)
&=\left[\begin{array}{cc}
0.0189&-0.12\\
-0.12&0.0174
\end{array}
\right].
\end{align}
A numerical plot of the second-order expansion $\chi_l^*(n,\varepsilon|P^{M_1},Q^{M_2},K,\alpha)$ is provided in Fig. \ref{illus:chi:r2}.

\begin{figure}[tb]
\centering
\includegraphics[width=.5\columnwidth]{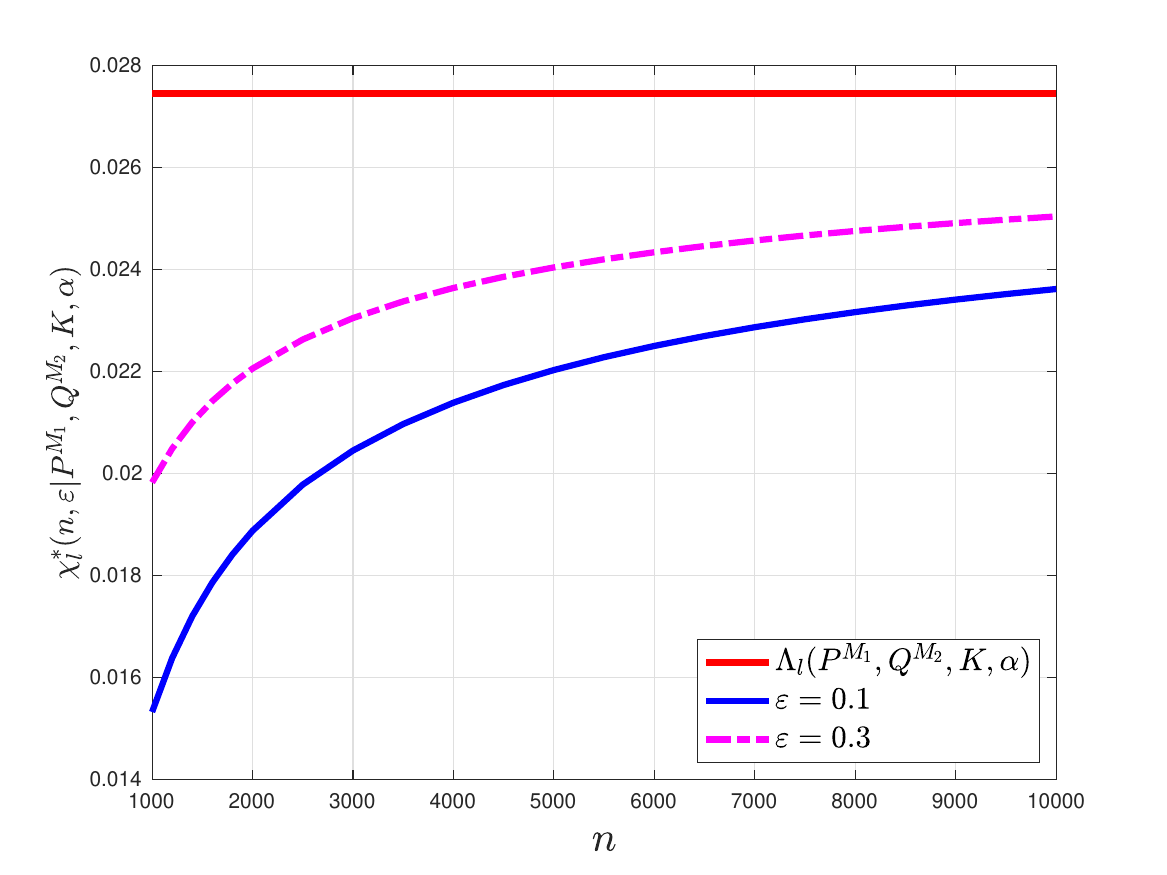}
\caption{Plot of $\chi_l^*(n,\varepsilon|P^{M_1},Q^{M_2},K,\alpha)$ as a function of $n$ for various values of $\varepsilon$ when $\alpha=2$, $M_1=4$, $M_2=K=1$, $(P_1,P_2,P_3,P_4)=\mathrm{Bern}(0.1,0.2,0.327,0.4)$ and $Q=P_1$. In this case, $\tau_l=2$ and the multivariate Gaussian cdf is required to obtain the second-order expansion $\chi_l^*(n,\varepsilon|P^{M_1},Q^{M_2},K,\alpha)$.}
\label{illus:chi:r2}
\end{figure}

\subsection{Comparison with A Simple Test}
In this section, we study the test in Algorithm \ref{simpletest}~\cite[Section IV.A]{unnikrishnan2015asymptotically} and compare its performance with Unnikrishnan's test in \eqref{test:unn}. The test in Algorithm \ref{simpletest} repeatedly checks whether each sequence in the database $\bY^n$ is generated from the same distribution as one of the sequences in the other database $\bX^N$ and outputs a decision for each sequence $Y_j^n$ with $j\in[M_2]$ using Unnikrishnan's test $\phi_{n,N}^{\rmU,K=1}$.

\begin{algorithm}[bt]
\caption{A Simple test $\phi_{n,N}^{\rmS}$}
\label{simpletest}
\begin{algorithmic}
\REQUIRE Two databases $\bx^N=(x_1^N,\ldots,x_{M_1}^N)$ and $\by^n=(y_1^n,\ldots,y_{M_2}^n)$ and a threshold $\lambda\in\bbR_+$
\ENSURE A sequence matching decision $\hatH$
\FOR{$j\in[M_2]$}
\STATE Use Unnikrishnan's test with $K=1$ to obtain $\hat{\rmH}_j=\phi_{n,N}^{\rmU,K=1}(\bx^N,y_j^n)$ with threshold $\lambda_n$ in \eqref{def:lambdan}, where $\hat{\rmH}_j$ takes values in $\{\{\rmH_l^K\}_{l\in[M_1]},\rmH_\rmr\}$, $\rmH_\rmr$ means that $y_j^n$ is not matched to any sequence in $\bx^N$ and $\rmH_l^K$ means that $y_j^n$ is matched to sequence $x_l^N$
\ENDFOR
\end{algorithmic}
\end{algorithm}

Similar to \cite{unnikrishnan2015asymptotically}, we assume that $M_2=K$ and thus $\calB_l^K=[M_2]$ for all $l\in[T_K]$. The assumption of $M_2=K$ simplifies the analyses of mismatch and false reject probabilities for the simple test. In this case, a false reject event occurs if for any $j\in[M_2]$, the sequence $y_j^n$ is given a reject decision, i.e., $\hat{\rmH}_j=H_\rmr$ and an mismatch event occurs if for any $j\in[M_2]$, the sequence $y_j^n$ is misclassified, i.e., $\hat{\rmH}_j=i$ such that $\sigma_l^K(i)\neq j$. 

For each $l\in[T_K]$ and any tuple of generating distributions $(\tilP^{M_1},\tilQ^{M_2})\in\calP_l^K$, the mismatch probability of $\phi_{n,N}^{\rmS}$ satisfies
Fix each $l\in[T_K]$ and any tuple of distributions $(P^{M_1},Q^{M_2})\in\calP_l^K$, the mismatch and false reject probabilities of 
\begin{align}
\beta(\phi_{n,N}^{\rmS}|P^{M_1},Q^{M_2})
&\leq \sum_{j\in[M_2]}\Pr\{\phi_{n,N}^{\rmU,K=1}(\bX^N,Y_j^n)\notin\{\rmH_\rmr,\rmH_l^K\}\},\\
\zeta(\phi_{n,N}^{\rmS}|P^{M_1},Q^{M_2})
&\leq \sum_{j\in[M_2]}\Pr\{\phi_{n,N}^{\rmU,K=1}(\bX^N,Y_j^n)=\rmH_\rmr\},
\end{align}

Recursively applying Theorem \ref{ld:known} with $M_2=K=1$, we have 
\begin{enumerate}
\item the universal mismatch exponent satisfies
\begin{align}
\liminf_{n\to\infty}-\frac{1}{n}\log\Big(\max_{l\in[T_K]}\sup_{(\tilP^{M_1},\tilQ^{M_2})\in\calP_l^K}\beta(\phi_{n,N}^{\rmS}|\tilP^{M_1},\tilQ^{M_2})\Big)&\geq \lambda,\label{simple:same:mismatch}
\end{align}
\item for each $l\in[T_K]$ and any tuple of generating distributions $(P^{M_1},Q^{M_2})\in\calP_l^K$, the false reject exponent satisfies
\begin{align}
\liminf_{n\to\infty}-\frac{1}{n}\log\zeta(\phi_{n,N}^{\rmS}|P^{M_1},Q^{M_2})
&\geq \min_{j\in[M_2]}F_l(P^{M_1},Q_j,\alpha,\lambda,K)\\
&=\min_{j\in[M_2]}\min_{\substack{(t,s)\in[M_1]:\\t\neq s}}
\min_{\substack{(\Omega^{M_1},\Psi)\in\calP^{M_1+M_2}(\calX):\\\mathrm{GJS}(\Omega_t,\Psi,\alpha)\leq \lambda\\
\mathrm{GJS}(\Omega_s,\Psi,\alpha)\leq \lambda}}\Big(\sum_{i\in[M_1]}\alpha D(\Omega_i\|P_i)+D(\Psi\|Q_j)\Big).
\end{align}
\end{enumerate}
Compared with Theorem \ref{ld:known}, we find that the simple test in Algorithm \ref{simpletest} achieves the same universal mismatch exponent rate as the optimal test $\phi_{n,N}^{\rmU}$ but the false reject exponent rate can be smaller. Specifically, in Appendix \ref{just:simpleworse}, we show that if $K>1$,
\begin{align}
F_l(P^{M_1},Q^{M_2},\alpha,\lambda,K)
\geq\min_{j\in[M_2]}F_l(P^{M_1},Q_j,\alpha,\lambda,K)\label{simpleworse}.
\end{align}
However, whether the above inequality holds strictly or with equality remains to be explored. Our numerical example in Fig. \ref{freject_com} implies that the simulated false reject exponent of both tests are equal as the sample length $n\to\infty$, although the non-asymptotic performance of Unnikrishnan's test is better.

Analogously, one could also apply the small deviations results in Theorem \ref{sr:known} to obtain corresponding results for the simple test. However, as the second-order expansion in the small deviations regime involves rather complicated equations, we omit this analysis here and instead, we compare the numerical performances of both tests by simulation in the next subsection.

\subsection{A Numerical Example}
Consider the binary alphabet $\calX=\{0,1\}$ and let $M_1=4$, $M_2=K=2$. This yields $T_K=12$ hypotheses as discussed in Section \ref{sec:pf:known}, covering the set of all possible matches of sequences across two databases. Without loss of generality, assume that hypothesis $\rmH_1$ is true such that $P_1=Q_1$ and $P_2=Q_2$. If not otherwise stated, we consider distributions $(P^4,Q^2)$ such that $P_1=Q_1=\mathrm{Bern}(0.1)$, $P_2=Q_2=\mathrm{Bern}(0.11)$, $P_3=\mathrm{Bern}(0.12)$ and $P_4=\mathrm{Bern}(0.13)$. Note that the distributions are required in the numerical example to generate the sequence samples. The tests studied in this paper proceed without knowledge of generating distributions and our theoretical results hold under any tuple of generating distributions.

Firstly, to illustrate the theoretical findings in Theorem \ref{ld:known}, we present simulated results of mismatch and false reject exponents of Unnikrishnan's test in \eqref{test:unn} and the simple test in Algorithm \ref{simpletest}, both with the threshold $\lambda=10^{-4}$. Similar to \cite{zhou2022second}, the choice of such a small $\lambda$ is selected to ensure that the error probabilities can be numerically approximated without too many simulation trials. Ideally, one would simulate the maximal mismatch probability over all possible tuples of generating distributions $(P^{M_1},Q^{M_2})\in\calP_1^K$ and calculate the corresponding exponent rate. However, in practice, we could only simulate the maximal mismatch probability for a limited set of distributions. In our numerical here, we consider any $(P^{M_1},Q^{M_2})$ such that $P_1=Q_1$ can be any distribution among Bernoulli distributions with parameters $(0.1,0.105,0.115,0.125.0.135,0.14)$ and other distributions are set as in the last paragraph. For simplicity, we use $\calQ$ to denote the set of above distributions. For each sample size $n$, we run each test $10^6$ times and calculate the empirical exponents for the maximal mismatch probability under all considered tuples of generating distributions. The simulated mismatch error exponents of both tests are plotted and compared with theoretical bounds in Fig. \ref{misclassify_com}. We observe that simulated mismatch exponents of both tests converge towards the theoretical bound $\lambda$ as the sample size $n$ increases, which numerically confirm theoretical findings in Theorems \ref{ld:known} and \eqref{simple:same:mismatch}, respectively. Finally, we remark that the gap between the simulated exponent values and the asymptotic theoretical limit results from the uncharacterized high-order terms in the large deviations regime~\cite{altug14a,altug14b}.

\begin{figure}[tb]
\centering
\includegraphics[width=.5\columnwidth]{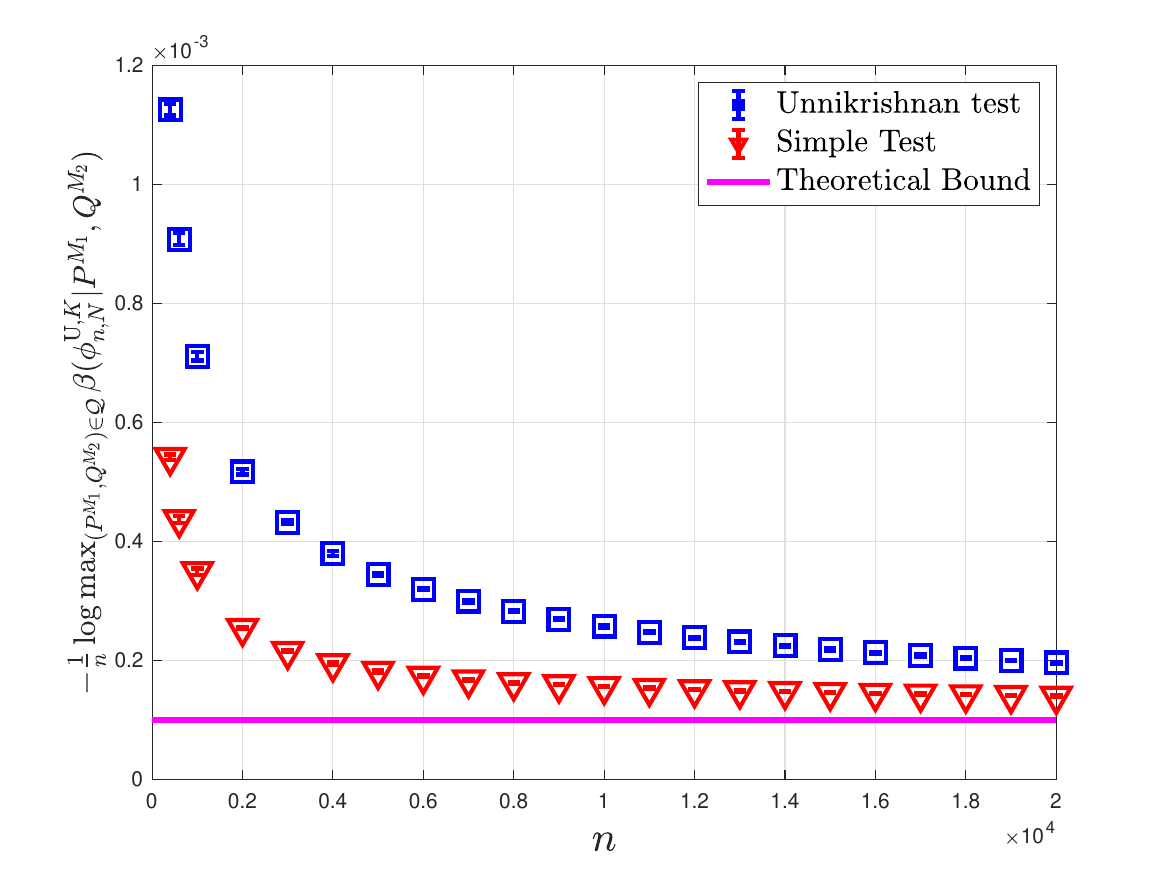}
\caption{Simulated exponents for the maximal mismatch probabilities of Unnikrishnan's test and the simple test with threshold $\lambda=10^{-4}$ when $M_1=4$, $M_2=K=2$ under any tuple of generating distributions $(P^{M_1},Q^{M_2})\in\calQ$. The error bar denotes two standard deviations below and above the mean value. The plot empirically confirms the mismatch exponents of both tests converge to the asymptotic lower bounded $\lambda$ claimed in Theorems \ref{ld:known} and \eqref{simple:same:mismatch} as the sample size $n$ increases. On the right hand side, the simulated false reject exponents of both tests are plotted.}
\label{misclassify_com}
\end{figure}

\begin{figure}[tb]
\centering
\includegraphics[width=.5\columnwidth]{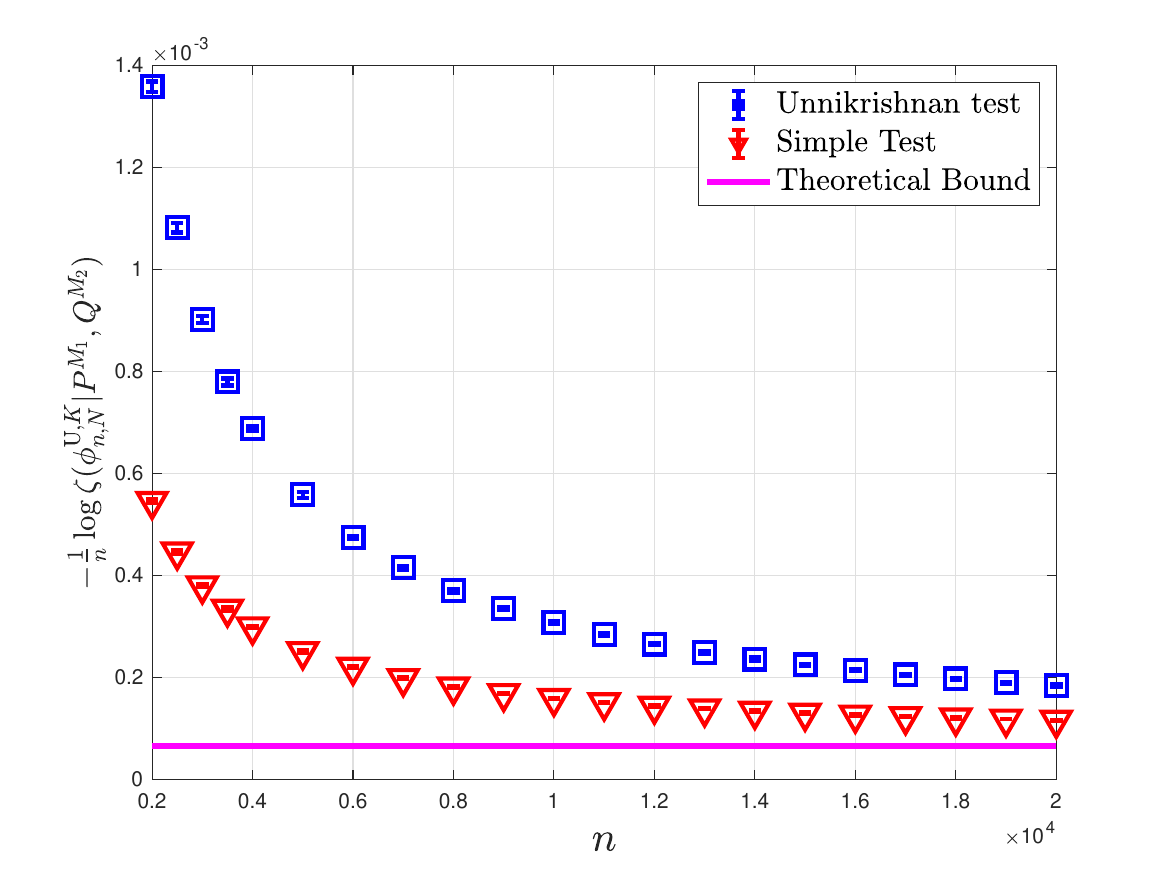}
\caption{Simulated false reject exponents of Unnikrishnan's test and the simple test with threshold $\lambda=10^{-4}$ when $M_1=4$, $M_2=K=2$ for generating distributions $P_1=Q_1=\mathrm{Bern}(0.1)$, $P_2=Q_2=\mathrm{Bern}(0.11)$, $P_3=\mathrm{Bern}(0.12)$ and $P_4=\mathrm{Bern(0.13)}$. The error bar denotes two standard deviations below and above the mean value. As observed, the plots empirically confirms that the false reject exponent of Unnikrishnan's test converges towards the theoretical bound in Theorem \ref{ld:known}.}
\label{freject_com}
\end{figure}

Furthermore, we numerically simulate the false reject exponents of both tests and compare the results with theoretical benchmarks in Fig. \ref{freject_com} under the same setting of Fig. \ref{misclassify_com} except we fix $P_1=Q_1=\mathrm{Bern}(0.1)$.  As observed, Unnikrishnan's test achieves larger non-asymptotic false reject exponent than the simple test and the false reject exponents of both tests converge to the same theoretical value as the sample size $n$ increases, which numerically confirms theoretical findings in \eqref{simpleworse}. However, whether the simple test achieves strictly smaller false reject exponent requires further investigation.

Finally, to demonstrate the tightness of the second-order expansion in the small deviations regime in Theorem \ref{sr:known}, we simulate Unnikrishnan's test in \eqref{test:unn}. We set the target false reject probability as $\varepsilon=0.1$ and set the threshold $\lambda=\chi_l^*(n,\varepsilon|P^{M_1},Q^{M_2},K,\alpha)$ as in Theorem \ref{sr:known}. The simulated false reject probabilities for $P_1=\mathrm{Bern}(0.1)$ and various choices of $(P_2,P_3,P_4)$ are plotted in Fig \ref{freject_coro}. The distributions are chosen so that the second-order expansion $\chi_l^*(\varepsilon|P^{M_1},Q^{M_2},K,\alpha)$ in the small deviations regime is positive and large enough for the simulated sample sizes. As observed, the simulated false reject probabilities converge towards the target value $\varepsilon$ as the sample size $n$ increases. The difference in the speed of approaching for different curves mainly results from the uncharacterized higher order terms in $O(\log n/n)$, which becomes relatively smaller for larger $\Lambda_l(P^{M_1},Q^{M_2},K,\alpha)$.

\begin{figure}[tb]
\centering
\includegraphics[width=.5\columnwidth]{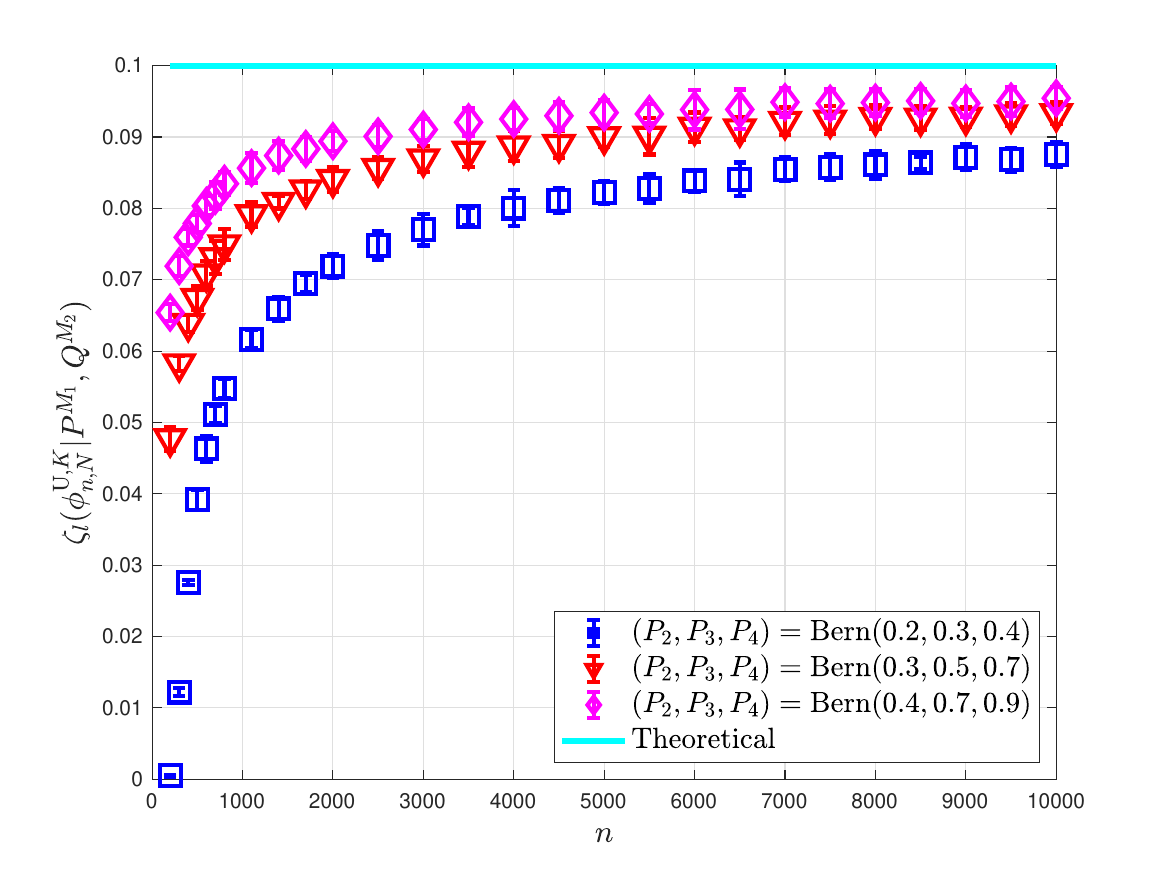}
\caption{Simulated false reject probabilities of Unnikrishnan's test with the threshold $\lambda$ in Theorem \ref{sr:known} when $M_1=4$, $M_2=K=2$ for $P_1=\mathrm{Bern}(0.1)$ and $(P_2,P_3,P_4)$ selected as three different pairs of Bernoulli distributions. The error bar denotes two standard deviations below and above the mean value. For all considered cases, as the sample size $n$ increases, the simulated false reject probability approaches the theoretical bound $\varepsilon$ as desired.}
\label{freject_coro}
\end{figure}

\section{Results for Unknown Number of Matches}
\label{sec:unknown}

In this section, we generalize the results in Section \ref{sec:known} to the more practical case where the number of matches between two databases $\bX^N$ and $\bY^n$ is \emph{unknown}. Compared with the case of known number of matches, we need to search for all possibilities of the number of matches from $0$ to $\min\{M_1,M_2\}$ and analyze the additional false alarm probability $\eta(\phi_{n,N}|P^{M_1},Q^{M_2})$ in \eqref{def:etar} that bounds the probability of the event where matches of sequences are claimed while there is no matched sequence pair. We first propose a slightly modified version of Unnikrishnan's test in Algorithm \ref{modifedtest}, which first estimates the unknown number of matches $K$ as $\hatK$ and then runs Unnikrishnan's test in \eqref{test:unn} with $K=\hatK$. Subsequently, we analyze the achievable performance of the test based on the same techniques that were used to prove results in Section \ref{sec:known}. 

\subsection{Variant of Unnikrishnan's Test}
Recall that $M_1\geq M_2$ and $T_K={M_1\choose K}{M_2\choose K}K!$. When the number of matches $K$ is unknown, it can be any number from $0$ to $M_2$. Thus, the total number of hypotheses is $T+1=\sum_{K\in[M_2]}T_K+1$. For each $K\in[M_2]$, we use $\calH_K$ to denote the set of all $T_K$ hypotheses with $K$ matches between the two databases $\bx^N=(x_1^N,\ldots,x_{M_1}^N)$ and $\by^n=(y_1^n,\ldots,y_{M_2}^n)$. Recall the definitions of $\rmS_l^K(\cdot)$ in \eqref{def:Sl}, $l_K^*(\cdot)$ and $h_K(\cdot)$ in \eqref{def:l4unn} and \eqref{def:h4unn}, respectively. Furthermore, for each $K\in[M_2]$, let
\begin{align}
\underline{S}_{K}(\bx^N,\by^n):=\rmS_{l_K^*(\bx^N,\by^n)}^K(\bx^N,\by^n),
\end{align}
be the minimal scoring function when the number of matches is $K$. Given any $(\lambda_1,\lambda_2)\in\bbR_+^2$, let
\begin{align}
\lambda_{i,n}:=\lambda_i+\frac{K|\calX|\log ((1+\alpha)n+1)}{n},~i\in[2]\label{def:lambdain}.
\end{align}

We present a test in Algorithm \ref{modifedtest} that first estimates the number of matches by comparing a certain scoring function with the threshold $\lambda_{1,n}$ and then implements Unnikrishnan's test with the threshold $\lambda_{2,n}$ using the estimated number of matches.

\begin{algorithm}[bt]
\caption{A test $\phi_{n,N}^{\rmM}$ for unknown number of matches}
\label{modifedtest}
\begin{algorithmic}
\REQUIRE Two databases $\bx^N=(x_1^N,\ldots,x_{M_1}^N)$ and $\by^n=(y_1^n,\ldots,y_{M_2}^n)$ and two thresholds $(\lambda_1,\lambda_2)\in\bbR_+^2$
\ENSURE A decision $\hat{\rmH}$ 
\STATE Set $\hatK=M_2$
\WHILE{$\hatK>0$}
\STATE Calculate $\underline{S}_{\hatK}(\bx^N,\by^n)$ with $K$ replaced by $\hatK$.
\IF{$\underline{S}_{\hatK}(\bx^N,\by^n)\leq \lambda_{1,n}$}
\STATE Return $\hat{\rmH}=\phi_{n,N}^{\rmU,\hatK}(\bx^N,\by^n)$ using Unnikrishnan's test in \eqref{test:unn} with the threshold $\lambda_n$ replaced by $\lambda_{2,n}$.
\ELSIF{$\hatK=0$}
\STATE Return $\hat{\rmH}=\rmH_\rmr$
\ENDIF
\STATE $\hatK=\hatK-1$
\ENDWHILE
\end{algorithmic}
\end{algorithm}
Again, we provide intuition into why such a test works well using the weak law of large numbers. It suffices to establish that the number of matches is estimated accurately given that the weak law of large numbers has already been used to show asymptotic convergence of Unnikrishnan's test in the case of known number of matches (cf. Section \ref{asymp:intuition}). We first assume that true hypothesis is $\rmH_l^K$ where $l\in[T_K]$ for some $K\in[M_2]$. In this case, the true number of matches $K$ is strictly positive. For any $\hatK>K$, we have
\begin{align}
\underline{S}_{\hatK}(\bx^N,\by^n)
&=\min_{t\in[T_{\hatK}]}\sum_{(i,j)\in\calM_t^{\hatK}}\mathrm{GJS}(\hatT_{x_i^N},\hatT_{y_j^n},\alpha)\\
&\to \min_{t\in[T_{\hatK}]}\sum_{(i,j)\in\calM_t^{\hatK}}\mathrm{GJS}(P_i,Q_j,\alpha)\label{almostsure:conver2}\\
&=\min_{t\in[T_{\hatK}]}\sum_{(i,j)\in(\calM_t^{\hatK}\cap(\calM_l^K)^\rmc)}\mathrm{GJS}(P_i,Q_j,\alpha)\label{strictpos0},
\end{align}
where \eqref{almostsure:conver2} holds almost surely due to the weak law of large numbers and the continuous property of the function $\mathrm{GJS}(P,Q,\alpha)$ and \eqref{strictpos0} follows since under hypothesis $\rmH_l^K$, $P_i=Q_j$ if $(i,j)\in\calM_l^K$.

Given any $\hatK>1$, for any hypothesis $\rmH_t^{\hatK}$ with $\hatK$ matched pairs specified by the set $\calM_t^{\hatK}$, one can always find another hypothesis 
$\rmH_j^{\hatK-1}$ with $\hatK-1$ matched pairs specified by $\calM_j^{\hatK-1}$ such that $\calM_t^{\hatK}\cap(\calM_j^{\hatK-1})^\rmc$ is a single matched pair $(\bari,\barj)\in[M_1]\times[M_2]$. Thus,
\begin{align}
\sum_{(i,j)\in\calM_t^{\hatK}\cap(\calM_l^K)^\rmc}\mathrm{GJS}(P_i,Q_j,\alpha)
\geq \sum_{(i,j)\in\calM_j^{\hatK-1}\cap(\calM_l^K)^\rmc}\mathrm{GJS}(P_i,Q_j,\alpha)\label{nonincreasing},
\end{align}
where the equality holds if $(\bari,\barj)\in\calM_l^K$. It follows that
\begin{align}
\nn&\min_{\hatK\in[M_2]:\hatK>K}\min_{t\in[T_{\hatK}]}\sum_{(i,j)\in(\calM_t^{\hatK}\cap(\calM_l^K)^\rmc)}\mathrm{GJS}(P_i,Q_j,\alpha)\\
&=\min_{t\in[T_{K+1}]}\sum_{(i,j)\in\calM_t^{K+1}\cap(\calM_l^K)^\rmc}\mathrm{GJS}(P_i,Q_j,\alpha)\label{just:asymp:intui:step1}\\
&=\min_{\substack{(i,j)\in[M_1]\times[M_2]:\\i\notin\calA_l^K,~j\notin\calB_l^K}}\mathrm{GJS}(P_i,Q_j,\alpha)\label{just:asymp:intui:step2}\\
&=:\rmG_{\mathrm{min}}^{l,K}(P^{M_1},Q^{M_2},\alpha)\label{def:GminlK},
\end{align}
where \eqref{just:asymp:intui:step1} follows from the result in \eqref{nonincreasing}, and \eqref{just:asymp:intui:step2} follows since $\calM_t^{K+1}\cap(\calM_l^K)^\rmc$ is a single pair of indices $(i,j)\in[M_1]\times[M_2]$ such that $i\notin\calA_l^K$ and $j\notin\calB_l^K$.

Furthermore, when $\hatK=K$,
\begin{align}
\underline{S}_{\hatK}(\bx^N,\by^n)
&=\min_{t\in[T_K]}\sum_{(i,j)\in\calM_t^K}\mathrm{GJS}(\hatT_{x_i^N},\hatT_{y_j^n},\alpha)\\
&\to \min_{t\in[T_{\hatK}]}\sum_{(i,j)\in\calM_t^K}\mathrm{GJS}(P_i,Q_j,\alpha)\label{same:almostsure}\\
&\leq \sum_{(i,j)\in\calM_l^K}\mathrm{GJS}(P_i,Q_j,\alpha)\\
&=0\label{sameaspos0},
\end{align}
where \eqref{same:almostsure} follows almost surely similarly to \eqref{almostsure:conver2} and \eqref{sameaspos0} follows for the same reason as \eqref{strictpos0}.

Combining \eqref{strictpos0}, \eqref{def:GminlK} and \eqref{sameaspos0}, we conclude that if the threshold $\lambda_1$ satisfies $0<\lambda_1<\rmG_{\mathrm{min}}^{l,K}(P^{M_1},Q^{M_2},\alpha)$, the estimated number of matches satisfies $\hatK=K$. It then follows from the asymptotic intuition in Section \ref{asymp:intuition} that no mismatch or false alarm event occurs as $n\to\infty$ when the number of matches $K$ is known.

We next consider the case that $K=0$, which corresponds to the null hypothesis $\rmH_\rmr$ that there are no matches. Note that in this case, for any $\hatK>0$, it follows that
\begin{align}
\min_{\hatK\in[M_2]}\underline{S}_{\hatK}(\bx^N,\by^n)
&\to \min_{\hatK\in[M_2]}\min_{t\in[T_{\hatK}]}\sum_{(i,j)\in\calM_t^{\hatK}}\mathrm{GJS}(P_i,Q_j,\alpha)\label{same2:almostsure}\\
&=\min_{(i,j)\in[M_1]\times[M_2]}\mathrm{GJS}(P_i,Q_j,\alpha)\\
&=:\rmG_{\mathrm{min}}^0(P^{M_1},Q^{M_2},\alpha),
\end{align}
where \eqref{same2:almostsure} follows almost surely similarly to \eqref{almostsure:conver2}. Therefore, when $0<\lambda_1<\rmG_{\mathrm{min}}^0(P^{M_1},Q^{M_2},\alpha)$, the output of the test $\phi_{n,N}^\rmM$ is $\rmH_\rmr$ when $K=0$, which implies that no false alarm occurs.

Combining the above intuition together, we conclude that the test in Algorithm \ref{modifedtest} has good asymptotic performance. In the next two subsections, analogous to Section \ref{sec:known} for known number of matches, we characterize its performance by characterizing the tradeoff among the probabilities of mismatch, false reject and false alarm in both large and small deviations regimes.

One might wonder whether it is possible to run Unnikrishnan's test in \eqref{test:unn} repeatedly instead of using our proposed two-phase test when the number of matches is unknown. Unfortunately, this is not practical. Since the real number of matches $K$ is unknown, one would need to run the test in \eqref{test:unn} repeatedly for all $\hatK\in[M_2]$. However, each run of the test in \eqref{test:unn} returns either a decision $\rmH_l^{\hatK}$ or the null hypothesis $\rmH_\rmr$ when $\hatK$ is used as the number of matches. There is no obvious way of combining all test results since these runs may give inconsistent decisions.

\subsection{Large Deviations}
We first consider the asymptotic case when the sample size $n\to\infty$ and characterize the achievable exponent rate of all three probabilities under each hypothesis. Given any $K\in[M_2]$ and $l\in[T_K]$, for any $\lambda_1\in\bbR_+$, define the exponent function
\begin{align}
f_{l,K}(\lambda_1,P^{M_1},Q^{M_2})
&:=\min_{\substack{(i,j)\in[M_1]\times[M_2]:\\i\notin\calA_l^K,j\notin\calB_l^K}}\min_{\substack{(\Omega,\Psi)\in\calP^2(\calX):\\\mathrm{GJS}(\Omega,\Psi,\alpha)\leq\lambda_1}}\Big(\alpha D(\Omega\|P_i)+D(\Psi\|Q_j)\Big)\label{def:flambda1}.
\end{align}
Furthermore, let
\begin{align}
f_0(\lambda_1,P^{M_1},Q^{M_2})
&:=\min_{\substack{(i,j)\in[M_1]\times[M_2]}}\min_{\substack{(\Omega,\Psi)\in\calP^2(\calX):\\\mathrm{GJS}(\Omega,\Psi,\alpha)\leq\lambda_1}}\Big(\alpha D(\Omega\|P_i)+D(\Psi\|Q_j)\Big)\label{def:flambda0}.
\end{align}
As we shall see, $f_{l,K}(\lambda_1,P^{M_1},Q^{M_2})$ bounds mismatch and false reject exponents while $f_0(\lambda_1,P^{M_1},Q^{M_2})$ bounds the false alarm exponent rate. Note that when $K=M_2$, $f_{l,K}(\lambda_1,P^{M_1},Q^{M_2})=\infty$ since the feasible set is empty for the outer minimization.

Useful properties of $f_{l,K}(\lambda_1,P^{M_1},Q^{M_2})$ are summarized in the following lemma and proved in Appendix \ref{proof:prop:flambda1}.
\begin{lemma}
\label{prop:flambda1}
The following claims hold for each $K\in[M_2]$ and $l\in[T_K]$.
\begin{enumerate}
\item Monotonicity: $f_{l,K}(\lambda_1,P^{M_1},Q^{M_2})$ is non-increasing in $\lambda_1$.
\item Zero condition: $f_{l,K}(\lambda_1,P^{M_1},Q^{M_2})=0$ if and only if
\begin{align}
\lambda_1\geq \rmG_{\mathrm{min}}^{l,K}(P^{M_1},Q^{M_2},\alpha),
\end{align}
where $\rmG_{\mathrm{min}}^{l,K}(\cdot)$ was defined in \eqref{def:GminlK}. 
\item Maximal value: the maximal value of $f_{l,K}(\lambda_1,P^{M_1},Q^{M_2})$ is
\begin{align}
f_{l,K}(0,P^{M_1},Q^{M_2})=\min_{\substack{(i,j)\in[M_1]\times[M_2]:\\i\notin\calA_l^K,j\notin\calB_l^K}}D_{\frac{\alpha}{1+\alpha}}(P_i\|Q_j),
\end{align}
where the R\'enyi divergence~\cite{renyi1961measures} of order $\gamma\in\bbR_+$ between distributions $(P,Q)\in\calP(\calX)^2$ is defined as
\begin{align}
D_{\gamma}(P\|Q):=\frac{1}{\gamma-1}\log \bigg(\sum_{x \in\calX}P^{\gamma}(x)Q^{ 1-\gamma}(x)\bigg)\label{def:renyi}.
\end{align}

\end{enumerate}
\end{lemma}
The claims in Lemma \ref{prop:flambda1} hold for $f_0(\lambda_1,P^{M_1},Q^{M_2})$ by taking $\calM_l^K$ as an empty set and replacing $\rmG_{\mathrm{min}}^{l,K}(\cdot)$ with $\rmG_{\mathrm{min}}^0(\cdot)$.

Recall the definitions of $\calP_l^K$ in \eqref{def:calp:lk}, $\calP_0$ in \eqref{def:calP:r} and the exponent function $F_l(P^{M_1},Q^{M_2},\alpha,\lambda_2,K)$ in \eqref{def:Fl}. Fix $K\in[M_2]$ and $l\in[T_K]$ for some $K\in[M_2]$. We have the following analogous result to Theorem \ref{ld:known}.
\begin{theorem}
\label{ld:unknown}
Given any positive real numbers $(\lambda_1,\lambda_2)\in\bbR_+^2$, the test in Algorithm \ref{modifedtest} ensures
\begin{enumerate}
\item for any tuple of generating distributions $(P^{M_1},Q^{M_2})\in\calP_l^K$,
\begin{align}
\liminf_{n\to\infty}-\frac{1}{n}\log\beta(\phi_{n,N}^\rmM|P^{M_1},Q^{M_2})&\geq \min\big\{\lambda_1,\lambda_2,f_{l,K}(\lambda_1,P^{M_1},Q^{M_2})\big\},\label{ld:betal:exponent}\\
\liminf_{n\to\infty}-\frac{1}{n}\log\zeta(\phi_{n,N}^\rmM|P^{M_1},Q^{M_2})&\geq \min\Big\{\lambda_1,f_{l,K}(\lambda_1,P^{M_1},Q^{M_2}),F_l(P^{M_1},Q^{M_2},\alpha,\lambda_2,K)\Big\},\label{ld:zetal:exponent}
\end{align}
\item for any tuple of generating distributions $(P^{M_1},Q^{M_2})\in\calP_0$,
\begin{align}
\liminf_{n\to\infty}-\frac{1}{n}\log\eta(\phi_{n,N}^\rmM|P^{M_1},Q^{M_2})&\geq f_0(\lambda_1,P^{M_1},Q^{M_2})\label{ld:eta:exponent}.
\end{align}
\end{enumerate}
\end{theorem}
The proof of Theorem \ref{ld:unknown} follows by analyzing the performance of the test in Algorithm \ref{modifedtest} and is given in Section \ref{proof:ld:unknown}. Note that when hypothesis $l\in[T_K]$ is true, a mismatch event occurs if the number of matches is estimated incorrectly, i.e., $\hatK\neq K$ or if the test in \eqref{test:unn} produces an incorrect decision $\rmH_\rmt$ where $t\in[T_K]$ and $t\neq l$ while a false reject event occurs if the number of matches is estimated incorrectly or if the test in \eqref{test:unn} outputs the reject decision $\rmH_\rmr$. When the null hypothesis is true, a false alarm event occurs if the estimated number of matches is non-zero while $K=0$. The proof of Theorem \ref{ld:unknown} proceeds by upper bounding the probability of the above events using the method of types.

Compared with Theorem \ref{ld:known}, we cannot have distribution free bound for the mismatch exponent when the number of matches is unknown. As explained above in the proof sketch, a mismatch event occurs either when the number of matches is estimated positive and incorrectly and when the test makes a mismatch error when the number of matches is estimated correctly. The latter event is exactly the mismatch event when the number of matches is known, whose probability is distribution free as shown in Theorem \ref{ld:known}. However, the probability of the former event concerning the error in estimating the number of matches is not distribution free. Thus, when the number of matches is unknown, we cannot have distribution free bounds for the mismatch exponent. This result also leads to the the difficulty of deriving a matching converse result. In contrast to Theorem \ref{ld:known} where we prove the optimality of Unnikrishnan's test under the generalized Neyman-Pearson criterion, it is difficult to derive a matching converse result for the case where the number of matches is unknown. The reason is as follows. Since the generating distributions are unknown, ideally we would derive a converse result by assuming universal performance guarantee for some error probabilities and lower bound the remaining error probabilities. However, as gleaned in the expressions of exponents in Theorem \ref{ld:unknown}, all three exponents involve generating distributions, which contradicts the usual optimality criterion adopted in hypothesis testing problems with unknown generating distributions and renders it challenging for us to derive a tight result. It is thus left as future work to check the optimality of the results in Theorem \ref{ld:unknown}.

There are two tradeoffs among exponents of the three kinds of error probabilities governed by thresholds $\lambda_1$ and $\lambda_2$, respectively. Specifically, $\lambda_1$ tradeoffs the mismatch exponent, the false reject exponent rate and the false alarm exponent. It follows from Claim (i) of Lemma \ref{prop:flambda1} that if $\lambda_1$ increases, both $f_{l,K}(\lambda_1,P^{M_1},Q^{M_2})$ and $f_0(\lambda_1,P^{M_1},Q^{M_2})$ are non increasing in $\lambda_1$. Thus, the false alarm exponent is a decreasing function of $\lambda_1$. Both the mismatch and the false reject exponents depend on $\lambda_1$ via $\min\{\lambda_1,f_{l,K}(\lambda_1,P^{M_1},Q^{M_2})\}$, which bounds the exponential decay rate of the probability that the estimated number of matches is wrong. On the other hand, $\lambda_2$ tradeoffs the mismatch exponent rate and the false reject exponent rate. If $\lambda_2$ increases, the mismatch exponent rate is non-decreasing while the false reject exponent rate is non-increasing.

We next discuss the extreme values of each exponent. It follows from Claim (iii) of Lemma \ref{prop:flambda1} that the maximal false alarm exponent rate equals the minimal pairwise R\'enyi divergence of order $\frac{\alpha}{1+\alpha}$ between generating distributions of the two databases, a rate that is achieved when $\lambda_1=0$. Furthermore, combining Lemmas \ref{prop:Fl} and \ref{prop:flambda1}, we conclude that the maximal mismatch exponent rate under hypothesis $\rmH_l^K$ equals $\min\{f_{l,K}(\lambda_1^*,P^{M_1},Q^{M_2}),\Lambda_l(P^{M_1},Q^{M_2},K,\alpha)\}$, where $\Lambda_l(\cdot)$ was defined in \eqref{def:Lambdal} and $\lambda_1^*$ is the solution to $\lambda_1=f_{l,K}(\lambda_1,P^{M_1},Q^{M_2})$. The maximal false reject exponent rate equals to $\min\{f_{l,K}(\lambda_1^*,P^{M_1},Q^{M_2}),\Upsilon_l(P^{M_1},Q^{M_2},K,\alpha)\}$
where $\Upsilon_l(P^{M_1},Q^{M_2},K,\alpha)$ was defined in \eqref{def:Upsilonl}. However, the maximal values of all three exponents can not be achieved simultaneously.

Finally, when specialized to statistical classification among multiple hypotheses with $M_2=K=1$, our results strengthen the achievability part of \cite[Theorem 3]{gutman1989asymptotically} by allowing the null hypothesis to be true, which means that the testing sequence can be generated from a distribution different from the generating distributions of all training sequences.

\subsection{Small Deviations}
In this section, we consider the case where the sample size is finite and derive bounds on the three kinds of error probabilities. 

Recall that $T=\sum_{K\in[M_2]}T_K=\sum_{K\in[M_2]}{M_1\choose K}{M_2\choose K}K!$. For simplicity, let
\begin{align}
\xi(M_2,N,n):=TM_2(N+1)^{M_2|\calX|}(n+1)^{M_2|\calX|}.
\end{align}

Analogous to Theorem \ref{sr:known}, the following theorem presents the main result for the small deviations regime when the number of matches is unknown. Recall that $\lambda_{1,n}$ was defined in \eqref{def:lambdain}.
\begin{theorem}
\label{sr:unknown}
Given any positive real numbers $(\lambda_1,\lambda_2)\in\bbR_+^2$, the test in Algorithm \ref{modifedtest} ensures
\begin{enumerate}
\item for each $K\in[M_2]$, $l\in[T_K]$ and any tuple of generating distributions $(P^{M_1},Q^{M_2})\in\calP_l^K$, the misclassification probability satisfies
\begin{align}
\beta(\phi_{n,N}^\rmM|P^{M_1},Q^{M_2})&\leq \xi(M_2,N,n)\exp(-nf_{l,K}(\lambda_{1,n},P^{M_1},Q^{M_1}))+\exp(-n\lambda_1)+\exp(-n\lambda_2),
\end{align}
and if $\lambda_2\leq \chi_l^*(n,\varepsilon|P^{M_1},Q^{M_2},K,\alpha)+O(\log n/n)$, the false reject probability satisfies
\begin{align}
\zeta(\phi_{n,N}^\rmM|P^{M_1},Q^{M_2})&\leq\varepsilon,
\end{align}
\item for any tuple of generating distributions $(P^{M_1},Q^{M_2})\in\calP_0$, the false alarm probability satisfies
\begin{align}
\eta(\phi_{n,N}^\rmM|P^{M_1},Q^{M_2})&\leq \xi(M_2,N,n)\exp(-nf_0(\lambda_{1,n},P^{M_1},Q^{M_1})).
\end{align}
\end{enumerate}
\end{theorem}
The proof of Theorem \ref{sr:unknown} follows by combining the proofs of Theorems \ref{sr:known} and \ref{ld:unknown} and the details are discussed in Section \ref{proof:sr:unknown}. Similar to Theorem \ref{ld:unknown}, $\lambda_1$ tradeoffs all three kinds of error probabilities while $\lambda_2$ tradeoffs the mismatch and the false reject probabilities, where the latter tradeoff is exactly the same as what $\lambda$ does in Theorem \ref{sr:known} when the number of matches $K$ is known.

\section{Proofs for Known Number of Matches}
\label{sec:proofs:known}
\subsection{Proof of Large Deviations (Theorem \ref{ld:known})}
\label{proof:ld:known}
\subsubsection{Achievability}
The performance guarantee for mismatch probabilities are the same as in \cite[Appendix G]{unnikrishnan2015asymptotically}. Here for completeness and for ease of readers, we present a detailed proof here, especially considering that the notations of both papers are slightly different. Recall Unnikrishnan's test $\phi_{n,N}^{\rmU,K}$ in \eqref{test:unn}. For each $l\in[T_K]$, under any tuples of generating distributions $(\tilP^{M_1},\tilQ^{M_2})\in\calP_l^K$, the mismatch probability satisfies
\begin{align}
\beta(\phi_{n,N}^{\rmU,K}|\tilP^{M_1},\tilQ^{M_2})
&=\Pr\big\{\phi_{n,N}^{\rmU,K}(\bX^N,\bY^n,\lambda)\notin\{\rmH_l^K,\rmH_\rmr\}\big\}\\
&=\Pr\big\{l_K^*(\bX^N,\bY^n)\neq l,~\mathrm{and~}h_K(\bX^N,\bY^n)>\lambda_n\big\}\label{defunntest}\\
&\leq \Pr\big\{\rmS_l^K(\bX^N,\bY^n)>\lambda_n\big\}\label{easyreason}\\
&=\Pr\big\{\rmG_l^K(\hatT_{\bX^N},\hatT_{\bY^n},\alpha)>\lambda_n\big\}\label{useGl}\\
&=\sum_{\substack{\{(x_i^N,y_{\sigma_l^K(i)}^n)\}_{i\in\calA_l^K}:\\\sum_{j\in\calA_l^K}\mathrm{GJS}(\hatT_{x_j^N},\hatT_{y_{\sigma_l^K(j)}^n},\alpha)>\lambda_n}}\prod_{i\in\calA_l^K}\tilP_i^N(x_i^N)\tilP_i^n(y_{\sigma_l^K(i)}^n)\label{usedefpl}\\
&=\sum_{\substack{ (\Omega^K,\Phi^K)\in(\calP^N(\calX))^{K}\times (\calP^n(\calX))^{K}:\\\sum_{i\in[K]}\mathrm{GJS}(\Omega_i,\Psi_i,\alpha)>\lambda_n}}
\prod_{i\in[K]}\tilP_i^N(\calT_{\Omega_i}^N)\tilP_i^n(\calT_{\Psi_i}^n)\label{usemoftypes}\\
&\leq \sum_{\substack{ (\Omega^K,\Phi^K)\in(\calP^N(\calX))^{K}\times (\calP^n(\calX))^{K}:\\\sum_{i\in[K]}\mathrm{GJS}(\Omega_i,\Psi_i,\alpha)>\lambda_n}}
\exp\bigg(-n\sum_{i\in[K]}\Big(\alpha D(\Omega_i\|\tilP_i)+D(\Psi_i\|\tilP_i)\Big)\bigg)\label{typeprob:upp}\\
&\leq \sum_{\substack{ (\Omega^K,\Phi^K)\in(\calP^N(\calX))^{K}\times (\calP^n(\calX))^{K}:\\\sum_{i\in[K]}\mathrm{GJS}(\Omega_i,\Psi_i,\alpha)>\lambda_n}}\exp(-n\lambda_n)
\exp\bigg(-n(1+\alpha)\sum_{i\in[K]}D\bigg(\frac{\alpha\Omega_i+\Psi_i}{1+\alpha}\bigg\|\tilP_i\bigg)\bigg)\label{useineqofme}\\
&\leq \sum_{(\Omega^K)\in(\calP^{N+n}(\calX))^K}\exp(-n\lambda_n)\exp\big(-n(1+\alpha)\sum_{i\in[K]}D(\Omega_i\|\tilP_i)\big)\label{usecombinedtypes}\\
&\leq \sum_{(\Omega^K)\in(\calP^{N+n}(\calX))^K}\exp(-n\lambda)\prod_{i\in[K]}\tilP_i^{N+n}(\calT_{\Omega_i}^{N+n})\label{usetypeclasslowp}\\
&\leq \exp(-n\lambda)\label{achldfinal},
\end{align}
where \eqref{defunntest} follows from the definition of $\phi_{n,N}^{\rmU,K}$ in \eqref{test:unn}, \eqref{easyreason} follows since $\rmS_l^K(\bX^N,\bY^n)\geq h_K(\bX^N,\bY^n)$ when $l_K^*(\bX^N,\bY^n)\neq l$, \eqref{useGl} follows from  the definition of $S_l(\cdot)$ in \eqref{def:Sl}, \eqref{usedefpl} follows from the definitions of $\Pr$ and $\rmS_l^K(\cdot)$ in \eqref{def:Sl}, \eqref{usemoftypes} follows from the method of types~\cite{csiszar1998mt} (see also
\cite[Chap. 2]{csiszar2011information}), \eqref{typeprob:upp} follows from the upper bound on the probability of a type class~\cite[Lemma 2.6]{csiszar2011information} and the fact that $N=n\alpha$, \eqref{useineqofme} follows from \cite[Eq. (5.30)]{zhou2018binary} that implies $\alpha D(\Omega_i\|\tilP_i)+D(\Psi_i\|\tilP_i)=\mathrm{GJS}(\Omega_i,\Psi_i,\alpha)+(1+\alpha)D\bigg(\frac{\alpha\Omega_i+\Psi_i}{1+\alpha}\bigg\|\tilP_i\bigg)$, \eqref{usecombinedtypes} follows by considering the average of two types $\frac{\alpha \Omega_i+\Phi_i}{1+\alpha}$ as a type of length-$(N+n)$ and by removing the constraints on the types, \eqref{usetypeclasslowp} follows from the lower bound on the probability of a type class~\cite[Lemma 2.6]{csiszar2011information} and the definition of $\lambda_n$ in \eqref{def:lambdan} and \eqref{achldfinal} follows since the sum of the probabilities of all type classes equal to one and thus
\begin{align}
\sum_{(\Omega^K)\in(\calP^{N+n}(\calX))^K}\prod_{i\in[K]}\tilP_i^{N+n}(\calT_{\Omega_i}^{N+n})
&=\sum_{i\in[K]}\Big(\sum_{\Omega_i\in\calP^{N+n}(\calX)}\tilP_i^{N+n}(\calT_{\Omega_i}^{N+n})\Big)=1.
\end{align}
We next explain in detail why \eqref{usetypeclasslowp} holds. Noting that $N=n\alpha$ and using the lower bound on the probability of a type class~\cite[Lemma 2.6]{csiszar2011information}, we have 
\begin{align}
\prod_{i\in[K]}\tilP_i^{N+n}(\calT_{\Omega_i}^{N+n})
&\geq \prod_{i\in[K]}(N+n+1)^{-|\calX|}\exp\Big(-(N+n)D(\Omega_i\|\tilP_i)\Big)\\
&=(n(1+\alpha)+1)^{-K|\calX|}\exp\Big(-n(1+\alpha)\sum_{i\in[K]}D(\Omega_i\|\tilP_i)\Big).
\end{align}
Recall the definition of $\lambda_n$ in \eqref{def:lambdan}. It follows that
\begin{align}
\nn&\exp(-n\lambda)\prod_{i\in[K]}\tilP_i^{N+n}(\calT_{\Omega_i}^{N+n})\\*
&\geq \exp(-n\lambda )(N+n+1)^{-K|\calX|}\exp\Big(-n(1+\alpha)\sum_{i\in[K]}D(\Omega_i\|\tilP_i)\Big)\\
&=\exp(-n\lambda_n)\exp\Big(-n(1+\alpha)\sum_{i\in[K]}D(\Omega_i\|\tilP_i)\Big).
\end{align}

For any $l\in[T_K]$ and any particular tuple of generating distributions $(P^{M_1},Q^{M_2})\in\calP_l^K$, the false reject probability satisfies
\begin{align}
\zeta(\phi_{n,N}^{\rmU,K}|P^{M_1},Q^{M_2})
&=\Pr\big\{\phi_{n,N}^{\rmU,K}(\bX^N,\bY^n,\lambda)=\rmH_\rmr\big\}\\
&=\Pr\big\{h_K(\bX^N,\bY^n)\leq\lambda_n\big\}\\
&=\Pr\Big\{\exists~(t,j)\in[T_K]^2:~t\neq j,~\rmS_t^K(\bX^N,\bY^n)\leq \lambda_n\mathrm{~and~}\rmS_j^K(\bX^N,\bY^n)\leq \lambda_n\Big\}\\
&\leq \sum_{\substack{(t,j)\in[T_K]^2:\\t\neq j}}\Pr\Big\{\rmS_t^K(\bX^N,\bY^n)\leq \lambda_n\mathrm{~and~}\rmS_j^K(\bX^N,\bY^n)\leq \lambda_n\Big\}\\
&\leq \frac{T_K(T_K-1)}{2}\max_{\substack{(t,j)\in[T_K]^2:\\t\neq j}}\Pr\Big\{\rmS_t^K(\bX^N,\bY^n)\leq \lambda_n\mathrm{~and~}\rmS_j^K(\bX^N,\bY^n)\leq \lambda_n\Big\}\label{frej:ld:step1}.
\end{align}
Each probability term in \eqref{frej:ld:step1} can be further upper bounded as follows:
\begin{align}
\nn&\bbP\Big\{\rmS_t^K(\bX^N,\bY^n)\leq \lambda_n\mathrm{~and~}\rmS_j^K(\bX^N,\bY^n)\leq \lambda_n\Big\}\\
&=\sum_{\substack{(\bx^N,\by^n):\\\rmS_t^K(\bX^N,\by^n)\leq \lambda_n\\\rmS_j^K(\bX^N,\by^n)\leq \lambda_n}}\Big(\prod_{i\in\calA_l^K}P_i^N(x_i^N)P_i^n(y_{\sigma_l^K(i)}^n)\Big)\Big(\prod_{i\notin\calA_l^K}P_i^N(x_i^N)\Big)\Big(\prod_{j\notin\calB_l^K}Q_j^n(y_j^n)\Big)\\
&=\sum_{\substack{(\Omega^{M_1},\Psi^{M_2})\in(\calP^N(\calX))^{M_1}\times (\calP^n(\calX))^{M_2}:\\\rmG_t^K(\Omega^M,\Psi^K,\alpha)\leq \lambda_n\\\rmG_j^K(\Omega^M,\Psi^K,\alpha)\leq \lambda_n}}\sum_{\substack{(\bx^N,\by^n):\\\forall(i,j)\in[M_1]\times[M_2]\\
x_i^N\in\calT_{\Omega_i}^N,y_j^n\in\calT_{\Psi_j}^n}}\Big(\prod_{i\in\calA_l^K}P_i^N(x_i^N)P_i^n(y_{\sigma_l^K(i)}^n)\Big)\Big(\prod_{i\notin\calA_l^K}P_i^N(x_i^N)\Big)\Big(\prod_{j\notin\calB_l^K}Q_j^n(y_j^n)\Big)\\
&\leq\sum_{\substack{(\Omega^{M_1},\Psi^{M_2})\in(\calP^N(\calX))^{M_1}\times (\calP^n(\calX))^{M_2}:\\\rmG_t^K(\Omega^M,\Psi^K,\alpha)\leq \lambda_n\\\rmG_j^K(\Omega^M,\Psi^K,\alpha)\leq \lambda_n}}
\exp\Big(-nE_l(P^{M_1},Q^{M_2},\Omega^{M_1},\Psi^{M_2},\alpha)\Big)\label{useEl}\\
&\leq \sum_{\substack{(\Omega^{M_1},\Psi^{M_2})\in(\calP^N(\calX))^{M_1}\times (\calP^n(\calX))^{M_2}}}\exp\bigg(-n\min_{\substack{(\Omega^{M_1},\Psi^{M_2})\in(\calP^N(\calX))^{M_1}\times (\calP^n(\calX))^{M_2}:\\\rmG_t^K(\Omega^M,\Psi^K,\alpha)\leq \lambda_n\\\rmG_j^K(\Omega^M,\Psi^K,\alpha)\leq \lambda_n}} E_l(P^{M_1},Q^{M_2},\Omega^{M_1},\Psi^{M_2},\alpha)\bigg)\\
&\leq (n+1)^{(M+N)|\calX|}\exp\bigg(-n\min_{\substack{(\Omega^{M_1},\Psi^{M_2})\in(\calP(\calX))^{M_1}\times (\calP(\calX))^{M_2}:\\\rmG_t^K(\Omega^M,\Psi^K,\alpha)\leq \lambda_n\\\rmG_j^K(\Omega^M,\Psi^K,\alpha)\leq \lambda_n
}} E_l(P^{M_1},Q^{M_2},\Omega^{M_1},\Psi^{M_2},\alpha)\bigg)\label{relaxations},
\end{align}
where \eqref{useEl} follows from the upper bound of the probability of the type class~\cite[Lemma 2.6]{csiszar2011information}, the definition of $E_l(\cdot)$ in \eqref{def:El} and the fact that $N=n\alpha$, and \eqref{relaxations} follows from the upper bound on the number of types and enlarging the minimization region by changing from types to distributions. We next explain \eqref{useEl} in detail. Specifically, for any tuple of types $(\Omega^{M_1},\Psi^{M_2})$,
\begin{align}
\nn&\sum_{\substack{(\bx^N,\by^n):\\\forall(i,j)\in[M_1]\times[M_2]\\
x_i^N\in\calT_{\Omega_i}^N,y_j^n\in\calT_{\Psi_j}^n}}\Big(\prod_{i\in\calA_l^K}P_i^N(x_i^N)P_i^n(y_{\sigma_l^K(i)}^n)\Big)\Big(\prod_{i\notin\calA_l^K}P_i^N(x_i^N)\Big)\Big(\prod_{j\notin\calB_l^K}Q_j^n(y_j^n)\Big)\\
&=\Big(\prod_{i\in\calA_k}P_i^N(\calT_{\Omega_i}^N)P_i^n(\Psi_{\sigma_l^K(i)}^n)
\Big)\Big(\prod_{i\notin\calA_l^K}P_i^N(\calT_{\Omega_i}^N)\Big)\Big(\prod_{j\notin\calB_l^K}Q_j^n(\Psi_j^n)\Big)\\
&\leq \exp\bigg(-n\Big(\sum_{i\in\calA_l^K}(D(\Omega_i\|P_i)+D(\Psi_{\sigma_l^K(i)}\|P_i))+\sum_{i\notin\calA_l^K}D(\Omega_i\|P_i)+\sum_{j\notin\calB_l^K}D(\Psi_j\|Q_j)\Big)\bigg)\label{el:probtp}\\
&=\exp\bigg(-n\Big(\sum_{i\in[M_1]}D(\Omega_i\|P_i)+\sum_{j\in\calB_l^K}D(\Psi_j\|P_{(\sigma_l^K)^{-1}(j)})+\sum_{j\notin\calB_l^K}D(\Psi_j\|Q_j)\Big)\bigg)\label{el:reorg}\\
&=\exp(-nE_l(P^{M_1},Q^{M_2},\Omega^{M_1},\Psi^{M_2},\alpha))\label{el:useel},
\end{align}
where \eqref{el:probtp} follows from the upper bound for the probability of a type class~\cite[Lemma 2.6]{csiszar2011information}, \eqref{el:reorg} follows by re-organizing the exponent terms in \eqref{el:probtp} and using the relationship between $\calA_l^K,\calB_l^K$ and $\sigma_l^K$, and \eqref{el:useel} follows from the definition of $E_l(\cdot)$ in \eqref{def:El}.

Note that for any finite integers $(M_1,M_2,K)$, $\liminf_{n\to\infty}\frac{\log T_K}{n}=0$. Combining \eqref{frej:ld:step1} and \eqref{relaxations} and using the definitions of $\lambda_n$ in \eqref{def:lambdan} and $F_l(\cdot)$ in \eqref{def:Fl}, for any $(P^{M_1},Q^{M_2})\in\calP_l^K$, we have
\begin{align}
\liminf_{n\to\infty}-\frac{1}{n}\log\zeta(\phi_{n,N}^{\rmU,K}|P^{M_1},Q^{M_2})
&\geq F_l(P^{M_1},Q^{M_2},\alpha,\lambda,K)\label{ld:laststep}.
\end{align}
The achievability proof of Theorem \ref{ld:known} is completed by combining \eqref{def:e*:ld:2}, \eqref{achldfinal} and \eqref{ld:laststep}.

\subsubsection{Converse}
The following lemma strengthens the converse part of \cite[Theorem 4.1]{unnikrishnan2015asymptotically} by providing a non-asymptotic converse bound on the false reject probability in the generalized Neyman-Pearson sense.

Given each $(n,N)\in\bbN^2$, let
\begin{align}
\delta_{n,N}
&:=\frac{M_1|\calX|\log (N+1)}{N}+\frac{M_2|\calX|\log(n+1)}{n}\label{def:deltann}.
\end{align}
Given any $\lambda\in\bbR_+$, define
\begin{align}
\tlambda_{n,N}:=\lambda-\delta_{n,N}-\log n/n\label{def:tlambda}.
\end{align}
\begin{lemma}
\label{converse:gnp}
Consider any test $\phi_{n,N}$ such that
\begin{align}
\max_{l\in[T_K]}\sup_{(\tilP^{M_1},\tilQ^{M_2})\in\calP_l^K}\beta(\phi_{n,N}|\tilP^{M_1},\tilQ^{M_2})\leq \exp(-n\lambda)\label{nonasymp:constraint}.
\end{align}
Then, for each $l\in[T_K]$, under any tuple of generating distributions $(P^{M_1},Q^{M_2})\in\calP_l^K$, the false reject probability satisfies
\begin{align}
\zeta(\phi_{n,N}|P^{M_1},Q^{M_2})
&\geq \bigg(1-\frac{T_K}{n}\bigg)\Pr\{h_K(\bX^N,\bY^n)\leq\tlambda_{n,N}\}\label{lb:frej}.
\end{align}
\end{lemma}
The proof of Lemma \ref{converse:gnp} is available in Appendix \ref{proof:con:gnp} and is inspired by~\cite{gutman1989asymptotically}. Specifically, we first relate the mismatch and false reject probabilities of any test with those of a test that uses only types of sequences of two databases. Subsequently, we show that for such a type-based test if $h_K(\bX^N,\bY^n)\leq\tlambda_{n,N}$, a decision of reject must be output, otherwise the constraint in \eqref{nonasymp:constraint} will be violated. Finally, the above two claims are combined together to yield the desired result in Lemma \ref{converse:gnp}.

Lemma \ref{converse:gnp} is known as the optimality criterion under the generalized Neyman-Pearson criterion. The constraint in \eqref{nonasymp:constraint} ensures that the test $\phi_{n,N}$ is universal for mismatch probabilities because under any hypothesis, for any unknown tuple of generating distributions, the mismatch probability decays exponentially fast with respect to the sample size $n$. Such a constraint dates back to Gutman in his study of statistical classification~\cite{gutman1989asymptotically} and generalizes the traditional Neyman-Pearson criterion for hypothesis testing where the generating distribution under each hypothesis is known. Under the universal constraint in \eqref{nonasymp:constraint}, \eqref{lb:frej} provides a lower bound to the false reject probability under each hypothesis for any generating distributions $(P^{M_1},Q^{M_2})$, under which we would like to evaluate the performance of the test.

It follows from Lemma \ref{converse:gnp} that for each $l\in[T_K]$, under any tuple of generating distributions $(P^{M_1},Q^{M_2})\in\calP_l^K$,
\begin{align}
\zeta(\phi_{n,N}|P^{M_1},Q^{M_2})
&\geq \Pr\big\{h_K(\bX^N,\bY^n)\leq\tlambda_{n,N}\big\}\\
&=\bigg(1-\frac{T_K}{n}\bigg)\Pr\Big\{\exists(t,j)\in[T_K]^2:~t\neq j,~\rmS_t^K(\bX^N,\bY^n)\leq \tlambda_{n,N}\mathrm{~and~}\rmS_j^K(\bX^N,\bY^n)\leq \tlambda_{n,N}\Big\}\\
&\geq \bigg(1-\frac{T_K}{n}\bigg)\max_{\substack{(t,j)\in[T_K]^2:\\t\neq j}}\Pr\Big\{\rmS_t^K(\bX^N,\bY^n)\leq \tlambda_{n,N}\mathrm{~and~}\rmS_j^K(\bX^N,\bY^n)\leq \tlambda_{n,N}\Big\}\label{converse:ld:00}.
\end{align}
Similar to \eqref{relaxations}, we have
\begin{align}
\nn&\Pr\Big\{\rmS_t^K(\bX^N,\bY^n)\leq \tlambda_{n,N}\mathrm{~and~}\rmS_j^K(\bX^N,\bY^n)\leq \tlambda_{n,N}\Big\}\\
&\geq (n+1)^{-(M+N)|\calX|}\sum_{\substack{(\Omega^{M_1},\Psi^{M_2})\in(\calP^N(\calX))^{M_1}\times (\calP^n(\calX))^{M_2}:\\\rmG_t^K(\Omega^M,\Psi^K,\alpha)\leq \tlambda_{n,N}\\\rmG_j^K(\Omega^M,\Psi^K,\alpha)\leq \tlambda_{n,N}}}
\exp\Big(-nE_l(P^{M_1},Q^{M_2},\Omega^{M_1},\Psi^{M_2},\alpha)\Big)\\
&\geq (n+1)^{-(M+N)|\calX|}\exp\bigg(-n\min_{\substack{(\Omega^{M_1},\Psi^{M_2})\in(\calP^N(\calX))^{M_1}\times (\calP^n(\calX))^{M_2}:\\\rmG_t^K(\Omega^M,\Psi^K,\alpha)\leq \tlambda_{n,N}\\\rmG_j^K(\Omega^M,\Psi^K,\alpha)\leq \lambda_n}} E_l(P^{M_1},Q^{M_2},\Omega^{M_1},\Psi^{M_2},\alpha)\bigg)\label{converse:ld:01}.
\end{align}
The converse proof of Theorem \ref{ld:known} is completed by combining \eqref{def:e*:ld:2}, \eqref{converse:ld:00}, \eqref{converse:ld:01} and using the continuity of $E_l(P^{M_1},Q^{M_2},\Omega^{M_1},\Psi^{M_2})$ in $(\Omega^{M_1},\Psi^{M_2})$, the continuity of $F_l(P^{M_1},Q^{M_2},\alpha,\lambda,K)$ in $\lambda$ and the fact that $\lim_{n\to\infty}\tlambda_{n,N}=\lambda$ and $\lim_{n\to\infty}\frac{T_K}{n}=0$.

\subsection{Proof of Small Deviations (Theorem \ref{sr:known})}
\label{proof:sr:known}

\subsubsection{Achievability}
The achievability proof analyzes Unnikrishnan's test in \eqref{test:unn}. The analysis of the mismatch probability $\beta(\cdot)$ is exactly the same as that for large deviations in Section \ref{proof:ld:known} and is thus omitted. For each $l\in[T_K]$ and any tuple of generating distributions $(P^{M_1},Q^{M_2})\in\calP_l^K$, we next bound the false reject probability as follows:
\begin{align}
\zeta(\phi_{n,N}^{\rmU,K}|P^{M_1},Q^{M_2})
&=\Pr\Big\{h_K(\bX^N,\bY^n)\leq \lambda_n\Big\}\\
&\leq \Pr\Big\{\min_{t\in([T_K]\setminus\{l\})}\rmS_t^K(\bX^N,\bY^n)\leq \lambda_n\}\label{usedefh}\\
&=1-\Pr\Big\{\forall~t\in([T_K]\setminus\{l\}),~\rmG_t^K(\hatT_{\bX^N},\hatT_{\bY^n},\alpha)>\lambda_n\Big\}\label{usedefst},
\end{align}
where \eqref{usedefh} follows from the definition of $h_K(\cdot)$ in \eqref{def:h4unn} that implies $h_K(\bX^N,\bY^n)\geq \min_{t\in([T_K]\setminus\{l\})}\rmS_t^K(\bX^N,\bY^n)$. Specifically, for any $(\bX^n,\by^n)$, if $l_K^*(\bx^N,\by^n)=l$, then $h_K(\bx^N,\by^n)=\min_{t\in([T_K]\setminus\{l\})}\rmS_t^K(\bx^N,\by^n)$; otherwise, $h_K(\bx^N,\by^n)\geq \min_{t\in([T_K]\setminus\{l\})}\rmS_t^K(\bx^N,\by^n)=\min_{t\in[T_K]}\rmS_t^K(\bx^N,\by^n)$, and \eqref{usedefst} follows from the definition of the scoring function $S_t(\cdot)$ in \eqref{def:Sl}.

For subsequent analyses, we need the following definitions. Define the following typical set
\begin{align}
\calT(P^{M_1},Q^{M_2})
&:=\Big\{(\bx^N,\by^n):~\max_{i\in[M_1]}\|\hatT_{x_i^N}-P_i\|_{\infty}\leq \sqrt{\log n/n},~\max_{j\in[M_2]}\|\hatT_{y_j^n}-Q_j\|_{\infty}\leq \sqrt{\log n/n}\Big\}\label{def:typical}.
\end{align}
It follows from \cite[Lemma 24]{tan2014state} that for each $l\in[T_K]$,
\begin{align}
\Pr\Big\{(\bX^N,\bY^n)\notin\calT(P^{M_1},Q^{M_2})\Big\}\leq \frac{2M_1|\calX|}{N^2}+\frac{2M_2|\calX|}{n^2}=:\theta_n\label{def:thetan}.
\end{align}

Recall the definitions of the information densities $\imath_1(\cdot)$ and $\imath_2(\cdot)$ in \eqref{def:i1} and \eqref{def:i2}, respectively. Fix $t\in([T_K]\setminus\{l\})$. For any $(\bx^N,\by^n)\in\calT(P^{M_1},Q^{M_2})$, under hypothesis $\rmH_l^K$, the second-order Taylor expansion of $\rmG_t^K(\hatT_{x^N},\hatT_{y^n},\alpha)$ around $(P^{M_1},Q^{M_2})$ implies that there exists a constant $c_1$ such that
\begin{align}
\rmG_t^K(\hatT_{\bx^N},\hatT_{\by^n},\alpha)
&=\sum_{(i,j)\in\calM_t^K}\mathrm{GJS}(\hatT_{x_i^N},\hatT_{y_j^n},\alpha)\\
\nn&=\sum_{(i,j)\in\calM_t^K}\bigg(\mathrm{GJS}(P_i,Q_j,\alpha)+\sum_{x\in\calX}(\hatT_{x_i^N}(x)-P_i(x))\alpha\imath_1(x|P_i,Q_j,\alpha)\\*
&\qquad\qquad\qquad+\sum_{y\in\calX}(\hatT_{y_j^n}(y)-Q_j(y))\imath_2(y|P_i,Q_j,\alpha)\bigg)+\frac{c_1\log n}{n}\label{taylor1:1}\\
&=\frac{1}{n}\sum_{(i,j)\in\calM_t^K}\bigg(\sum_{s\in[N]}\imath_1(x_{i,s}|P_i,Q_j,\alpha)
+\sum_{s\in[n]}\imath_2(y_{j,s}|P_i,Q_j,\alpha)\bigg)+\frac{c_1\log n}{n}\label{taylor1},
\end{align}
where \eqref{taylor1:1} follows since $N=n\alpha$, the alphabet $\calX$ is finite and $|\calM_t^K|=K$ is finite so that the remainder term equals to $\frac{c_1\log n}{n}$ as a function of $\|\hatT_{x_i^N}-P_i\|_2^2$ and $\|\hatT_{y_j^n}-Q_j\|_2^2$ for $(i,j)\in\calM_t^K$ for some constant $c_1$, and \eqref{taylor1} follows since
\begin{align}
\mathrm{GJS}(P_i,Q_j,\alpha)=\sum_{x\in\calX}P_i(x)\alpha\imath_1(x|P_i,Q_j,\alpha)+\sum_{y\in\calX}Q_j(y)\imath_2(y|P_i,Q_j,\alpha).
\end{align}
Analogously, for any $(\bx^N,\by^n)\in\calT(P^{M_1},Q^{M_2})$, under hypothesis $\rmH_l^K$, if $t=l$, the first-order Taylor expansion of $\rmG_t^K(\hatT_{x^N},\hatT_{y^n},\alpha)$ around $(P^{M_1},Q^{M_2})$  satisfies
\begin{align}
\rmG_t^K(\hatT_{\bx^N},\hatT_{\by^n},\alpha)
&=O(\sqrt{\log n/n})\label{taylor2}.
\end{align}
 
Using \eqref{def:thetan} and \eqref{taylor1}, we have
\begin{align}
\nn&\Pr\Big\{\forall~t\in([T_K]\setminus\{l\}),~\rmG_t^K(\hatT_{\bX^N},\hatT_{\bY^n},\alpha)>\lambda_n\Big\}\\*
&\geq \Pr\Big\{\forall~t\in([T_K]\setminus\{l\}),~\rmG_t^K(\hatT_{\bX^N},\hatT_{\bY^n},\alpha)>\lambda_n,~(\bX^N,\bY^n)\in\calT(P^{M_1},Q^{M_2})\Big\}\\
\nn&=\Pr\Bigg\{\forall~t\in([T_K]\setminus\{l\}),~\frac{1}{n}\sum_{(i,j)\in\calM_t^K}\bigg(\sum_{s\in[N]}\imath_1(X_{i,s}|P_i,Q_j,\alpha)+\sum_{s\in[n]}\imath_2(Y_{j,s}|P_i,Q_j,\alpha)\bigg)>\lambda_n+\frac{c_1\log n}{n},\\*
&\qquad\qquad\qquad\qquad\qquad (\bX^N,\bY^n)\in\calT(P^{M_1},Q^{M_2})\Bigg\}\\
\nn&\geq\Pr\Bigg\{\forall~t\in([T_K]\setminus\{l\}),~\frac{1}{n}\sum_{(i,j)\in\calM_t^K}\bigg(\sum_{s\in[N]}\imath_1(X_{i,s}|P_i,Q_j,\alpha)+\sum_{s\in[n]}\imath_2(Y_{j,s}|P_i,Q_j,\alpha)\bigg)>\lambda_n+\frac{c_1\log n}{n}\Bigg\}\\*
&\qquad-\Pr\Big\{\bX^N,\bY^n)\notin\calT(P^{M_1},Q^{M_2})\Big\}\label{ine:lower:Pab}\\
&=\Pr\Bigg\{\forall~t\in([T_K]\setminus\{l\}),~\frac{1}{n}\sum_{(i,j)\in\calM_t^K}\bigg(\sum_{s\in[N]}\imath_1(X_{i,s}|P_i,Q_j,\alpha)+\sum_{s\in[n]}\imath_2(Y_{j,s}|P_i,Q_j,\alpha)\bigg)>\lambda_n+\frac{c_1\log n}{n}\Bigg\}-\theta_n\label{sd:proof:achstep2}.
\end{align}
where \eqref{ine:lower:Pab} follows since for any two events $(\calA,\calB)$, we have that $\Pr\{\calA\cap\calB\}\geq \Pr\{\calA\}-\Pr\{\calB^\rmc\}$.

For ease of notation of subsequent analyses, we need the following definitions. Given any $(i,j)\in[M_1]\times[M_2]$, for each $s\in[N+n]$, let
\begin{align}
Z_s^{i,j}
&:=
\left\{
\begin{array}{ll}
\imath_1(X_{i,s}|P_i,Q_j,\alpha)&\mathrm{if~}s\in[N],\\
\imath_2(Y_{j,s-N}|P_i,Q_j,\alpha)&\mathrm{otherwize}.
\end{array}
\right.
\label{def:Zsij}
\end{align}
Given any $(i,j)\in[M_1]\times[M_2]$, the random variables $Z_1^{i,j},\ldots,Z_{N+n}^{i,j}$ are independent since each $X_i^N$ and $Y_j^n$ are generated i.i.d. and independent of each other. Note that
\begin{align}
\bbE[Z_s^{i,j}]
&=\left\{
\begin{array}{ll}
\bbE_{P_i}[\imath_1(X|P_i,Q_j,\alpha)]&\mathrm{if~}s\in[N],\\
\bbE_{Q_j}[\imath_2(Y|P_i,Q_j,\alpha)]&\mathrm{otherwise}.
\end{array}
\right.
\end{align}
Fix $t\in[T_K]$ such that $t\neq l$. For each $s\in[N+n]$, let
\begin{align}
U_s^t&:=\sum_{(i,j)\in\calM_t^K}(Z_s^{i,j}-\bbE[Z_s^{i,j}])\label{def:Ust},\\
W_t&:=\frac{1}{\sqrt{n}}\sum_{s\in[N+n]}U_s^t
=\frac{1}{\sqrt{n}}\bigg(\sum_{(i,j)\in\calM_t^K}\sum_{s\in[N+n]}Z_s^{i,j}-\sum_{(i,j)\in\calM_t^K}\sum_{s\in[N+n]}\bbE[Z_s^{i,j}]\bigg)\label{def:Wt}.
\end{align}
Note that the random variables $U_1^t,\ldots,U_{N+n}^t$ are independent. It follows that
\begin{align}
\sum_{(i,j)\in\calM_t^K}\sum_{s\in[N+n]}\bbE[Z_s^{i,j}]
&=\sum_{(i,j)\in\calM_t^K}\bigg(\sum_{s\in[N]}\bbE\big[\imath_1(X_{i,s}|P_i,Q_j,\alpha)\big]+\sum_{s\in[n]}\bbE\big[\imath_2(Y_{j,s}|P_i,Q_j,\alpha)\big]\bigg)\\
&=n\sum_{(i,j)\in\calM_t^K}\Big(\alpha\bbE\big[\imath_1(X_i|P_i,Q_j,\alpha)\big]+\bbE\big[\imath_1(Y_{j,s}|P_i,Q_j,\alpha)\big]\Big)\\
&=n\rmG_t^K(P^{M_1},Q^{M_2},\alpha)\label{W:average},
\end{align}
where $X_{i,s}\sim P_i$ and $Y_{j,s}\sim Q_j$. For each $t\in[T_K]$, it follows that
\begin{align}
\bbE[W_t]=0.
\end{align}
For each pair $(t_1,t_2)\in[T_K]^2$, the covariance of $(W_{t_1},W_{t_2})$ satisfies
\begin{align}
\bbE[W_{t_1}W_{t_2}]
&=\frac{1}{n}\sum_{s\in[N+n]}\Big(\bbE[U_s^{t_1}U_s^{t_2}]\Big)\\
&=\frac{1}{n}\sum_{s\in[N+n]}\bbE\bigg[\Big(\sum_{(i,j)\in\calM_{t_1}^K}(Z_s^{i,j}-\bbE[Z_s^{i,j}])\Big)\Big(\sum_{(\bari,\barj)\in\calM_{t_2}^K}(Z_s^{\bari,\barj}-\bbE[Z_s^{\bari,\barj}])\Big)\bigg]\\
&=\frac{1}{n}\sum_{s\in[N+n]}\Bigg(\bbE\bigg[\Big(\sum_{(i,j)\in\calM_{t_1}^K}Z_s^{i,j}\Big)\Big(\sum_{(\bari,\barj)\in\calM_{t_2}^K}Z_s^{\bari,\barj}\Big)\bigg]-\Big(\sum_{(i,j)\in\calM_{t_1}^K}\bbE[Z_s^{i,j}]\Big)\Big(\sum_{(\bari,\barj)\in\calM_{t_2}^K}\bbE[Z_s^{\bari,\barj}]\Big)\Bigg)\\
&=\mathrm{Cov}_{t_1,t_2}(P^{M_1},Q^{M_2},\alpha)\label{useVt1t2},
\end{align}
where \eqref{useVt1t2} follows from the definition of $\mathrm{Cov}_{t_1,t_2}(P^{M_1},Q^{M_2},\alpha)$ in \eqref{def:vt1t2} and the algebra in Appendix \ref{just:usevt1t2}.

Combining \eqref{usedefst}, \eqref{taylor1}, \eqref{sd:proof:achstep2} and \eqref{W:average} and using the definition of $W_t$ in \eqref{def:Wt}, we have
\begin{align}
\zeta(\phi_{n,N}^{\rmU,K}|P^{M_1},Q^{M_2})
&\leq 1-\Pr\Big\{\forall~t\in([T_K]\setminus\{l\}),W_t>\sqrt{n}\big(\lambda_n-\rmG_t^K(P^{M_1},Q^{M_2},\alpha)+\frac{c_1\log n}{n}\big)\Big\}+\theta_n\label{com:ach:1}.
\end{align}

Recall the definitions of $\Lambda_l(P^{M_1},Q^{M_2},K,\alpha)$ in \eqref{def:Lambdal} and the set $\calI_l(P^{M_1},Q^{M_2})$ in \eqref{def:calIl}. Let $L\in\bbR$ be a positive real number to be specified. Choose $\lambda_n$ such that
\begin{align}
\lambda_n=\Lambda_l(P^{M_1},Q^{M_2},K,\alpha)-\frac{L}{\sqrt{n}}-\frac{c_1\log n}{n}\label{lambda:choose:sr:known}.
\end{align} 
For each $t\in[T_K]$ such that $t\neq l$ and $t\notin\calI_l(P^{M_1},Q^{M_2})$, when $n$ is  sufficiently large, there exists a positive real number $r\in\bbR_+$ such that
we have
\begin{align}
\lambda_n-\rmG_t^K(P^{M_1},Q^{M_2},\alpha)\leq -r<0.
\end{align}
Thus, for $n$ large, using Chebyshev's inequality, we have
\begin{align}
\Pr\bigg\{W_t\leq \sqrt{n}\Big(\lambda-\rmG_t^K(P^{M_1},Q^{M_2},\alpha)+\frac{c_1\log n}{n}\Big)\bigg\}
&\leq \Pr\big\{W_t\leq -\sqrt{n}r\big\}\\
&\leq \Pr\big\{|W_t|\geq \sqrt{n}r\big\}\\
&\leq \frac{n\mathrm{Cov}_{t,t}(P^{M_1},Q^{M_2},\alpha)}{n^2r^2}\\
&=O\left(\frac{1}{n}\right)\label{com:ach:3}.
\end{align}
Therefore, the probability term inside \eqref{com:ach:1} can be lower bounded as follows:
\begin{align}
\nn&\Pr\Big\{\forall~t\in([T_K]\setminus\{l\}),W_t>\sqrt{n}\big(\lambda_n-\rmG_t^K(P^{M_1},Q^{M_2},\alpha)+\frac{c_1\log n}{n}\big)\Bigg\}\\*
\nn&\geq\Pr\bigg\{\forall~t\in\calI_l(P^{M_1},Q^{M_2}),~W_t>\sqrt{n}\Big(\lambda_n-\rmG_t^K(P^{M_1},Q^{M_2},\alpha)+\frac{c_1\log n}{n}\Big)\bigg\}\\*
&\qquad-\sum_{t\in[T_K]:~t\neq l,~t\notin\calI_l(P^{M_1},Q^{M_2})}\Pr\bigg\{W_t\leq \sqrt{n}\Big(\lambda_n-\rmG_t^K(P^{M_1},Q^{M_2},\alpha)+\frac{c_1\log n}{n}\Big)\bigg\}\label{com:ach:2}\\
&=\Pr\bigg\{\forall~t\in\calI_l(P^{M_1},Q^{M_2}),~W_t>\sqrt{n}\Big(\lambda_n-\rmG_t^K(P^{M_1},Q^{M_2},\alpha)+\frac{c_1\log n}{n}\Big)\bigg\}+O\left(\frac{1}{n}\right)\label{com:ach:4}\\
&=\Pr\big\{\forall~t\in\calI_l(P^{M_1},Q^{M_2}),W_t>-L\big\}+O\left(\frac{1}{n}\right)\label{com:ach:5},
\end{align}
where \eqref{com:ach:2} follows by recursively applying the inequality $\Pr\{\calA\cap\calB\}\geq \Pr\{\calA\}-\Pr\{\calB^\rmc\}$ for any two sets $(\calA,\calB)$, \eqref{com:ach:4} follows from \eqref{com:ach:3} and \eqref{com:ach:5} follows from the the definitions of $\calI_l(P^{M_1},Q^{M_2})$ in \eqref{def:calIl} and $\lambda_n$ in \eqref{lambda:choose:sr:known}.

Recall that $N=n\alpha$. Combining \eqref{com:ach:1} and \eqref{com:ach:4} and using the definition of $\theta_n$ in \eqref{def:thetan}, when the covariance matrix $\bV^l(P^{M_1},Q^{M_2},\alpha)$ is positive definite, we obtain
\begin{align}
\zeta(\phi_{n,N}^{\rmU,K}|P^{M_1},Q^{M_2})
&=1-\Pr\big\{\forall~t\in\calI_l(P^{M_1},Q^{M_2}),W_t>-L\big\}+O\left(\frac{1}{n}\right)\\
&=1-\Pr\big\{\forall~t\in\calI_l(P^{M_1},Q^{M_2}),-W_t<L\big\}+O\left(\frac{1}{n}\right)\\
&\leq 1-\rvbPh_{\tau_l}\big(L\times \mathbf{1}_{\tau_l};\mathbf{0}_{\tau_l};\bV^l(P^{M_1},Q^{M_2},\alpha)\big)+O\left(\frac{1}{\sqrt{n}}\right)\label{multivariateuse},
\end{align}
where $\tau_l=|\calI_l(P^{M_1},Q^{M_2})|$ and \eqref{multivariateuse} follows from the multivariate Berry-Esseen theorem~\cite[Cor. 29]{watanabe2015}. The case of positive semidefinite $\bV^l(P^{M_1},Q^{M_2},\alpha)$ can be handled similarly to \cite[Cor. 8]{tan2014dispersions}.

Recall the definition of $\nu_l^*(\cdot)$ in \eqref{def:nul*}. Choose $L$ such that
\begin{align}
L=\nu_l^*\big(\varepsilon-O\left(1/\sqrt{n}\right)|P^{M_1},Q^{M_2},K,\alpha\big)\label{choose:L}.
\end{align}
It follows from \eqref{multivariateuse} that
\begin{align}
\zeta(\phi_{n,N}^{\rmU,K}|P^{M_1},Q^{M_2})\leq \varepsilon.
\end{align}
Using the definition of $\lambda^*_{\mathrm{SD}}(\cdot)$ in \eqref{def:l*:sd}, combining \eqref{lambda:choose:sr:known} and \eqref{choose:L}, and applying Taylor's expansion of $\nu_l^*(\cdot)$ around $\varepsilon$, the achievability proof of Theorem \ref{sr:known} is completed.

\subsubsection{Converse}
Recall the definitions of $\delta_{n,N}$ in \eqref{def:deltann} and $\tilde{\lambda}_{n,N}$ in \eqref{def:tlambda}. Let $\delta_{n,N}'=\delta_{n,N}+\log n/n$. Using Lemma \ref{converse:gnp}, we have that for any test satisfying the generalized Neyman-Pearson criterion, for each $l\in[T_K]$, under any tuples of generating distributions $(P^{M_1},Q^{M_2})\in\calP_l^K$, the false reject probability is lower bounded by 
\begin{align}
\zeta(\phi_{n,N}^{\rmU,K}|P^{M_1},Q^{M_2})
&\geq\bigg(1-\frac{T_K}{n}\bigg)\Pr\{h_K(\bX^N,\bY^n)\leq\tlambda_{n,N}\}\\
&=\bigg(1-\frac{T_K}{n}\bigg)\Pr\{h_K(\bX^N,\bY^n)\leq\lambda+\delta_{n,N}'\}\label{sd:proof:constep1}.
\end{align}
It suffices to bound the probability term in \eqref{sd:proof:constep1}. Note that
\begin{align}
\nn&\Pr\{h_K(\bX^N,\bY^n)\leq\lambda+\delta_{n,N}'\}\\
&=\Pr\bigg\{h_K(\bX^N,\bY^n)=\min_{t\in([T_K]\setminus\{l\})}\rmS_t^K(\bX^N,\bY^n),~h_K(\bX^N,\bY^n)\leq\lambda+\delta_{n,N}'\bigg\}\\
&\geq \Pr\bigg\{\min_{t\in([T_K]\setminus\{l\})}\rmS_t^K(\bX^N,\bY^n)\leq\lambda+\delta_{n,N}'\bigg\}-\Pr\bigg\{h_K(\bX^N,\bY^n)\neq \min_{t\in([T_K]\setminus\{l\})}\rmS_t^K(\bX^N,\bY^n)\bigg\}\label{sd:proof:constep2}.
\end{align}

Analogous to the achievability proof, we can lower bound the first term in \eqref{sd:proof:constep2} as follows:
\begin{align}
\nn&\Pr\bigg\{\min_{t\in([T_K]\setminus\{l\})}\rmS_t^K(\bX^N,\bY^n)\leq\lambda+\delta_{n,N}'\bigg\}\\*
&=1-\Pr\Big\{\forall~t\in([T_K]\setminus\{l\}),~\rmS_t^K(\bX^N,\bY^n)>\lambda+\delta_{n,N}'\Big\}\\
\nn&\geq 1-\Pr\Big\{\forall~t\in([T_K]\setminus\{l\}),~\rmS_t^K(\bX^N,\bY^n)>\lambda+\delta_{n,N}',~(\bX^N,\bY^n)\in\calT(P^{M_1},Q^{M_2})\Big\}\\*
&\qquad-\Pr\big\{(\bX^N,\bY^n)\notin\calT(P^{M_1},Q^{M_2})\big\}\\
&\ge 1-\theta_n-\Pr\Big\{\forall~t\in([T_K]\setminus\{l\}),~W_t>\sqrt{n}(\lambda-\rmG_t^K(P^{M_1},Q^{M_2},\alpha)+\delta_{n,N}')\Big\}\\
&\ge 1-\theta_n-\Pr\Big\{\forall~t\in\calI_l(P^{M_1,Q^{M_2}}),~W_t>\sqrt{n}(\lambda-\rmG_t^K(P^{M_1},Q^{M_2},\alpha)+\delta_{n,N}')\Big\}\label{sd:proof:constep3}.
\end{align}
The second term in \eqref{sd:proof:constep2} is upper bounded as follows:
\begin{align}
\nn&\Pr\bigg\{h_K(\bX^N,\bY^n)\neq \min_{t\in([T_K]\setminus\{l\})}\rmS_t^K(\bX^N,\bY^n)\bigg\}\\*
&=\Pr\bigg\{\exists~t\in([T_K]\setminus\{l\}),\rmS_t^K(\bX^N,\bY^n)<\rmS_l^K(\bX^N,\bY^n)\bigg\}\\
&\leq \Pr\bigg\{\exists~t\in([T_K]\setminus\{l\}),\rmS_t^K(\bX^N,\bY^n)<\rmS_l^K(\bX^N,\bY^n),~(\bX^N,\bY^n)\in\calT(P^{M_1},Q^{M_2})\bigg\}+\Pr\big\{(\bX^N,\bY^n)\notin\calT(P^{M_1},Q^{M_2})\big\}\\
&\leq \sum_{t\in([T_K]\setminus\{l\})}\Pr\big\{\rmS_t^K(\bX^N,\bY^n)<\rmS_l^K(\bX^N,\bY^n),~(\bX^N,\bY^n)\in\calT(P^{M_1},Q^{M_2})\big\}+\theta_n\label{usesomeresult}\\
&=\sum_{t\in([T_K]\setminus\{l\})}\Pr\big\{\rmG_t^K(\hatT_{\bX^N},\hatT_{\bY^n},\alpha)<\rmG_l^K(\hatT_{\bX^N},\hatT_{\bY^n},\alpha),~(\bX^N,\bY^n)\in\calT(P^{M_1},Q^{M_2})\big\}+\theta_n\\
&=\sum_{t\in([T_K]\setminus\{l\})}\Pr\Bigg\{\sum_{(i,j)\in\calM_t^K}\bigg(\sum_{s\in[N]}\imath_1(X_{i,s}|P_i,Q_j,\alpha)
+\sum_{s\in[n]}\imath_2(Y_{j,s}|P_i,Q_j,\alpha)\bigg)<O(\sqrt{n\log n})\Bigg\}+\theta_n\label{usetaylorandothers},
\end{align}
where \eqref{usesomeresult} follows from the result in \eqref{def:thetan}, \eqref{usetaylorandothers} follows from the Taylor expansions of $\rmG_t^K(\hatT_{\bX^N},\hatT_{\bY^n},\alpha)$ in \eqref{taylor1} and \eqref{taylor2}.

The probability term in \eqref{usetaylorandothers} can be further upper bounded using Chebyshev's inequality. For ease of analyses, we use $\imath_1(\cdot)$ to denote $\imath_1(\cdot|\cdot)$ and use $\imath_2(\cdot)$ similarly. Fix ant $t\in([T_K]\setminus\{l\})$.
Note that there are in total $|\calM_t^K|(N+n)=Kn(1+\alpha)$ independent variables in the sum term for the probability term in \eqref{usetaylorandothers}. Let the average of these $Kn(1+\alpha)$ random variables be
\begin{align}
S_n^t:=\frac{1}{Kn(1+\alpha)}\sum_{(i,j)\in\calM_t^K}\bigg(\sum_{s\in[N]}\imath_1(X_{i,s})
+\sum_{s\in[n]}\imath_2(Y_{j,s})\bigg).
\end{align}
It follows that
\begin{align}
\bbE[S_n^t]
&=\frac{1}{Kn(1+\alpha)}\sum_{(i,j)\in\calM_t^K}\Big(n(1+\alpha)\bbE_{P_i}[\imath_1(X)]+n\bbE_{Q_j}[\imath_2(Y)]\Big)\\
&=\frac{1}{K(1+\alpha)}\sum_{(i,j)\in\calM_t^K}\Big((1+\alpha)\bbE_{P_i}[\imath_1(X)]+\bbE_{Q_j}[\imath_2(Y)]\Big)\\
&=:\rho(P^{M_1},Q^{M_2},\calM_t^K,\alpha)
\end{align}
and
\begin{align}
\mathrm{Var}[S_n^t]
&=\frac{1}{(Kn(1+\alpha))^2}\sum_{(i,j)\in\calM_t^K}\bigg(\sum_{s\in[N]}\Var[\imath_1(X_{i,s})]
+\sum_{s\in[n]}\Var[\imath_2(Y_{j,s})]\bigg)\\*
&=\frac{1}{K^2(1+\alpha)^2n}\sum_{(i,j)\in\calM_t^K}\Big((1+\alpha)\Var_{P_i}[\imath_1(X)]+\Var_{Q_j}[\imath_2(Y)]\Big)\\
&=:\frac{\rmV(P^{M_1},Q^{M_2},\calM_t^K,\alpha)}{n}.
\end{align}
Thus, for each $t\in([T_K]\setminus\{l\})$, the probability term in \eqref{usetaylorandothers} can be bounded as follows:
\begin{align}
\nn&\Pr\Bigg\{\sum_{(i,j)\in\calM_t^K}\bigg(\sum_{s\in[N]}\imath_1(X_{i,s})
\sum_{s\in[n]}\imath_2(Y_{j,s})\big)<O\left(\sqrt{n\log n}\right)\Bigg\}\\
&=\Pr\Big\{S_n^t<O(\sqrt{\log n/n})\Big\}\\*
&=\Pr\Big\{S_n^t-\bbE[S_n^t]<-\rho(P^{M_1},Q^{M_2},\calM_t^K,\alpha)+O(\sqrt{\log n/n})\Big\}\\
&\leq \Pr\Big\{|S_n^t-\bbE[S_n^t]|>\rho(P^{M_1},Q^{M_2},\calM_t^K,\alpha)+O(\sqrt{\log n/n})\Big\}\\
&\leq \frac{\rmV(P^{M_1},Q^{M_2},\calM_t^K,\alpha)}{n(\rho(P^{M_1},Q^{M_2},\calM_t^K,\alpha)+O(\sqrt{\log n/n}))^2}\\
&=O(1/n)\label{sd:proof:constep4}.
\end{align}

Let $L\in\bbR_+$ be specified later and consider any $\lambda$ such that
\begin{align}
\lambda=\Lambda_l(P^{M_1},Q^{M_2},K,\alpha)-\frac{L}{\sqrt{n}}-\delta_{n,N}'\label{lambda:choose:sr:known:2}.
\end{align}
Combining \eqref{sd:proof:constep1}, \eqref{sd:proof:constep2}, \eqref{sd:proof:constep3}, \eqref{usetaylorandothers} and \eqref{sd:proof:constep4} and using the definitions of $\theta_n$ in \eqref{def:thetan} and $\calI_L(P^{M_1},Q^{M_2})$ in \eqref{def:calIl}, we have
\begin{align}
\zeta(\phi_{n,N}^{\rmU,K}|P^{M_1},Q^{M_2})
&\geq (1-T_K/n)\Big(1-\Pr\big\{\forall~t\in\calI_l(P^{M_1,Q^{M_2}}),~W_t>-L\big\}+O(1/n)\Big)\\
&=1-\Pr\big\{\forall~t\in\calI_l(P^{M_1,Q^{M_2}}),-W_t<L\big\}+O(1/n)\\
&\geq 1-\rvbPh_{\tau_l}\big(L\times \mathbf{1}_{\tau_l};\mathbf{0}_{\tau_l};\bV^l(P^{M_1},Q^{M_2},\alpha)\big)+O\left(\frac{1}{\sqrt{n}}\right)\label{sd:proof:constep5}
\end{align}
where \eqref{sd:proof:constep5} follows by applying the multivariate Berry-Esseen theorem similarly to \eqref{multivariateuse}.

Recall the definition of $\nu_l^*(n,\varepsilon|\cdot)$ in \eqref{def:nul*} and note that $\nu_l^*(n,\varepsilon|\cdot)$ is a decreasing function of $\varepsilon$. Fix any positive real number $\delta\in(0,1-\varepsilon)$. If
\begin{align}
L=\nu_l^*(\varepsilon+\delta+O(1/\sqrt{n})|P^{M_1},Q^{M_2},K,\alpha),
\end{align}
it follows that
\begin{align}
\zeta(\phi_{n,N}^{\rmU,K}|P^{M_1},Q^{M_2})\geq \varepsilon+\delta>\varepsilon.
\end{align}
Thus, to ensure that $\zeta(\phi_{n,N}^{\rmU,K}|P^{M_1},Q^{M_2})\leq \varepsilon$, we should have $L\geq \nu_l^*(\varepsilon+O(1/\sqrt{n})|P^{M_1},Q^{M_2},K,\alpha)$. Therefore, using the definition of $\lambda^*_{\mathrm{SD}}(\cdot)$ in \eqref{def:l*:sd} and the choice of $\lambda$ in \eqref{lambda:choose:sr:known:2}, noting that $\delta_{n,N'}=O(\log n/n)$ and applying Taylor's expansion of $\nu_l^*(\cdot)$ around $\varepsilon$, the converse proof of Theorem \ref{sr:known} is completed.

\section{Proofs for Unknown Number of Matches}
\label{sec:proofs:unknown}

\subsection{Proof of Large Deviations (Theorem \ref{ld:unknown})}
\label{proof:ld:unknown}
The proof of Theorem \ref{ld:unknown} is similar to that of Theorem \ref{ld:known} and we only emphasize the differences here. Recall that when the number of matches $K$ is unknown, we need to consider all possible hypotheses from $K=0$ to $K=M_2$ and we use $\calH_K$ to denote the set of all hypotheses when then number of matches is $K$.

\subsubsection{Mismatch Probability} We first upper bound the mismatch probability. Fix $l\in[T_K]$, where $K\in[M_2]$ is arbitrary. Since a mismatch event occurs if the number of matches is incorrectly estimated or if the decision of the test $\phi_{n,N}^{\rmU,K}$ is incorrect, it follows that
\begin{align}
\beta(\phi_{n,N}^\rmM|P^{M_1},Q^{M_2})
&=\Pr\{\phi_{n,N}^\rmM(\bX^N,\bY^n)\notin\{\rmH_l^K,\rmH_\rmr\}\}\\
\nn&\leq \Pr\{\exists~\hatK\in[M_2]:~\hatK> K\mathrm{~and~}\underline{S}_{\hatK}(\bX^N,\bY^n)\leq\lambda_{1,n}\}
+\Pr\{\underline{S}_K(\bX^N,\bY^n)>\lambda_{1,n}\}
\\*
&\qquad+\Pr\{\phi_{n,N}^{\rmU,K}(\bX^N,\bY^n)\notin\{\rmH_l^K,\rmH_\rmr\}\}\label{ld:unknown:step1}.
\end{align}
Note that the third term in \eqref{ld:unknown:step1} is exactly the mismatch probability of Unnikrishnan's test with threshold $\lambda_2$ when the number of matches is known to be $K$, which is upper bounded in \eqref{achldfinal}.  The second term in \eqref{ld:unknown:step1} is upper bounded as follows:
\begin{align}
\Pr\{\underline{S}_K(\bX^N,\bY^n)>\lambda_{1,n}\}
&=\Pr\Big\{\min_{t\in[T_K]}\rmS_t^K(\bX^N,\bY^n)\geq \lambda_{1,n}\Big\}\\
&\leq\Pr\big\{\rmS_l^K(\bX^N,\bY^n)\geq \lambda_{1,n}\big\}\\
&\leq \exp(-n\lambda_1)\label{sameasachldknown},
\end{align}
where \eqref{sameasachldknown} follows from same steps to prove \eqref{achldfinal}. The first term in \eqref{ld:unknown:step1} is upper bounded as follows:
\begin{align}
\nn&\Pr\{\exists~\hatK\in[M_2]:~\hatK>K\mathrm{~and~}\underline{S}_{\hatK}(\bX^N,\bY^n)\leq\lambda_{1,n}\}\\*
&\leq \sum_{\hatK\in[K+1:M_2]}\Pr\{\underline{S}_{\hatK}(\bX^N,\bY^n)\leq\lambda_{1,n}\}\\
&=\sum_{\hatK\in[K+1:M_2]}\Pr\{\exists~t\in[T_{\hatK}],~\rmS_t^{\hatK}(\bX^N,\bY^n)\leq\lambda_{1,n}\}\\
&\leq \sum_{\hatK\in[K+1:M_2]}\sum_{t\in[T_{\hatK}]}\Pr\{\rmS_t^{\hatK}(\bX^N,\bY^n)\leq\lambda_{1,n}\}\label{ld:unknown:step2},
\end{align}
where the probability term in the last inequality can be bounded using the method of types as follows:
\begin{align}
\nn&\Pr\{\rmS_t^{\hatK}(\bX^N,\bY^n)\leq\lambda_{1,n}\}\\*
&=\sum_{(\bx^N,\by^n):\rmS_t^{\hatK}(\bx^N,\by^n)\leq\lambda_{1,n}}\Big(\prod_{i\in[M_1]}P_i^N(x_i^N)\Big)\Big(\prod_{j\in[M_2]}Q_j^n(y_j^n)\Big)\\
&=\sum_{\substack{\{(\Omega^{M_1},\Psi^{M_2})\}\in(\calP^N(\calX))^{M_1}\times(\calP^n(\calX))^{M_2}:\\
\sum_{(i,j)\in\calM_t^{\hatK}}\mathrm{GJS}(\Omega_i,\Psi_i,\alpha)\leq\lambda_{1,n}}}\Big(\prod_{i\in[M_1]}P_i^N(\calT_{\Omega_i}^N)\Big)\Big(\prod_{j\in[M_2]}
Q_j^n(\calT_{\Psi_j}^n)\Big)\\
&\leq \sum_{\substack{\{(\Omega^{M_1},\Psi^{M_2})\}\in(\calP^N(\calX))^{M_1}\times(\calP^n(\calX))^{M_2}:\\
\sum_{(i,j)\in\calM_t^{\hatK}}\mathrm{GJS}(\Omega_i,\Psi_i,\alpha)\leq\lambda_{1,n}}}\exp\bigg(-n\Big(\sum_{i\in[M_2]}D(\Omega_i\|P_i)+\sum_{j\in[M_2]}D(\Psi_j\|Q_j)\Big)\bigg)\\
&\leq (N+1)^{M_1|\calX|}(n+1)^{M_2|\calX|}\exp\bigg(-n\min_{\substack{\{(\Omega^{M_1},\Psi^{M_2})\}\in(\calP^N(\calX))^{M_1}\times(\calP^n(\calX))^{M_2}:\\
\sum_{(i,j)\in\calM_t^{\hatK}}\mathrm{GJS}(\Omega_i,\Psi_i,\alpha)\leq\lambda_{1,n}}}\Big(\sum_{i\in[M_2]}D(\Omega_i\|P_i)+\sum_{j\in[M_2]}D(\Psi_j\|Q_j)\Big)\bigg)\label{ld:unknown:step3}.
\end{align}

Note that under hypothesis $\rmH_l^K$, $P_i=Q_j$ for $(i,j)\in\calM_l^K$. Thus,
\begin{align}
\nn&\sum_{i\in[M_2]}D(\Omega_i\|P_i)+\sum_{j\in[M_2]}D(\Psi_j\|Q_j)\\*
\nn&=\sum_{(i,j)\in\calM_l^K}\big(D(\Omega_i\|P_i)+D(\Psi_j\|P_i)\big)+\sum_{\substack{(i,j)\in(\calM_t^{\hatK})^\rmc\cap(\calM_l^K)^\rmc}}\big(D(\Omega_i\|P_i)+D(\Psi_j\|Q_j)\big)\\*
&\qquad+\sum_{\substack{(i,j)\in\calM_t^{\hatK}\cap(\calM_l^K)^\rmc}}\big(D(\Omega_i\|P_i)+D(\Psi_j\|Q_j)\big).
\end{align}
It follows that
\begin{align}
\nn&\min_{\hatK\in[K+1:M_2]}\min_{t\in[T_{\hatK}]}\min_{\substack{\{(\Omega^{M_1},\Psi^{M_2})\}\in(\calP^N(\calX))^{M_1}\times(\calP^n(\calX))^{M_2}:\\
\sum_{(i,j)\in\calM_t^{\hatK}}\mathrm{GJS}(\Omega_i,\Psi_i,\alpha)\leq\lambda_{1,n}}}\Big(\sum_{i\in[M_2]}D(\Omega_i\|P_i)+\sum_{j\in[M_2]}D(\Psi_j\|Q_j)\Big)\\*
&=\min_{\hatK\in[K+1:M_2]}\min_{t\in[T_{\hatK}]}\min_{\substack{\{\Omega_i,\Psi_j\}_{(i,j)\in\calM_t^{\hatK}\cap(\calM_l^K)^\rmc}:\\
\sum_{\substack{(i,j)\in\calM_t^{\hatK}\cap(\calM_l^K)^\rmc}}\mathrm{GJS}(\Omega_i,\Psi_i,\alpha)\leq\lambda_{1,n}
}}
\sum_{\substack{(i,j)\in\calM_t^{\hatK}\cap(\calM_l^K)^\rmc}}\big(D(\Omega_i\|P_i)+D(\Psi_j\|Q_j)\big)\\
&=\min_{t\in[T_{K+1}]}\min_{\substack{\{\Omega_i,\Psi_j\}_{(i,j)\in\calM_t^{K+1}\cap(\calM_l^K)^\rmc}:\\
\sum_{\substack{(i,j)\in\calM_t^{K+1}\cap(\calM_l^K)^\rmc}}\mathrm{GJS}(\Omega_i,\Psi_i,\alpha)\leq\lambda_{1,n}
}}
\sum_{\substack{(i,j)\in\calM_t^{K+1}\cap(\calM_l^K)^\rmc}}\big(D(\Omega_i\|P_i)+D(\Psi_j\|Q_j)\big)\label{whysmallK}\\
&=\min_{\substack{(i,j)\in[M_1]\times[M_2]:\\i\notin\calA_l^K,~j\notin\calB_l^K}}
\min_{\substack{(\Omega,\Psi):\mathrm{GJS}(\Omega,\Psi,\alpha)\leq\lambda_{1,n}}}\big(\alpha D(\Omega\|P_i)+D(\Psi\|Q_j)\big)\label{whysingleK},
\end{align}
where \eqref{whysmallK} follows since it suffices to consider $\hatK=K+1$ so that $\calM_t^{\hatK}\cap(\calM_l^K)^\rmc$ is a single element to achieve the outer minimization over $\sum_{\hatK\in[K+1:M_2]}$ since other values of $\hatK$ would lead to a larger value for similar reasons leading to \eqref{nonincreasing}, and \eqref{whysingleK} follows since, when $\calM_t^{\hatK}\cap(\calM_l^K)^\rmc$ is a single element, it is equivalent to optimize over a single pair of distributions $(\Omega,\Psi)$ and the sum term in \eqref{whysmallK} reduces to a single term.

Thus, using the definition of $f_{l,K}(\cdot)$ in \eqref{def:flambda1}, and combining \eqref{ld:unknown:step2}, \eqref{ld:unknown:step3} and \eqref{whysingleK},
we have
\begin{align}
\nn&\Pr\{\exists~\hatK\in[M_2]:~\hatK>K\mathrm{~and~}\underline{S}_{\hatK}(\bX^N,\bY^n)\leq\lambda_{1,n}\}\\*
&\leq \sum_{\hatK\in[K+1:M_2]}\sum_{t\in[T_{\hatK}]}(N+1)^{\hatK|\calX|}(n+1)^{\hatK|\calX|}\times\exp(-nf_{l,K}(\lambda_{1,n},P^{M_1},Q^{M_2}))\\
&\leq M_2T(N+1)^{M_2|\calX|}(n+1)^{M_2|\calX|}\exp(-nf_{l,K}(\lambda_{1,n},P^{M_1},Q^{M_2}))\label{upp:firstterm:unknowk}.
\end{align}

Combining the results in \eqref{achldfinal}, \eqref{ld:unknown:step1}, \eqref{sameasachldknown} and \eqref{upp:firstterm:unknowk}, we have
\begin{align}
\nn&\beta(\phi_{n,N}^\rmM|P^{M_1},Q^{M_2})\\*
&\leq \sum_{\hatK\in[K+1:M_2]}\sum_{t\in[T_{\hatK}]}(N+1)^{\hatK|\calX|}(n+1)^{\hatK|\calX|}\exp(-nf_{l,K}(\lambda_{1,n},P^{M_1},Q^{M_2}))+\exp(-n\lambda_1)+\exp(-n\lambda_2)\\
&\leq M_2T(N+1)^{M_2|\calX|}(n+1)^{M_2|\calX|}\exp(-nf_{l,K}(\lambda_{1,n},P^{M_1},Q^{M_2}))+\exp(-n\lambda_1)+\exp(-n\lambda_2)\label{ld:unknown:betal:upp}.
\end{align}
Thus, using the definitions of $\lambda_{i,n}$ in \eqref{def:lambdain}, we have
\begin{align}
\liminf_{n\to\infty}-\frac{1}{n}\log\beta(\phi_{n,N}^\rmM|P^{M_1},Q^{M_2})
&\geq \min\big\{\lambda_1,\lambda_2,f_{l,K}(\lambda_1,P^{M_1},Q^{M_2})\big\}.
\end{align}

\subsubsection{False Reject Probability} We next bound the false reject probability. Note that $\phi_{n,N}^{\rmM}(\bX^N,\bY^n)=\rmH_\rmr$ if one of the following events occurs:
\begin{align}
\calR_1&:=\Big\{\min_{\hatK\in[M_2]}\underline{S}_{\hatK}(\bX^N,\bY^n)\geq \lambda_{1,n}\Big\},\\
\calR_{2,\hatK}&:=\Big\{\underline{S}_{\hatK}(\bX^N,\bY^n)\leq\lambda_{1,n},~\mathrm{and}~\phi_{n,N}^{\rmU,\hatK}(\bX^N,\bY^n)=\rmH_\rmr\Big\},~\hatK\in[M_2].
\end{align} 
It follows that
\begin{align}
\zeta(\phi_{n,N}^\rmM|P^{M_1},Q^{M_2})
&=\Pr\{\phi_{n,N}^\rmM(\bX^N,\bY^n)=\rmH_\rmr\}\\
&\leq \bbP\{\calR_1\}+\sum_{\hatK\in[M_2]}\Pr\{\calR_{2,\hatK}\}\label{proof:ld:step4}.
\end{align}
The first term in \eqref{proof:ld:step4} is upper bounded as follows:
\begin{align}
\bbP\{\calR_1\}
&\leq \Pr\big\{\underline{S}_K(\bX^N,\bY^n)\geq \lambda_{1,n}\big\}\\*
&\leq \exp(-n\lambda_1)\label{sameasachldknown2},
\end{align}
where \eqref{sameasachldknown2} follows from \eqref{sameasachldknown}.

The second term in \eqref{proof:ld:step4} is upper bounded as follows:
\begin{align}
\sum_{\hatK\in[M_2]}\Pr\{\calR_{2,j}\}
&\leq \sum_{\hatK\in[K+1:M_2]}\Pr\{\underline{S}_{\hatK}(\bX^N,\bY^n)\leq\lambda_{1,n}\}+\Pr\{\phi_{n,N}^{\rmU,K}(\bX^N,\bY^n)=\rmH_\rmr\}\label{proof:ld:step6},
\end{align}
where the first term is upper bounded by \eqref{ld:unknown:step2} and the second term is exactly the false reject probability of Unnikrishnan's test analyzed in \eqref{frej:ld:step1} and \eqref{relaxations}. Combining \eqref{proof:ld:step4}, \eqref{sameasachldknown} and \eqref{proof:ld:step6}, we have
\begin{align}
\nn&\zeta(\phi_{n,N}^\rmM|P^{M_1},Q^{M_2})\\*
&\leq \exp(-n\lambda_1)+M_2T(N+1)^{M_2|\calX|}(n+1)^{M_2|\calX|}\exp(-nf_{l,K}(\lambda_{1,n},P^{M_1},Q^{M_2}))+\Pr\{\phi_{n,N}^{\rmU,K}(\bX^N,\bY^n)=\rmH_\rmr\}\label{ld:unknown:freject:upp}.
\end{align}
Therefore,
\begin{align}
\liminf_{n\to\infty}-\frac{1}{n}\log\zeta(\phi_{n,N}^\rmM|P^{M_1},Q^{M_2})
&\geq \min\Big\{\lambda_1,f_{l,K}(\lambda_1,P^{M_1},Q^{M_2}),F_l(P^{M_1},Q^{M_2},\alpha,\lambda_2,K)\Big\}.
\end{align}

\subsubsection{False Alarm Probability} Finally, we bound the false alarm probability. It follows that
\begin{align}
\nn&\eta(\phi_{n,N}^\rmM|P^{M_1},Q^{M_2})\\*
&=\Pr\{\phi_{n,N}^\rmM(\bX^N,\bY^n)\neq\rmH_\rmr\}\\
&=\Pr\Big\{\exists~j\in[M_2]:~\min_{\hatK\in[j+1:M_2]}\underline{S}_{\hatK}(\bX^N,\bY^n)\geq \lambda_{1,n},~\underline{S}_j(\bX^N,\bY^n)\leq\lambda_{1,n},~\mathrm{and}~\phi_{n,N}^{\rmU,j}(\bX^N,\bY^n)\neq\rmH_\rmr\Big\}\\
&\leq \sum_{j\in[M_2]}\Pr\{\underline{S}_j(\bX^N,\bY^n)\leq\lambda_{1,n}\}\label{ld:falarm:step1}.
\end{align}
For each $j\in[M_2]$, the probability term in \eqref{ld:falarm:step1} can be upper bounded as follows:
\begin{align}
\nn&\Pr\{\underline{S}_j(\bX^N,\bY^n)\leq\lambda_{1,n}\}\\*
&=\Pr\{\exists~t\in[T_j]:~\rmS_t^j(\bX^N,\bY^n)\leq\lambda_{1,n}\}\\
&\leq \sum_{t\in[T_j]}\Pr\{\rmS_t^j(\bX^N,\bY^n)\leq\lambda_{1,n}\}\\
&\leq \sum_{t\in[T_j]}(N+1)^{j|\calX|}(n+1)^{j|\calX|}\exp\bigg(-n\min_{\substack{\{(\Omega_i,\Psi_i)\}_{i\in\calA_t^j}\in\calP^{2j}(\calX):\\
\sum_{i\in\calA_t^j}\mathrm{GJS}(\Omega_i,\Psi_i,\alpha)\leq\lambda_{1,n}}}\sum_{i\in\calA_t^j}\Big(\alpha D(\Omega_i\|P_i)+D(\Psi_i\|Q_{\sigma_t^j(i)})\Big)\bigg)\label{ld:falarm:step2},
\end{align}
where \eqref{ld:falarm:step2} follows from similar steps leading to \eqref{ld:unknown:step3}.

Note that
\begin{align}
\nn&\min_{j\in[M_2]}\min_{t\in[T_j]}\min_{\substack{\{(\Omega_i,\Psi_i)\}_{i\in\calA_t^j}\in\calP^{2j}(\calX):\\
\sum_{i\in\calA_t^j}\mathrm{GJS}(\Omega_i,\Psi_i,\alpha)<\lambda_{1,n}}}\sum_{i\in\calA_t^j}\Big(\alpha D(\Omega_i\|P_i)+D(\Psi_i\|Q_{\sigma_t^j(i)})\Big)\\
&=\min_{t\in[T_1]}\min_{\substack{\{(\Omega_i,\Psi_i)\}_{i\in\calA_t^1}\in\calP^2(\calX):\\
\sum_{i\in\calA_t^1}\mathrm{GJS}(\Omega_i,\Psi_i,\alpha)<\lambda_{1,n}}}\sum_{i\in\calA_t^1}\Big(\alpha D(\Omega_i\|P_i)+D(\Psi_i\|Q_{\sigma_t^1(i)})\Big)\label{whyt=1}\\
&=\min_{\substack{(i,j)\in[M_1]\times[M_2]}}\min_{\substack{(\Omega,\Psi)\in\calP^2(\calX):\\\mathrm{GJS}(\Omega,\Psi,\alpha)\leq \lambda_{1,n}}} (D(\Omega\|P_i)+D(\Psi\|Q_j))\label{uset1}\\
&=f_0(\lambda_{1,n},P^{M_1},Q^{M_2})\label{usef:falarm},
\end{align}
where \eqref{whyt=1} follows since the KL divergence terms are non-negative and the objective function is smaller for a smaller value of $j$, \eqref{uset1} follows since $|\calA_t^K|=1$ for each $t\in[T_1]$, \eqref{usef:falarm} follows from the definition of $f_0(\cdot)$ in \eqref{def:flambda0}.

It follows from \eqref{ld:falarm:step1}, \eqref{ld:falarm:step2} and \eqref{usef:falarm} that
\begin{align}
\eta(\phi_{n,N}^\rmM|P^{M_1},Q^{M_2})
&\leq \sum_{j\in[M_2]}\sum_{t\in[T_j]}(N+1)^{j|\calX|}(n+1)^{j|\calX|}\exp(-nf_0(\lambda_{1,n},P^{M_1},Q^{M_2}))\\
&\leq T(N+1)^{M_2|\calX|}(n+1)^{M_2|\calX|}\exp(-nf_0(\lambda_{1,n},P^{M_1},Q^{M_2}))\label{usedefT},
\end{align}
where \eqref{usedefT} follows from the definition of $T=\sum_{j\in[M_2]}T_j$. 
Therefore, using the definition of $\lambda_{1,n}$ in \eqref{def:lambdain}, we have
\begin{align}
\liminf_{n\to\infty}-\frac{1}{n}\log\eta(\phi_{n,N}^\rmM|P^{M_1},Q^{M_2})
&\geq f_0(\lambda_1,P^{M_1},Q^{M_2}).
\end{align}

\subsection{Proof of Small Deviations (Theorem \ref{sr:unknown})}
\label{proof:sr:unknown}
For brevity and to minimize redundancy, we only sketch the steps in the proof. The bounds on mismatch and false alarm probabilities are given in \eqref{ld:unknown:betal:upp} and \eqref{usedefT}, respectively. The bounds on the false reject probability follows by combining \eqref{multivariateuse} and
\eqref{ld:unknown:freject:upp}.

\section{Conclusion}
\label{sec:conc}
We revisited the problem of statistical sequence matching and derived theoretical results under both large and small deviations regimes for the GLRT. When the number of matches is known, we completely characterized the tradeoff between the mismatch and false reject probabilities for Unnikrishnan's GLRT in \eqref{test:unn} under each hypothesis. We illustrated our results via numerical examples and compared the performance of Unnikrishnan's with a simple test that repeatedly uses Gutman's multiple classification test. When specialized to multiple classification, our large deviations result recovered those of Gutman~\cite{gutman1989asymptotically} and our small deviations result strengthened \cite[Theorem 4.1]{zhou2018binary}. Finally, we generalized our achievability analyses to the case when the number of matches is unknown and derived the tradeoff among the probabilities of mismatch, false reject and false alarm.

There are several avenues for future research. Firstly, we focused on discrete sequences in this paper so that the method of types~\cite{csiszar1998mt} applies. It would be worthwhile to generalize our results to continuous sequences that are generated from a probability density function. In this case, novel ideas such as the kernel method using maximum mean discrepancy~\cite{gretton2012jmlr} could be helpful. Secondly, we focused on exact match in this paper so that each matched pair of sequences should be identified. It would be of interest to generalize our results to account for partial match where only a subset of all matched pairs of sequences is identified. The ideas in partial recovery for group testing~\cite{Scarlett2016,Scarlett2017} might be helpful in this direction of studies. Thirdly, we focused on fixed-length tests in this paper, where the sequences are collected before a decision is made. In practical scenarios, it is possible that the sequence in each database is collected in a sequential manner. It is thus of interest to study the performance of a sequential test. For this purpose, the analyses in \cite{mahdi2021sequential,Ihwang2022sequential} for statistical classification could be generalized to sequential statistical sequence matching. Fourthly, we focused on theoretical benchmarks and used highly complicated tests of generalized likelihood ratio tests that checked all possibilities. It would be interesting to develop low complexity tests as discussed in \cite[Section V]{unnikrishnan2015asymptotically} for both discrete and continuous sequences. Some successful attempts have been made in~\cite{Unnikrishnan2013,Naini2016}. Finally, for the case where the number of matches is unknown, we only derived achievability results for the proposed two stage test that first estimates the number of matches and then applies Unnikrishnan's GLRT. However, the converse part remains unexplored. It is worthwhile to derive a converse result under a certain performance criterion to either prove or disapprove the optimality of our proposed two-phase test.

\appendix
\subsection{Proof of Lemma \ref{prop:Fl}}
\label{proof:prop:Fl}
Claim (i) follows directly from the definition of $F_l(\cdot)$. We next justify Claim (ii). It follows from the definition of $F_l(\cdot)$ in \eqref{def:Fl} that $F_l(P^{M_1},Q^{M_2},\alpha,\lambda,K)=0$ if and only if $(\Omega^{M_1},\Psi^{M_2})=(P^{M_1},Q^{M_2})$ is a feasible pair of distributions so that there exists $(t,j)\in[T_K]^2:t\neq s$ satisfying $\rmG_t^K(P^{M_1},Q^{M_2},\alpha)\leq \lambda$ and $\rmG_s^K(P^{M_1},Q^{M_2},\alpha)\leq \lambda$. Under hypothesis $\rmH_l^K$, we have $\rmG_l^K(P^{M_1},Q^{M_2},\alpha)=0$ and thus if and only if
\begin{align}
\lambda\geq \min_{t\in[T_K]:t\neq l}\rmG_t^K(P^{M_1},Q^{M_2},\alpha)=\Lambda_l(P^{M_1},Q^{M_2},K,\alpha),
\end{align}
we have $F_l(P^{M_1},Q^{M_2},\alpha,\lambda,K)=0$

Finally, we justify Claim (iii). Recall that for each hypothesis $l\in[T_K]$, $\calA_l^K$ specifies the indices of sequences of $\bX^N$ that have matched sequences in $\bY^n$, $\calB_l^K$ specifies the indices of sequences o $\bY^n$ that have matched sequences in $\bX^N$, $\sigma_l^K:\calA_l^K\to\calB_l^K$ specifies the unique pairs of matching while $\calM_l^K=\{(i,\sigma_l^K(i))\}_{i\in\calA_l^K}$. Given any $(t,s)\in[T_K]^2$ such that $t\neq s$, define the set 
\begin{align}
\calP_{t,s}^{M+N}(\calX)
&:=\Big\{(\Omega^{M_1},\Psi^{M_2})\in\calP^{M_1+M_2}(\calX):~\Psi_j=\Omega_i,~\forall~(i,j)\in(\calM_t^K\cup\calM_s^K)\Big\}.
\end{align}
It follows from the definition of $\rmG_l^K(\cdot)$ in \eqref{def:Gl} that $\rmG_t^K(\Omega^{M_1},\Psi^{M_2},\alpha)=0$ and $\rmG_s^K(\Omega^{M_1},\Psi^{M_2},\alpha)=0$ imply that $(\Omega^{M_1},\Psi^{M_2})\in\calP_{t,s}^{M+N}(\calX)$ and the definition of $E_l(\cdot)$ in \eqref{def:El} leads to
\begin{align}
\nn&E_l(P^{M_1},Q^{M_2},\Omega^{M_1},\Psi^{M_2},\alpha)\\*
&=\sum_{i\in[M_1]}\alpha D(\Omega_i\|P_i)+\sum_{j\in\calB_l^K}D(\Psi_j\|P_{(\sigma_l^K)^{-1}(j)})+\sum_{j\notin\calB_l^K}D(\Psi_j\|Q_j)\\
\nn&=\sum_{j\in\calB_l^K:j\in(\calB_t^K\cup\calB_s^K)}\Big(D(\Psi_j\|P_{(\sigma_l^K)^{-1}(j)})+\sum_{\substack{i\in[M_1]:\sigma_t^K(i)=j\\\mathrm{or}~\sigma_s^K(i)=j}}\alpha D(\Psi_j\|P_i)\Big)+\sum_{j\in\calB_l^K:j\notin(\calB_t^K\cup\calB_s^K)}D(\Psi_j\|P_{(\sigma_l^K)^{-1}(j)})\\*
\nn&\qquad+\sum_{j\in(\calB_l^K)^\rmc:j\in(\calB_t^K\cup\calB_s^K)}\Big(D(\Psi_j\|Q_j)+\sum_{\substack{i\in[M_1]:\sigma_t^K(i)=j\\\mathrm{or}~\sigma_s^K(i)=j}}\alpha D(\Psi_j\|P_i)\Big)+\sum_{j\in(\calB_l^K)^\rmc:j\notin(\calB_t^K\cup\calB_s^K)}D(\Psi_j\|Q_j)\\
&\qquad+\sum_{i\in[M_1]:i\notin(\calA_t^K\cup\calA_s)}\alpha D(\Omega_i\|P_i)\label{El:decompose}.
\end{align}
It follows from the definition of $F_l(\cdot)$ in \eqref{def:Fl} that
\begin{align}
F_l(P^{M_1},Q^{M_2},\alpha,0,K)
&=\min_{\substack{(t,s)\in[T_K]^2:\\t\neq s}}\min_{\substack{(\Omega^{M_1},\Psi^{M_2})\in\calP^{M_1+M_2}(\calX):\\\rmG_t^K(\Omega^{M_1},\Psi^{M_2},\alpha)=0\\\rmG_s^K(\Omega^{M_1},\Psi^{M_2},\alpha)=0}}E_l(P^{M_1},Q^{M_2},\Omega^{M_1},\Psi^{M_2},\alpha)\\
&=\min_{(t,s)\in[T_K]^2:t\neq s}\bigg(\sum_{j\in\calB_l^K:j\in(\calB_t^K\cup\calB_s^K)}f_j^{l,t,s}(P^{M_1})+\sum_{j\in(\calB_l^K)\rmc:j\in(\calB_t^K\cup\calB_s^K)}g_j^{l,t,s}(P^{M_1},Q^{M_2})\bigg)\label{usefgj}\\
&=\Upsilon_l(P^{M_1},Q^{M_2})\label{useUpsilonl},
\end{align}
where \eqref{usefgj} follows from simple algebra and the definitions of $\rmG_{\cdot}(\cdot)$ in \eqref{def:Gl}, $f_j(\cdot)$ in \eqref{def:fj} and $g_j$ in \eqref{def:gj} and \eqref{useUpsilonl} follows from the definition of $\Upsilon_l(\cdot)$ in \eqref{def:Upsilonl}.

\subsection{Justification of \eqref{simpleworse}}
\label{just:simpleworse}

It follows from the definition of $E_l(\cdot)$ in \eqref{def:El} that when $M_2=K$,
\begin{align}
E_l(P^{M_1},Q^{M_2},\Omega^{M_1},\Psi^{M_2},\alpha)
&=\sum_{i\in[M_1]}\alpha D(\Omega_i\|P_i)+\sum_{j\in[M_2]}D(\Psi_j\|P_{(\sigma_l^K)^{-1}(j)}).
\end{align}
Thus, using the definition of $F_l(\cdot)$ in \eqref{def:Fl}, we have
\begin{align}
F_l(P^{M_1},Q^{M_2},\alpha,\lambda,K)
&=\min_{\substack{(t,s)\in[T_K]^2:t\neq s}}\min_{\substack{(\Omega^M,\Psi^M)\in\calP^{2M}(\calX):\\\rmG_t^K(\Omega^{M_1},\Psi^{M_2},\alpha)\leq \lambda\\\rmG_s^K(\Omega^{M_1},\Psi^{M_2},\alpha)\leq \lambda}}E_l(P^{M_1},Q^{M_2},\Omega^{M_1},\Psi^{M_2},\alpha)\\
&=\min_{\substack{(t,s)\in[T_K]^2:t\neq s}}\min_{\substack{(\Omega^M,\Psi^M)\in\calP^{2M}(\calX):
\\\sum_{\barj\in[M_2]}\mathrm{GJS}(\Omega_{(\sigma_t^K)^{-1}(\barj)},\Psi_{\barj},\alpha)\leq \lambda
\\\sum_{\barj\in[M_2]}\mathrm{GJS}(\Omega_{(\sigma_s^K)^{-1}(\barj)},\Psi_{\barj},\alpha)\leq \lambda
}}\Big(\sum_{i\in[M_1]}\alpha D(\Omega_i\|P_i)+\sum_{j\in[M_2]}D(\Psi_j\|P_{(\sigma_l^K)^{-1}(j)})\Big)\\
&\geq \min_{\substack{(t,s)\in[T_K]^2:t\neq s}}\min_{\substack{(\Omega^M,\Psi^M)\in\calP^{2M}(\calX):
\\\sum_{\barj\in[M_2]}\mathrm{GJS}(\Omega_{(\sigma_t^K)^{-1}(\barj)},\Psi_{\barj},\alpha)\leq \lambda
\\\sum_{\barj\in[M_2]}\mathrm{GJS}(\Omega_{(\sigma_s^K)^{-1}(\barj)},\Psi_{\barj},\alpha)\leq \lambda
}}\max_{j\in[M_2]}\Big(\sum_{i\in[M_1]}\alpha D(\Omega_i\|P_i)+D(\Psi_j\|P_{(\sigma_l^K)^{-1}(j)})\Big)\label{lb:unconstrained0}\\
&\geq \min_{\substack{(t,s)\in[T_K]^2:t\neq s}}\max_{\substack{j\in[M_2]:\\(\sigma_t^K)^{-1}(j)\neq (\sigma_s^K)^{-1}(j)}}
\min_{\substack{(\Omega^M,\Psi^M)\in\calP^{2M}(\calX):
\\\mathrm{GJS}(\Omega_{(\sigma_t^K)^{-1}(j)},\Psi_j,\alpha)\leq \lambda
\\\mathrm{GJS}(\Omega_{(\sigma_s^K)^{-1}(j)},\Psi_j,\alpha)\leq \lambda
}}
\Big(\sum_{i\in[M_1]}\alpha D(\Omega_i\|P_i)+D(\Psi_j\|P_{(\sigma_l^K)^{-1}(j)})\Big)\label{lb:unconstrained1}\\
&\geq \min_{\substack{(t,s)\in[T_K]^2:t\neq s}}\min_{\substack{j\in[M_2]}}
\min_{\substack{(\Omega^M,\Psi_j)\in\calP^{M+1}(\calX):
\\(\sigma_t^K)^{-1}(j)\neq (\sigma_s^K)^{-1}(j)
\\\mathrm{GJS}(\Omega_{(\sigma_t^K)^{-1}(j)},\Psi_j,\alpha)\leq \lambda
\\\mathrm{GJS}(\Omega_{(\sigma_s^K)^{-1}(j)},\Psi_j,\alpha)\leq \lambda
}}
\Big(\sum_{i\in[M_1]}\alpha D(\Omega_i\|P_i)+D(\Psi_j\|P_{(\sigma_l^K)^{-1}(j)})\Big)\label{lb:unconstrained2}\\
&=\min_{\substack{j\in[M_2]}}
\min_{\substack{(t,s)\in[M_1]^2:t\neq s}}
\min_{\substack{(\Omega^M,\Psi_j)\in\calP^{M+1}(\calX):
\\\mathrm{GJS}(\Omega_t,\Psi_j,\alpha)\leq \lambda
\\\mathrm{GJS}(\Omega_s,\Psi_j,\alpha)\leq \lambda
}}
\Big(\sum_{i\in[M_1]}\alpha D(\Omega_i\|P_i)+D(\Psi_j\|P_{(\sigma_l^K)^{-1}(j)})\Big)\label{lb:unconstrained3}\\
&=\min_{j\in[M_2]}F_l(P^{M_1},P_{(\sigma_l^K)^{-1}(j)},\alpha,\lambda,K)\\
&=\min_{j\in[M_2]}F_l(P^{M_1},Q_j,\alpha,\lambda,K),
\end{align}
where \eqref{lb:unconstrained0} follows by lower bounding the second sum in the objective function with the maximal divergence term $D(\Psi_j\|P_{(\sigma_l^K)^{-1}(j)})$, \eqref{lb:unconstrained1} follows since achieved value of a minimax optimization problem is lower bounded by the achieved value of a corresponding maximin optimization problem, the optimization range $j\in[M_2]$ is squeezed by adding a constraint $(\sigma_t^K)^{-1}(j)\neq (\sigma_s^K)^{-1}(j)$, which is valid since any two different hypotheses map at least one $y_j^n$ sequences into difference sequences in $\bx^N$, \eqref{lb:unconstrained2} follows by replacing the maximization over $j$ with minimization, and \eqref{lb:unconstrained3} holds since $(\sigma_t^K)^{-1}\in[M_1]$, $(\sigma_s^K)^{-1}\in[M_1]$ and $(\sigma_t^K)^{-1}(j)\neq (\sigma_s^K)^{-1}(j)$ is equivalent to the constraint that $(t,s)\in[M_1]^2:t\neq s$, where we reuse the notation $(t,s)$ for different meanings. 

\subsection{Justification of \eqref{useVt1t2}}
\label{just:usevt1t2}
It follows that
\begin{align}
\nn&\sum_{s\in[N+n]}\bbE\bigg[\Big(\sum_{(i,j)\in\calM_{t_1}^K}Z_s^{i,j}\Big)\Big(\sum_{(\bari,\barj)\in\calM_{t_2}^K}Z_s^{i,j}\Big)\bigg]\\*
&=\sum_{s\in[N]}\bbE\bigg[\Big(\sum_{(i,j)\in\calM_{t_1}^K}Z_s^{i,j}\Big)\Big(\sum_{(\bari,\barj)\in\calM_{t_2}^K}Z_s^{\bari,\barj}\Big)\bigg]+\sum_{s\in[N+1:N+n]}\bbE\bigg[\Big(\sum_{(i,j)\in\calM_{t_1}^K}Z_s^{i,j}\Big)\Big(\sum_{(\bari,\barj)\in\calM_{t_2}^K}Z_s^{\bari,\barj}\Big)\bigg]\\
\nn&=N\bbE\bigg[\Big(\sum_{(i,j)\in\calM_{t_1}^K}\imath_1(X_i|P_i,Q_j,\alpha)\Big)\Big(\sum_{(\bari,\barj)\in\calM_{t_2}^K}\imath_1(X_{\bari}|P_{\bari},Q_{\barj},\alpha)\Big)\bigg]\\*
&\qquad+n\bbE\bigg[\Big(\sum_{(i,j)\in\calM_{t_1}^K}\imath_2(Y_j|P_i,Q_j,\alpha)\Big)\Big(\sum_{(\bari,\barj)\in\calM_{t_2}^K}\imath_2(Y_{\barj}|P_{\bari},Q_{\barj},\alpha)\Big)\bigg]\label{verify:step1},
\end{align}
and
\begin{align}
\nn&\sum_{s\in[N+n]}\Big(\sum_{(i,j)\in\calM_{t_1}^K}\bbE[Z_s^{i,j}]\Big)\Big(\sum_{(\bari,\barj)\in\calM_{t_2}^K}\bbE[Z_s^{\bari,\barj}]\Big)\\*
\nn&=\sum_{s\in[N]}\Big(\sum_{(i,j)\in\calM_{t_1}^K}\bbE_{P_i}[\imath_1(X|P_i,Q_j,\alpha)]\Big)\Big(\sum_{(\bari,\barj)\in\calM_{t_2}^K}\bbE_{P_{\bari}}[\imath_1(X|P_{\bari},Q_{\barj},\alpha)]\Big)\\*
&\qquad+\sum_{s\in[n]}\Big(\sum_{(i,j)\in\calM_{t_1}^K}\bbE_{Q_j}[\imath_2(Y|P_i,Q_j,\alpha)]\Big)\Big(\sum_{(\bari,\barj)\in\calM_{t_2}^K}\bbE_{Q_{\barj}}[\imath_2(Y|P_{\bari},Q_{\barj},\alpha)]\Big)\label{verify:step2}.
\end{align}

The first expectation term in \eqref{verify:step1} satisfies
\begin{align}
\nn&\bbE\bigg[\Big(\sum_{(i,j)\in\calM_{t_1}^K}\imath_1(X_i|P_i,Q_j,\alpha)\Big)\Big(\sum_{(\bari,\barj)\in\calM_{t_2}^K}\imath_1(X_{\bari}|P_{\bari},Q_{\barj},\alpha)\Big)\bigg]\\*
\nn&=\sum_{(i,j)\in\calM_{t_1}^K}\sum_{(\bari,\barj)\in\calM_{t_2}^K:\bari=i}\bbE\big[\imath_1(X_i|P_i,Q_j,\alpha)\imath_1(X_{\bari}|P_{\bari},Q_{\barj},\alpha)\big]\\*
&\qquad+\sum_{(i,j)\in\calM_{t_1}^K}\sum_{(\bari,\barj)\in\calM_{t_2}^K:\bari\neq i}\bbE[\imath_1(X_i|P_i,Q_j,\alpha)]\bbE[\imath_1(X_{\bari}|P_{\bari},Q_{\barj},\alpha)]\label{verify:suc1}.
\end{align}
Similarly, the second expectation term in \eqref{verify:step1} satisfies
\begin{align}
\nn&\bbE\bigg[\Big(\sum_{(i,j)\in\calM_{t_1}^K}\imath_2(Y_j|P_i,Q_j,\alpha)\Big)\Big(\sum_{(\bari,\barj)\in\calM_{t_2}^K}\imath_2(Y_{\barj}|P_{\bari},Q_{\barj},\alpha)\Big)\bigg]\\
\nn&=\sum_{(i,j)\in\calM_{t_1}^K}\sum_{(\bari,\barj)\in\calM_{t_2}^K:\barj=j}\bbE\big[\imath_2(Y_j|P_i,Q_j,\alpha)\imath_2(Y_{\bari}|P_{\bari},Q_{\barj},\alpha)\big]\\*
&\qquad+\sum_{(i,j)\in\calM_{t_1}^K}\sum_{(\bari,\barj)\in\calM_{t_2}^K:\barj\neq j}\bbE[\imath_1(Y_j|P_i,Q_j,\alpha)]\bbE[\imath_1(Y_{\barj}|P_{\bari},Q_{\barj},\alpha)]\label{verify:suc2}.
\end{align}
Analogously, the two expectation terms in \eqref{verify:step2} satisfy
\begin{align}
\nn&\Big(\sum_{(i,j)\in\calM_{t_1}^K}\bbE[\imath_1(X_i|P_i,Q_j,\alpha)]\Big)\Big(\sum_{(\bari,\barj)\in\calM_{t_2}^K}\bbE[\imath_1(X_{\bari}|P_{\bari},Q_{\barj},\alpha)]\Big)\\*
&=\sum_{(i,j)\in\calM_{t_1}^K}\sum_{(\bari,\barj)\in\calM_{t_2}^K}\bbE[\imath_1(X_i|P_i,Q_j,\alpha)]\bbE[\imath_1(X_{\bari}|P_{\bari},Q_{\barj},\alpha)],\label{verify:suc3}\\
\nn&\Big(\sum_{(i,j)\in\calM_{t_1}^K}\bbE[\imath_2(Y_j|P_i,Q_j,\alpha)]\Big)\Big(\sum_{(\bari,\barj)\in\calM_{t_2}^K}\bbE[\imath_2(Y_{\bari}|P_{\bari},Q_{\barj},\alpha)]\Big)\\*
&=\sum_{(i,j)\in\calM_{t_1}^K}\sum_{(\bari,\barj)\in\calM_{t_2}^K}\bbE[\imath_2(Y_j|P_i,Q_j,\alpha)]\bbE[\imath_2(Y_{\barj}|P_{\bari},Q_{\barj},\alpha)]\label{verify:suc4}.
\end{align}
The justification of \eqref{useVt1t2} is completed by combining \eqref{verify:suc1} to \eqref{verify:suc4}.

\subsection{Proof of Lemma \ref{prop:flambda1}}
\label{proof:prop:flambda1}

Claim (i) follows from the definition of the exponent function $f_{l,K}(\lambda_1,P^{M_1},Q^{M_2})$. We first prove Claim (ii). Note that $f_{l,K}(\lambda_1,P^{M_1},Q^{M_2})=0$ if $(\Omega,\Psi)=(P_i,Q_j)$ satisfies $\mathrm{GJS}(\Omega,\Psi,\alpha)\leq \lambda_1$ for some $(i,j)\in[M_1]\times[M_2]$ such that $i\notin\calA_l^K$ and $j\notin\calB_l^K$. Therefore, if and only if 
\begin{align}
\lambda_1
&\geq \min_{\substack{(i,j)\in[M_1]\times[M_2]:\\i\notin\calA_l^K,j\notin\calB_l^K}}\mathrm{GJS}(P_i,Q_j,\alpha)\\
&=\rmG_{\mathrm{min}}^{l,K}(P^{M_1},Q^{M_2},\alpha),
\end{align}
the exponent function $f_{l,K}(\lambda_1,P^{M_1},Q^{M_2})$ equals zero. Finally, it follows that
\begin{align}
f_{l,K}(\lambda_1,P^{M_1},Q^{M_2})
&\leq f_{l,K}(0,P^{M_1},Q^{M_2})\\
&=\min_{\substack{(i,j)\in[M_1]\times[M_2]:\\i\notin\calA_l^K,j\notin\calB_l^K}}\min_{\Psi\in\calP(\calX)}\Big(\alpha D(\Omega\|P_i)+D(\Omega\|Q_j)\Big)\\
&=\min_{\substack{(i,j)\in[M_1]\times[M_2]:\\i\notin\calA_l^K,j\notin\calB_l^K}}D_{\frac{\alpha}{1+\alpha}}(P_i\|Q_j),\label{usedef:renyi}
\end{align}
where \eqref{usedef:renyi} follows from the definition of the R\'enyi divergence in \eqref{def:renyi} and the KKT conditions~\cite[Chap. 5]{boyd2004convex} for the convex optimization problem, which implies that
\begin{align}
\min_{\Psi\in\calP(\calX)}\Big(\alpha D(\Omega\|P_i)+D(\Omega\|Q_j)\Big)=D_{\frac{\alpha}{1+\alpha}}(P_i\|Q_j),
\end{align}
and the detailed derivations are available in cf. \cite[Eq. (13)-(16)]{tuncel2005error}.

\subsection{Proof of Lemma \ref{converse:gnp}}
\label{proof:con:gnp}
We first relate the error probabilities of any test with a type-based test. Fix any vector $\bkappa=(\kappa_1,\ldots,\kappa_{T_K})\in[0,1]^{T_K}$ such that $\sum_{i\in[T_K]}\kappa_i\leq 1$, and let
\begin{align}
\kappa_-:=\min_{i\in[T_K]}\kappa_i,~\kappa_+:=\sum_{i\in[T_K]}\kappa_i.
\end{align}

\begin{lemma}
\label{type:optimal}
Given any test $\phi_{n,N}$,  we can construct a test $\phi_{n,N}^\rmT$  that uses types $(\bT_{\bX^N},\bT_{\bY^n})$ such that under each $l\in[T_K]$, under any tuples of generating distributions $(P^{M_1},Q^{M_2})\in\calP_l^K$,
\begin{align}
\beta(\phi_{n,N}|P^{M_1},Q^{M_2})&\geq \kappa_-\beta(\phi_{n,N}^\rmT|P^{M_1},Q^{M_2}),\\
\zeta(\phi_{n,N}|P^{M_1},Q^{M_2})&\geq (1-\kappa_+)\beta(\phi_{n,N}^\rmT|P^{M_1},Q^{M_2}).
\end{align}
\end{lemma}
\begin{proof}
Note that any test $\phi_{n,N}$ partitions the sample space $\calX^{M_1N}\times\calY^{M_2n}$ into disjoint acceptance regions $\calF_1(\phi_{n,N}),\ldots,\calF_{T_K}(\phi_{n,N})$ and a reject region $\calF_\rmr(\phi_{n,N})$ so that for each $l\in[T_K]$, if $(\bx^N,\by^n)\in\calF_l(\phi_{n,N})$, the test decides on $\rmH_l^K$ and if $(\bx^N,\by^n)\in\calF_\rmr(\phi_{n,N})$, the test $\phi_{n,N}$ outputs the decision $\rmH_\rmr$.

Given any tuple of types $(\Omega^{M_1},\Psi^{M_2})\in(\calP^N(\calX))^{M_1}\times(\calP^n(\calX))^{M_2}$, we use $\calT_{\Omega^{M_1},\Psi^{M_2}}$ to denote the set of sequences $(\bx^N,\by^n)$ such that for each $i\in[M]$, $\hatT_{x_i^N}=\Omega_i$ and for each $j\in[N]$, $\hatT_{y_i^n}=\Psi_j$. We construct a type-based test $\phi_{n,N}^\rmT$ as follows:
\begin{align}
\phi_{n,N}^\rmT(\Omega^{M_1},\Psi^{M_2})
&=
\left\{
\begin{array}{ll}
\rmH_l^K&\mathrm{if~}\frac{|\calF_l\cap\calT_{\Omega^{M_1},\Psi^{M_2}}|}{|\calT_{\Omega^{M_1},\Psi^{M_2}}|}\ge \kappa_l\mathrm{~and~}\max_{i\in[l-1]}\bigg(\frac{|\calF_i\cap\calT_{\Omega^{M_1},\Psi^{M_2}}|}{|\calT_{\Omega^{M_1},\Psi^{M_2}}|}-\kappa_i\bigg)\leq 0,\\
\rmH_\rmr&\mathrm{if~}\max_{i\in[T_K]}\bigg(\frac{|\calF_i\cap\calT_{\Omega^{M_1},\Psi^{M_2}}|}{|\calT_{\Omega^{M_1},\Psi^{M_2}}|}-\kappa_i\bigg)\leq 0.
\end{array}
\right.
\end{align}
For simplicity, we use $\calF_l(\phi_{n,N}^\rmT)$ to denote the set of types $(\Omega^{M_1},\Psi^{M_2})\in(\calP^N(\calX))^{M_1}\times(\calP^n(\calX))^{M_2}$ such that $\phi_{n,N}^\rmT(\Omega^{M_1},\Psi^{M_2})=\rmH_l^K$ for each $l\in[T_K]$ and use $\calF_\rmr(\phi_{n,N}^\rmT)$ similarly. Thus, for each $l\in[T_K]$, the mismatch probability satisfies
\begin{align}
\beta(\phi_{n,N}|P^{M_1},Q^{M_2})
&=\sum_{t\in[T_K]:t\neq l}\Pr\{(\bX^N,\bY^n)\in\calF_t(\phi_{n,N})\}\\
&=\sum_{t\in[T_K]:t\neq l}\sum_{\substack{(\Omega^{M_1},\Psi^{M_2})\in(\calP^N(\calX))^{M_1}\times(\calP^n(\calX))^{M_2}}}\Pr\Big\{(\bX^N,\bY^n)\in(\calF_t(\phi_{n,N})\cap\calT_{\Omega^{M_1},\Psi^{M_2}})\Big\}\\
&\geq \sum_{t\in[T_K]:t\neq l}\sum_{\substack{(\Omega^{M_1},\Psi^{M_2})\in\calF_l(\phi_{n,N}^\rmT)}}\Pr\Big\{(\bX^N,\bY^n)\in(\calF_t(\phi_{n,N})\cap\calT_{\Omega^{M_1},\Psi^{M_2}})\Big\}\\
&\ge \sum_{t\in[T_K]:t\neq l}\sum_{\substack{(\Omega^{M_1},\Psi^{M_2})\in\calF_l(\phi_{n,N}^\rmT)}}\kappa_- \Pr\Big\{(\bX^N,\bY^n)\in\calT_{\Omega^{M_1},\Psi^{M_2}}\Big\}\label{typeequal}\\
&=\kappa_-\Pr\Bigg\{(\bT_{\bX^n},\bT_{Y^n})\in\bigcup_{t\in[T_K]:t\neq l}\calF_l(\phi_{n,N}^\rmT)\Bigg\}\\
&\ge \kappa_-\beta(\psi_{n,N}^\rmT|P^M,Q^N)\label{deferror4typet},
\end{align}
where \eqref{typeequal} follows since each tuples of sequences $(\bx^N,\by^n)\in\calT_{\Omega^{M_1}\Psi^{M_2}}$ have the same probability and thus 
\begin{align}
\Pr\Big\{(\bX^N,\bY^n)\in(\calF_t(\phi_{n,N})\cap\calT_{\Omega^{M_1},\Psi^{M_2}})\Big\}
&\geq \frac{|\calF_t(\phi_{n,N})\cap\calT_{\Omega^{M_1},\Psi^{M_2}})|}{|\calT_{\Omega^{M_1},\Psi^{M_2}})|}\Pr\Big\{(\bX^N,\bY^n)\in\calT_{\Omega^{M_1},\Psi^{M_2}}\Big\}\\
&\geq \kappa_-\Pr\Big\{(\bX^N,\bY^n)\in\calT_{\Omega^{M_1},\Psi^{M_2}}\Big\},
\end{align}
and \eqref{deferror4typet} follows from the definition of the mismatch probability for the type-based test $\phi_{n,N}^\rmT$.

Analogously, for each $l\in[T_K]$, the false reject probability satisfies
\begin{align}
\zeta(\phi_{n,N}|P^{M_1},Q^{M_2})
&=\Pr\{(\bX^N,\bY^n)\in\calF_\rmr(\phi_{n,N})\}\\
&=\sum_{\substack{(\Omega^{M_1},\Psi^{M_2})\in(\calP^N(\calX))^{M_1}\times(\calP^n(\calX))^{M_2}}}\Pr\Big\{(\bX^N,\bY^n)\in(\calF_\rmr(\phi_{n,N})\cap\calT_{\Omega^{M_1},\Psi^{M_2}})\Big\}\\
&\geq \sum_{\substack{(\Omega^{M_1},\Psi^{M_2})\in\calF_\rmr(\phi_{n,N}^\rmT)}}\Pr\Big\{(\bX^N,\bY^n)\in(\calF_t(\phi_{n,N})\cap\calT_{\Omega^{M_1},\Psi^{M_2}})\Big\}\\
&\ge \sum_{\substack{(\Omega^{M_1},\Psi^{M_2})\in\calF_\rmr(\phi_{n,N}^\rmT)}}(1-\kappa_+) \Pr\Big\{(\bX^N,\bY^n)\in\calT_{\Omega^{M_1},\Psi^{M_2}}\Big\}\label{typeequal2}\\
&=(1-\kappa_+) \Pr\Big\{(\bT_{\bX^n},\bT_{Y^n})\in\calF_\rmr(\phi_{n,N}^\rmT)\Big\}\\
&\ge \kappa_-\beta(\psi_{n,N}^\rmT|P^M,Q^N),
\end{align}
where \eqref{typeequal2} follows since when $(\Omega^{M_1},\Psi^{M_2})\in\calF_\rmr(\phi_{n,N}^\rmT)$,
\begin{align}
\frac{|\calF_i(\phi_{n,N})\cap\calT_{\Omega^{M_1},\Psi^{M_2}})|}{|\calT_{\Omega^{M_1},\Psi^{M_2}})|}\leq \kappa_i,
\end{align}
for each $i\in[T_K]$ and thus
\begin{align}
\Pr\Big\{(\bX^N,\bY^n)\in(\calF_\rmr(\phi_{n,N})\cap\calT_{\Omega^{M_1},\Psi^{M_2}})\Big\}
&\geq \frac{|\calF_\rmr(\phi_{n,N})\cap\calT_{\Omega^{M_1},\Psi^{M_2}})|}{|\calT_{\Omega^{M_1},\Psi^{M_2}})|}\Pr\Big\{(\bX^N,\bY^n)\in\calT_{\Omega^{M_1},\Psi^{M_2}}\Big\}\\
&\geq \Bigg(1-\sum_{i\in[T_K]}\frac{|\calF_i(\phi_{n,N})\cap\calT_{\Omega^{M_1},\Psi^{M_2}})|}{|\calT_{\Omega^{M_1},\Psi^{M_2}})|}\Bigg)\Pr\Big\{(\bX^N,\bY^n)\in\calT_{\Omega^{M_1},\Psi^{M_2}}\Big\}\\
&\geq (1-\sum_{i\in[T_K]}\kappa_i)\Pr\Big\{(\bX^N,\bY^n)\in\calT_{\Omega^{M_1},\Psi^{M_2}}\Big\}\\
&=(1-\kappa_+)\Pr\Big\{(\bX^N,\bY^n)\in\calT_{\Omega^{M_1},\Psi^{M_2}}\Big\}.
\end{align}
\end{proof}

Recall the definition of $\delta_{n,N}$ in \eqref{def:deltann}. We next prove a lower bound on the false reject probability for a type-based test.
\begin{lemma}
\label{type:converse}
Consider any type-based test $\phi_{n,N}^\rmT$ such that
\begin{align}
\max_{l\in[T_K]}\sup_{\tilP^{M_1},\tilQ^{M_2})\in\calP_l^K}\beta(\phi_{n,N}^\rmT|\tilP^{M_1},\tilQ^{M_2})\leq \exp(-n\lambda).\label{typeerror:constraint}
\end{align}
For each $l\in[T_K]$, under any pair of unknown generating distributions $(P^M,Q^N)\in\calP_l^K$, the false reject probability of the test under hypothesis $\rmH_l^K$ satisfies
\begin{align}
\zeta(\phi_{n,N}|P^{M_1},Q^{M_2})
&\geq \Pr\{h_K(\bX^N,\bY^n)\leq \lambda-\delta_{n,N}\}.
\end{align}
\end{lemma}
\begin{proof}
Given any $\lambda\in\bbR_+$ and any types $(\Omega^{M_1},\Psi^{M_2})\in(\calP^N(\calX))^{M_1}\times(\calP^n(\calX))^{M_2}$, let
\begin{align}
h^*(\Omega^{M_1},\Psi^{M_2})
&:=\min_{t\in[T_K]}\sum_{i\in\calA_t^K}\mathrm{GJS}(\Omega_i,\Psi_{\sigma_l^K(i)},\alpha),\\
\tilh^*(\Omega^{M_1},\Psi^{M_2})
&:=\min_{\substack{t\in[T_K]:\sum_{i\in\calA_t^K}\mathrm{GJS}(\Omega_i,\Psi_{\sigma_l^K(i)},\alpha)\\\geq h^*(\Omega^{M_1},\Psi^{M_2})}}\sum_{i\in\calA_t^K}\mathrm{GJS}(\Omega_i,\Psi_{\sigma_l^K(i)},\alpha).
\end{align}

To prove the lemma, we need to show that for any type-based test $\phi_{n,N}^\rmT$ satisfying \eqref{typeerror:constraint}, if a tuple of types $(\Omega^{M_1},\Psi^{M_2})\in(\calP^N(\calX))^{M_1}\times(\calP^n(\calX))^{M_2}$ satisfy that
\begin{align}
\tilh^*(\Omega^{M_1},\Psi^{M_2})+\delta_{n,N}<\lambda,
\end{align}
then $\phi_{n,N}^\rmT(\Omega^{M_1},\Psi^{M_2})=\rmH_\rmr$. We prove the claim by contradiction. Assume that there exists types $(\tOmega^{M_1},\tPsi^{M_2})\in(\calP^N(\calX))^{M_1}\times(\calP^n(\calX))^{M_2}$ such that
\begin{align}
\tilh^*(\tOmega^{M_1},\tPsi^{M_2})+\delta_{n,N}&<\lambda,\label{contradict1}\\
\phi_{n,N}^\rmT(\Omega^{M_1},\Psi^{M_2})&=\rmH_k,
\end{align}
for some $k\in[T_K]$. Note that \eqref{contradict1} implies that there exists $(t,s)\in[T_K]^2$ such that $t\neq s$ and
\begin{align}
\sum_{i\in\calA_t^K}\mathrm{GJS}(\tOmega_i,\tPsi_{\sigma_l^K(i)},\alpha)+\delta_{n,N}<\lambda,\label{con:step1.5}\\
\sum_{i\in\calA_s}\mathrm{GJS}(\tOmega_i,\tPsi_{\sigma_l^K(i)},\alpha)+\delta_{n,N}<\lambda.
\end{align}
Furthermore, either $t\neq k$ or $s\neq k$. Without loss of generality, assume $t\neq k$. 

For any tuple of generating distributions $(\tilP^{M_1},\tilQ^{M_2})\in\calP_t$, under hypothesis $\rmH_t$, the mismatch probability satisfies 
\begin{align}
\beta(\Psi_{n,N}^\rmT|\tilP^{M_1},\tilQ^{M_2})
&\geq \bigg(\prod_{i\in[M_1]}\tilP_i^N(\calT_{\tOmega_i}^N)\bigg)\bigg(\prod_{j\in[M_2]}\tilQ_j^N(\calT_{\tPsi_j}^n)\bigg)\\
&\ge \exp(-n\delta_{n,N})\exp\Big(-nE_t(\tilP^{M_1},\tilQ^{M_2},\tOmega^{M_1},\tPsi^{M_2},\alpha)\Big)\label{usemethodoftypes2},
\end{align}
where \eqref{usemethodoftypes2} follows from \cite[Lemma 2.6]{csiszar2011information} and the definition of $E_t(\cdot)$ in \eqref{def:El}. Choose $(\tilP^{M_1},\tilQ^{M_2})$ such that
\begin{align}
\tilP_i
=\left\{
\begin{array}{ll}
\frac{\alpha\tOmega_i+\tPsi_{\sigma_l^K(i)}}{1+\alpha}&\mathrm{if}~i\in\calA_t^K,\\
\tOmega_i&\mathrm{if}~i\notin\calA_t^K,\
\end{array}
\right.
\end{align}
and 
\begin{align}
\tilQ_j
=\left\{
\begin{array}{ll}
\frac{\alpha\tOmega_{(\sigma_l^K)^{-1}(j)}+\tPsi_j}{1+\alpha}&\mathrm{if}~j\in\calB_t^K,\\
\tOmega_j&\mathrm{if}~j\notin\calB_t^K.
\end{array}
\right.
\end{align}
It follows that
\begin{align}
E_t(\tilP^{M_1},\tilQ^{M_2},\tOmega^{M_1},\tPsi^{M_2},\alpha)=\sum_{i\in\calA_t^K}\mathrm{GJS}(\tOmega_i,\tPsi_{\sigma_l^K(i)},\alpha)\label{con:step2}.
\end{align}
Thus, combining \eqref{con:step1.5}, \eqref{usemethodoftypes2} and \eqref{con:step2} leads to 
\begin{align}
\beta(\Psi_{n,N}^\rmT|\tilP^{M_1},\tilQ^{M_2})>\exp(-n\lambda),
\end{align}
which contradicts \eqref{typeerror:constraint}. The proof is thus completed.
\end{proof}

The proof of Lemma \ref{converse:gnp} is completed by combining Lemma \ref{type:optimal} with $\kappa_i=\frac{1}{n}$ and Lemma \ref{type:converse}.

\bibliographystyle{IEEEtran}
\bibliography{IEEEfull_lin}

\begin{thebibliography}{10}
\providecommand{\url}[1]{#1}
\csname url@samestyle\endcsname
\providecommand{\newblock}{\relax}
\providecommand{\bibinfo}[2]{#2}
\providecommand{\BIBentrySTDinterwordspacing}{\spaceskip=0pt\relax}
\providecommand{\BIBentryALTinterwordstretchfactor}{4}
\providecommand{\BIBentryALTinterwordspacing}{\spaceskip=\fontdimen2\font plus
\BIBentryALTinterwordstretchfactor\fontdimen3\font minus
  \fontdimen4\font\relax}
\providecommand{\BIBforeignlanguage}[2]{{%
\expandafter\ifx\csname l@#1\endcsname\relax
\typeout{** WARNING: IEEEtran.bst: No hyphenation pattern has been}%
\typeout{** loaded for the language `#1'. Using the pattern for}%
\typeout{** the default language instead.}%
\else
\language=\csname l@#1\endcsname
\fi
#2}}
\providecommand{\BIBdecl}{\relax}
\BIBdecl

\bibitem{blahut1974hypothesis}
R.~Blahut, ``Hypothesis testing and information theory,'' \emph{IEEE Trans.
  Inf. Theory}, vol.~20, no.~4, pp. 405--417, 1974.

\bibitem{lehmann2006testing}
E.~L. Lehmann and J.~P. Romano, \emph{{Testing Statistical Hypotheses}}.\hskip
  1em plus 0.5em minus 0.4em\relax Springer Science \& Business Media, 2006.

\bibitem{batu2013testing}
T.~Batu, L.~Fortnow, R.~Rubinfeld, W.~D. Smith, and P.~White, ``Testing
  closeness of discrete distributions,'' \emph{Journal of the ACM}, vol.~60,
  no.~1, pp. 4:1--4:25, 2013.

\bibitem{gutman1989asymptotically}
M.~Gutman, ``Asymptotically optimal classification for multiple tests with
  empirically observed statistics,'' \emph{IEEE Trans. Inf. Theory}, vol.~35,
  no.~2, pp. 401--408, 1989.

\bibitem{hoeffding1965}
W.~Hoeffding, ``Asymptotically optimal tests for multinomial distributions,''
  \emph{The Annals of Mathematical Statistics}, pp. 369--341, April 1965.

\bibitem{zhou2018binary}
L.~Zhou, V.~Y.~F. Tan, and M.~Motani, ``Second-order optimal statistical
  classification,'' \emph{Information and Inference: A Journal of the IMA},
  vol.~9, no.~1, pp. 81--111, 2020.

\bibitem{unnikrishnan2015asymptotically}
J.~Unnikrishnan, ``Asymptotically optimal matching of multiple sequences to
  source distributions and training sequences,'' \emph{IEEE Trans. Inf.
  Theory}, vol.~61, no.~1, pp. 452--468, 2015.

\bibitem{dembo2009large}
A.~Dembo and O.~Zeitouni, \emph{Large Deviations Techniques and
  Applications}.\hskip 1em plus 0.5em minus 0.4em\relax Springer, 2009,
  vol.~38.

\bibitem{polyanskiy2010finite}
Y.~Polyanskiy, H.~V. Poor, and S.~Verd\'u, ``Channel coding rate in the finite
  blocklength regime,'' \emph{IEEE Trans. Inf. Theory}, vol.~56, no.~5, pp.
  2307--2359, 2010.

\bibitem{hayashi2009information}
M.~Hayashi, ``Information spectrum approach to second-order coding rate in
  channel coding,'' \emph{IEEE Trans. Inf. Theory}, vol.~55, no.~11, pp.
  4947--4966, 2009.

\bibitem{altug14a}
Y.~Altu\u{g} and A.~B. Wagner, ``Refinement of the sphere-packing bound,''
  \emph{IEEE Trans. Inf. Theory}, vol.~60, no.~3, pp. 1592--1615, 2014.

\bibitem{altug14b}
------, ``Refinement of the random coding bound,'' \emph{IEEE Trans. Inf.
  Theory}, vol.~60, no.~10, pp. 6005--6023, 2014.

\bibitem{erseghe2016tit}
T.~Erseghe, ``Coding in the finite-blocklength regime: Bounds based on laplace
  integrals and their asymptotic approximations,'' \emph{IEEE Trans. Inf.
  Theory}, vol.~62, no.~12, pp. 6854--6883, 2016.

\bibitem{merhav1991bayesian}
N.~Merhav and J.~Ziv, ``A {Bayesian approach for classification of Markov}
  sources,'' \emph{IEEE Trans. Inf. Theory}, vol.~37, no.~4, pp. 1067--1071,
  1991.

\bibitem{shota2020beyasian}
S.~Saito and T.~Matsushima, ``Evaluation of error probability of classification
  based on the analysis of the {Bayes} code,'' in \emph{IEEE ISIT}, 2020, pp.
  2510--2514.

\bibitem{shota2021beyasian}
------, ``Evaluation of error probability of classification based on the
  analysis of the {B}ayes code: Extension and example,'' in \emph{IEEE ISIT},
  2021, pp. 1445--1450.

\bibitem{unnikrishnan2016weak}
J.~Unnikrishnan and D.~Huang, ``Weak convergence analysis of asymptotically
  optimal hypothesis tests,'' \emph{IEEE Trans. Inf. Theory}, vol.~62, no.~7,
  pp. 4285--4299, 2016.

\bibitem{hsu2020binary}
H.-W. Hsu and I.-H. Wang, ``On binary statistical classification from
  mismatched empirically observed statistics,'' in \emph{IEEE ISIT}, 2020, pp.
  2533--2538.

\bibitem{mahdi2021sequential}
M.~Haghifam, V.~Y.~F. Tan, and A.~Khisti, ``Sequential classification with
  empirically observed statistics,'' \emph{IEEE Trans. Inf. Theory}, vol.~67,
  no.~5, pp. 3095--3113, 2021.

\bibitem{Ihwang2022sequential}
C.~Y. Hsu, C.~F. Li, and I.~H. Wang, ``On universal sequential classification
  from sequentially observed empirical statistics,'' in \emph{IEEE ITW}, 2022,
  pp. 642--647.

\bibitem{kelly2013}
B.~G. Kelly, A.~B. Wagner, T.~Tularak, and P.~Viswanath, ``Classification of
  homogeneous data with large alphabets,'' \emph{IEEE Trans. Inf. Theory},
  vol.~59, no.~2, pp. 782--795, Feb 2013.

\bibitem{he2019distributed}
H.~He, L.~Zhou, and V.~Y. Tan, ``Distributed detection with empirically
  observed statistics,'' \emph{IEEE Trans. Inf. Theory}, vol.~65, no.~7, pp.
  4349--4367, 2020.

\bibitem{li2014}
Y.~Li, S.~Nitinawarat, and V.~V. Veeravalli, ``Universal outlier hypothesis
  testing,'' \emph{IEEE Trans. Inf. Theory}, vol.~60, no.~7, pp. 4066--4082,
  2014.

\bibitem{zhou2022second}
L.~Zhou, Y.~Wei, and A.~O. Hero, ``Second-order asymptotically optimal outlier
  hypothesis testing,'' \emph{IEEE Trans. Inf. Theory}, vol.~68, no.~6, pp.
  3585--3607, 2022.

\bibitem{zhou2022twophase}
L.~Zhou, J.~Diao, and L.~Bai, ``Achievable error exponents for two-phase
  multiple classification,'' \emph{arXiv 2210.12736}, 2022.

\bibitem{ahlswede1987}
R.~Ahlswede and I.~Wegener, \emph{Search Problems}.\hskip 1em plus 0.5em minus
  0.4em\relax Chichester, U.K.: Wiley, 1987.

\bibitem{ahlswede2006}
R.~Ahlswede and E.~Haroutunian, ``On logarithmically asymptotically optimal
  testing of hypotheses and identification,,'' \emph{General Theory of
  Information Transfer and Combinatorics (Lecture Notes in Computer Science)},
  vol. 4123, pp. 553--571, 2006.

\bibitem{ZhouBook}
L.~Zhou and M.~Motani, ``Finite blocklength lossy source coding for discrete
  memoryless sources,'' \emph{Foundations and Trends$\,$\textregistered $ $ in
  Communications and Information Theory}, vol.~20, no.~3, pp. 157--389, 2023.

\bibitem{Tanbook}
V.~Y.~F. Tan, ``Asymptotic estimates in information theory with non-vanishing
  error probabilities,'' \emph{{Foundations and Trends$\,$\textregistered $ $
  in Communications and Information Theory}}, vol.~11, no. 1--2, pp. 1--184,
  2014.

\bibitem{csiszar1998mt}
I.~Csiszar, ``The method of types [information theory],'' \emph{IEEE Trans.
  Inf. Theory}, vol.~44, no.~6, pp. 2505--2523, 1998.

\bibitem{lin1991divergence}
J.~Lin, ``Divergence measures based on the shannon entropy,'' \emph{IEEE Trans.
  Inf. Theory}, vol.~37, no.~1, pp. 145--151, 1991.

\bibitem{watanabe2015}
S.~Watanabe, S.~Kuzuoka, and V.~Y.~F. Tan, ``Nonasymptotic and second-order
  achievability bounds for coding with side-information,'' \emph{IEEE Trans.
  Inf. Theory}, vol.~61, no.~4, pp. 1574--1605, 2015.

\bibitem{Ben03}
V.~Bentkus, ``On the dependence of the {Berry-Esseen} bound on dimension,''
  \emph{J.\ Stat.\ Planning and Inference}, vol. 113, pp. 385--402, 2003.

\bibitem{molavianjazi2015second}
E.~MolavianJazi and J.~N. Laneman, ``A second-order achievable rate region for
  {Gaussian} multi-access channels via a central limit theorem for functions,''
  \emph{IEEE Trans. Inf. Theory}, vol.~61, no.~12, pp. 6719--6733, 2015.

\bibitem{iri2015third}
N.~Iri and O.~Kosut, ``Third-order coding rate for universal compression of
  {Markov} sources,'' in \emph{IEEE ISIT}, 2015, pp. 1996--2000.

\bibitem{renyi1961measures}
A.~R{\'e}nyi, ``On measures of entropy and information,'' in \emph{in Proc. 4th
  Berkeley Symp. Probability Theory and Mathematical Statist.}, Berkeley, CA,
  1961, pp. 547--561.

\bibitem{csiszar2011information}
I.~Csisz\'ar and J.~K{\"o}rner, \emph{Information Theory: Coding Theorems for
  Discrete Memoryless Systems}.\hskip 1em plus 0.5em minus 0.4em\relax
  Cambridge University Press, 2011.

\bibitem{tan2014state}
M.~Tomamichel and V.~Y.~F. Tan, ``Second-order coding rates for channels with
  state,'' \emph{IEEE Trans. Inf. Theory}, vol.~60, no.~8, pp. 4427--4448,
  2014.

\bibitem{tan2014dispersions}
V.~Y.~F. Tan and O.~Kosut, ``On the dispersions of three network information
  theory problems,'' \emph{IEEE Trans. Inf. Theory}, vol.~60, no.~2, pp.
  881--903, 2014.

\bibitem{gretton2012jmlr}
A.~Gretton, K.~Borgwardt, M.~Rasch, B.~Scholkopf, and A.~Smola, ``A kernel
  two-sample test,'' \emph{J. Mach. Learn. Res.}, vol.~13, pp. 723--773, 2012.

\bibitem{Scarlett2016}
J.~Scarlett and V.~Cevher, ``Phase transitions in group testing,'' in
  \emph{27th Annual ACM-SIAM Symposium on Discrete Algorithms}, 2016, pp.
  40--53.

\bibitem{Scarlett2017}
------, ``How little does non-exact recovery help in group testing?'' in
  \emph{IEEE ICASSP}, 2017, pp. 6090--6094.

\bibitem{Unnikrishnan2013}
J.~Unnikrishnan and F.~M. Naini, ``De-anonymizing private data by matching
  statistics,'' in \emph{Proc. 51st Annu. Allerton Conf.}, 2013, pp.
  1616--1623.

\bibitem{Naini2016}
F.~M. Naini, J.~Unnikrishnan, P.~Thiran, and M.~Vetterli, ``Where you are is
  who you are: User identification by matching statistics,'' \emph{IEEE Trans.
  Inf. Forensics Security}, vol.~11, no.~2, pp. 358--372, 2016.

\bibitem{boyd2004convex}
S.~Boyd and L.~Vandenberghe, \emph{Convex optimization}.\hskip 1em plus 0.5em
  minus 0.4em\relax Cambridge university press, 2004.

\bibitem{tuncel2005error}
E.~Tuncel, ``On error exponents in hypothesis testing,'' \emph{IEEE Trans. Inf.
  Theory}, vol.~51, no.~8, pp. 2945--2950, 2005.

\end{thebibliography}
\end{document}